\documentclass[a4paper,twocolumn,11pt,accepted=2025-11-05]{quantumarticle}
\pdfoutput=1

\usepackage[T1]{fontenc}
\usepackage{amsmath,amssymb,amsfonts,amsthm,latexsym}
\usepackage{bbm,dsfont}
\usepackage{graphicx}
\usepackage{mathtools}

\usepackage{hyperref}             % no options here
\hypersetup{colorlinks=true}      % set options here
\usepackage{cleveref}             % must be after hyperref

\usepackage{subcaption}

\usepackage{tikz}
\usetikzlibrary{shapes.geometric,arrows,calc,through,positioning,matrix,graphs,graphs.standard,arrows.meta}
\usepackage{lipsum}

\usepackage[normalem]{ulem}
\usepackage[numbers,sort&compress]{natbib}

\newtheorem{proposition}{Proposition}
\newtheorem{proposition?}{Proposition?}
\newtheorem{theorem}{Theorem}
\newtheorem{lemma}{Lemma}
\newtheorem{corollary}{Corollary}

\newtheorem{conjecture}{Conjecture}

\theoremstyle{definition}

\newtheorem{definition}{Definition}

\newcommand{\ket}[1]{|#1\rangle} %ket
\newcommand{\bra}[1]{\langle#1|} %bra
\newcommand{\kb}[2]{|#1\rangle\langle#2|} %ketbra
\newcommand{\no}[1]{\left\|#1\right\|} %norm
\newcommand{\tr}[1]{\textrm{tr}\left[#1\right]} %trace
\newcommand{\id}{\mathbbm{1}} %identity operator
\newcommand{\meo}{\Omega} %measurement outcomes
\newcommand{\var}{\textrm{Var}} %variance
\newcommand{\ve}{\mathbf{e}} %e
\newcommand{\F}{\mathsf{F}}%generic observable
\newcommand{\G}{\mathsf{G}}%generic joint observable
\renewcommand{\P}{\mathsf{P}}%sharp observable
\newcommand{\M}{\mathsf{M}}%observable
\begin{document}

\title{Optimal Fermionic Joint Measurements for Estimating Non-Commuting Majorana Observables}
\author{Daniel McNulty}
\affiliation{Dipartimento di Fisica, Università di Bari, 70126 Bari, Italy}
\affiliation{Center for Theoretical Physics, Polish Academy of Sciences, Al. Lotników 32/46, 02-668 Warszawa, Poland}
\orcid{0000-0002-3276-0587}
\author{Susane Calegari}
\affiliation{Center for Theoretical Physics, Polish Academy of Sciences, Al. Lotników 32/46, 02-668 Warszawa, Poland}
\author{Michał Oszmaniec}
\affiliation{Center for Theoretical Physics, Polish Academy of Sciences, Al. Lotników 32/46, 02-668 Warszawa, Poland}
\affiliation{NASK National Research Institute, ul. Kolska 12, 01-045 Warszawa, Poland}
\orcid{0000-0002-4946-6835}

\begin{abstract}
An important class of fermionic observables, relevant in tasks such as fermionic partial tomography and estimating energy levels of chemical Hamiltonians, are the binary measurements obtained from the product of anti-commuting Majorana operators. In this work, we investigate efficient estimation strategies of these observables based on a joint measurement which, after classical post-processing, yields all sufficiently unsharp (noisy) Majorana observables of even-degree. By exploiting the symmetry properties of the Majorana observables, as described by the braid group, we show that the incompatibility robustness, i.e., the minimal classical noise necessary for joint measurability, relates to the spectral properties of the Sachdev-Ye-Kitaev (SYK) model. In particular, we show that for an $n$ mode fermionic system, the incompatibility robustness of all degree--$2k$ Majorana observables satisfies $\Theta(n^{-k/2})$ for $k\leq 5$. Furthermore, we present a joint measurement scheme achieving the asymptotically optimal noise, implemented by a small number of fermionic Gaussian unitaries and sampling from the set of all Majorana monomials. Our joint measurement, which can be performed via a randomization over projective measurements, provides rigorous performance guarantees for estimating fermionic observables comparable with fermionic classical shadows.
\end{abstract}

\maketitle
%\tableofcontents
 
 \section{Introduction} 

A fundamental non-classical feature of quantum mechanics, and a key ingredient in practical applications, is the inability to jointly measure non-commuting (Hermitian) observables \cite{busch,heinosaari08}. 
In the general setting of positive operator-valued measures (POVMs), incompatibility (or lack of joint measurability) and non-commutativity are related but inequivalent notions. If one adds a sufficient amount of (e.g., depolarising) noise to the set of observables, at some threshold, known as the incompatibility robustness \cite{heinosaari15}, the non-commuting measurements become compatible and can be implemented by a parent POVM and classical post-processing. 
This leads naturally to a resource theoretic formulation of incompatibility \cite{heinosaari15,chitambar19}, with its characterization and quantification connected to both foundational topics \cite{wolf09,kiukas19,xu19}, and practical tasks in quantum information theory \cite{quintino14,uola15,skrzypczyk19}.

Incompatibility is also relevant in the performance of estimation tasks on present-day quantum computers \cite{kandala17}. 
For example, the inability to simultaneously measure non-commuting observables presents a significant hurdle for quantum variational algorithms \cite{bharti22,preskill18,mcclean16,garcia21}, which require repeated estimations of expectation values of many-body Hamiltonians. 
To reduce the cost of performing individual measurements, and to improve scalability, several strategies have been developed. 
One method involves finding minimal groupings of the measurements into pairwise commuting subclasses that can be jointly measured \cite{jena19,gokhale19,yen20,verteletskyi20,crawford21,izmaylov20,zhao20,bonet20}. 
Another approach, known as classical shadows, implements a randomized measurement strategy to construct a classical representation of the unknown state \cite{huang20,hadfield20,gresch23,koh22,chen21}. The classical shadow can be used to simultaneously estimate multiple properties of the system such as expectations of arbitrary observables up to some required precision. Originally formulated for systems of qubits, classical shadows have since been extended to the fermion setting via several methods \cite{Zhao21,wan22,low22,ogorman22}.
%Initially developed for Pauli and Clifford measurements on multi-qubit systems, classical shadows have since been adapted to the fermionic setting \cite{zhao21,wan22,low22,ogorman22}.
Alternatively, one can navigate the incompatibility issue directly by applying a parent measurement that simultaneously measures unsharp versions of the non-commuting observables \cite{mcnulty22}. 
This approach has been used to simultaneously estimate expectation values of $n$--qubit Pauli strings and multi-qubit Hamiltonians, and provides efficiency improvements over current classical shadow techniques in noisy scenarios \cite{koh22,chen21}. 

In this work, we extend the joint measurability framework to the fermionic setting and characterize the measurement incompatibility of assemblages of Majorana observables. In particular, we consider strategies which simultaneously measure unsharp (noisy) versions of the set of $\binom{2n}{2k}$ degree--$2k$ Majorana observables of an $n$ mode fermionic system. The joint measurement, together with classical post-processing of the outcomes, yields single shot estimators for tasks such as partial state tomography, which upon repetition can approximate all $k$--body fermionic marginals to a required precision \cite{bonet20,Zhao21}. Equipped with this partial knowledge of the state, one can determine properties such as multipole moments \cite{gidofalvi07}, derivatives of energy \cite{obrien19,overy14}, as well as apply error mitigation techniques \cite{mcclean17,takeshita20} and fidelity certification \cite{gluza18}. The strategy also provides single-shot estimators for the energy levels of electronic structure Hamiltonians \cite{mcclean16}.

The optimality of the joint measurability approach is determined by the incompatibility robustness of the set of observables. While calculating this quantity usually requires semidefinite programming (SDP) techniques that quickly become computationally infeasible for large collections of observables \cite{designolle19a,cavalcanti17,guhne23}, symmetries between the observables can simplify the problem considerably \cite{nguyen20}. We exploit this simplification by showing that the set of degree--$k$ Majorana measurements are uniformly and rigidly symmetric with respect to the braid group \cite{bravyi06}, represented in fermionic Fock space. This reduces the problem of finding the optimal robustness to finding the maximal eigenvalue of a class of operators relevant in the Sachdev-Ye-Kitaev (SYK) model, which is an important model of strongly interacting quantum systems inspired by high energy physics \cite{sachdev93,kitaev15,feng19}. In the simplest case of $\binom{2n}{2}$ quadratic Majoranas, we derive an analytical upper bound for the robustness which is tight if and only if a real $2n\times2n$ skew-Hadamard matrix exists. More generally, for assemblages of degree--$2k$ Majorana observables with $k\leq 5$, we apply a bound derived by Hastings and O'Donnell on the spectrum of the SYK model \cite{hastings21}, to show that the incompatibility robustness satisfies $\mathcal{O}(n^{-k/2})$.

We also provide a joint measurement scheme for each $k=\mathcal{O}(1)$, which simultaneously measures all degree--$2k$ Majorana observables with sharpness $\Theta (n^{-k/2})$, and hence is asymptotically optimal for $k\leq 5$. Importantly, the joint measurement can be implemented on a quantum computer by the randomization of projective measurements \cite{oszmaniec17}. In a circuit implementation, the joint measurement corresponds to transforming the state by conjugating with a randomly sampled Majorana monomial, followed by a fermionic Gaussian unitary, and finally a fermionic computational basis measurement (see Fig. \ref{fig:general_idea}). We show that applying this measurement strategy to simultaneously estimate non-commuting fermionic observables yields the same sample complexity scalings as fermionic shadows but requires implementing significantly fewer fermionic Gaussian unitaries \cite{Zhao21,wan22}.

This work is accompanied by a complementary paper \cite{majsak24} that focuses on the practical problem of implementing the joint measurement strategy for Majorana pairs and quadruples, as well as the problem of estimating expectation values of molecular quantum chemistry Hamiltonians. In the current manuscript, we focus mainly on the structural and mathematical aspects of the joint measurability properties of Majoranas, establishing connections to graph and design theory, as well as to the SYK model. Some of the structures that arise include perfect matchings of Turán graphs, tournament matrices \cite{Ito17}, Hadamard and skew-Hadamard matrices \cite{koukouvinos08}, as well as flat orthogonal matrices \cite{jaming15}.

The paper is structured as follows. We begin in Sec. \ref{sec:majoranas} with some preliminaries and introduce the relevant notation. In Sec. \ref{sec:jm} we describe our general fermionic joint measurement strategy. Sec. \ref{sec:nearly_optimal} provides lower bounds on the sharpness of the degree--$2k$ observables that constitute the marginals of the parent POVM. In Sec. \ref{sec:ir_syk} we evaluate the incompatibility robustness of degree--$2k$ Majorana observables via a connection to the SYK model. In Sec.  \ref{sec:est} we employ the joint measurement strategy to estimate fermionic observables for partial state tomography and energy estimations of Hamiltonians, and in Sec. \ref{sec:comparison} we provide comparisons to other approaches. In particular, we highlight some connections to fermionic classical shadows, as well as a measurement strategy based on fermion-to-qubit transformations. Finally, Sec. \ref{sec:conclusion} contains some concluding remarks and open questions.

\section{Preliminaries}\label{sec:majoranas}

Let $\mathcal H$ be a finite-dimensional Hilbert space of $\dim \mathcal H=d<\infty$. A measurement is described by a POVM $\M$, with a finite outcome set $\Omega$, and consists of positive semidefinite matrices (effects) $\M(e)\geq 0$, for which $\sum_{e\in \Omega} \M(e)=\id$, where $\id$ is the identity of $\mathcal H$. For a quantum state $\rho$, the outcome probability distribution of $\M$ is given by $p(e|\rho)=\tr{\M(e)\rho}$.

A finite collection of measurements $\M_1,\ldots, \M_m$ is \emph{jointly measurable} (compatible) whenever their statistics can be reproduced by classical post-processing of the statistics from a single (parent) POVM. In particular, there exists a measurement $\G$ with outcome set $\Omega_{\G}$ such that for each measurement $\M_j$, its effects $\M_j(e_j)$ with outcomes $e_j\in\Omega_j$, can be obtained from $\G$ via the classical post-processing (stochastic transformation),
\begin{equation}\label{eq:postprocessing}
\M_{j}(e_j)=\sum_{\lambda\in\Omega_{\G}} D(e_j|j,\lambda)\G(\lambda)\,,
\end{equation}
where $0\leq D(e_j|j,\lambda)\leq 1$ and $\sum_{e_j\in\Omega_j}D(e_j|j,\lambda)=1$, for all $j=1,\ldots,m$. Equivalently, a set of POVMs is said to be jointly measurable if there exists a parent whose marginals yield all effects of the individual POVMs. A set of POVMs with no joint measurement is called incompatible \cite{busch}.

For an $n$ mode fermionic system, we label the $2n$ Majorana operators,
\begin{equation}\label{eq:majoranas}
\gamma_{2j-1}=a_j^{\dagger}+a_j\,,\qquad\gamma_{2j}=i(a_j^{\dagger}-a_j)\,,
\end{equation}
where $a_j$ and $a_j^{\dagger}$ are the annihilation and creation operators of mode $j\in[n]:=\{1,2,\ldots,n\}$. The operators anti-commute and are Hermitian, i.e., $\gamma_j\gamma_{j'}+\gamma_{j'}\gamma_j=2\delta_{j{j'}}\id$ and $\gamma^{\dagger}_j=\gamma_j$, for all $j,j' \in[2n]$.

In this work, we investigate the joint measurability properties of degree--$2k$ Majorana observables,
\begin{equation}
\gamma_S:=i^{\binom{2k}{2}}\gamma_{j_1}\gamma_{j_2}\cdots\gamma_{j_{2k}}\,,
\end{equation} 
where $S=\{j_1,\ldots,j_{2k}\}\subseteq [2n]$ is a $2k$-element subset, with its elements listed in ascending order $1\leq j_1<\cdots<j_{2k}\leq 2n$. The phase associated with the monomial ensures Hermiticity. If $S=\varnothing$, i.e., $|S|=0$, then $\gamma_{\varnothing}=\id$. Any pair of observables $\gamma_{S}$ and $\gamma_{S'}$, with $S,S'\subseteq \{1,2\ldots,2n\}$, either commute or anti-commute in accordance to,
\begin{equation}\label{eq:commutativity}
\gamma_S\gamma_{S'}=(-1)^{|S|\cdot |S'|-|S\cap S'|}\gamma_{S'}\gamma_S\,.
\end{equation}
We label the set of degree--$2k$ Majorana observables by the index set
%\begin{equation}
%\mathcal{S}_{2k}=\{(j_1,\ldots,j_{2k})\,|\,1\leq j_1<\ldots<j_{2k}\leq2n\}\,.
%\end{equation}
\begin{equation}
\mathcal{S}_{2k}=\binom{[2n]}{2k}=\{S\subset [2n]\,:\,|S|=2k\}\,,
\end{equation}
where each $2k$--element string $S\in\mathcal{S}_{2k}$ is ordered ascendingly.

Our aim is to construct a parent POVM which, after classical post-processing, yields the unsharp (noisy) degree--$2k$ Majorana measurements,
\begin{equation}\label{eq:unsharp}
\M_S(e_S)=\tfrac{1}{2}(\id+ e_S\eta_S\gamma_S)\,,
\end{equation}
for every $S\in\mathcal{S}_{2k}$, where $e_S\in\{\pm 1\}$ labels the outcome, and $\eta_S$ the sharpness, which satisfies $0\leq\eta_S\leq1$. Here, $\M_S$ is obtained by deforming the projective measurement $\P_S(\pm)=\frac{1}{2}(\id\pm \gamma_S)$ with the depolarising channel, i.e., $\M_S=\eta_S\P_S+(1-\eta_S)\id/2$. Note that $\P_S$ takes this form due to the fact that $\gamma_S^2 = \id$ by the canonical anticommutation relations. For a given parent POVM $\G$, we consider $\eta_S$ to be the maximal sharpness for which $\M_S$ is jointly measured by $\G$. Furthermore, we denote
\begin{equation}\label{eq:min_sharpness}
\eta_{2k}:=\min\{\eta_S\,|\,S\in\mathcal{S}_{2k}\}\,,
\end{equation}
as the (maximal) sharpness for which \emph{all} degree--$2k$ Majorana observables can be measured simultaneously by $\G$.

From the outcome of a joint measurement, it is simple to produce a single shot estimator of $\tr{\gamma_S\rho}$, for every degree--$2k$ Majorana observable. The general strategy, originally described for systems of qubits \cite{mcnulty22}, is to estimate $\tr{\gamma_S\rho}$ by sampling from outputs $e_S$ of the noisy Majorana measurement $\M_S$ of Eq. (\ref{eq:unsharp}). Importantly, for any input state $\rho$, the expectation value of $e_S$ equals $\eta_S\tr{\gamma_S\rho}$, so it is natural to set the unbiased estimator as $\hat{\gamma}_S=\eta_S^{-1}e_S$. We apply this strategy to simultaneously estimate non-commuting fermionic observables, and provide rigorous bounds on the number of copies of the state $\rho$ (i.e., the sample complexity) to approximate the expectations up to a required precision.

We will describe the parent measurement in both the Schr\"{o}dinger and Heisenberg pictures: In the Schr\"{o}dinger picture, the measurement is encoded in a unitary evolution acting on the quantum state (Fig. \ref{fig:general_idea}); in the Heisenberg picture, the transformation acts on the computational basis measurement while the state remains fixed. Both viewpoints are equivalent since implementing a POVM $\M$ on the transformed state $U\rho U^{\dagger}$ is equivalent to implementing $U^{\dagger}\M U$ on the state $\rho$ by the duality $\tr{\M U \rho U^{\dagger}}=\tr{U^{\dagger}\M U\rho}$.

To construct and implement the parent measurement we rely on fermionic Gaussian (or linear-optical) unitaries \cite{knill01,bravyi02,jozsa08}. These are a class of unitary operators which satisfy
\begin{equation}\label{eq:unitary}
\gamma_{j}^O:=U_O^{\dagger} \gamma_{j}U_O=\sum_{j'=1}^{2n}O_{jj'}\gamma_{j'}\,,
\end{equation}
for each $j\in[2n]$, where $O\in O(2n)$ is a real orthogonal $2n\times 2n$ matrix \cite[Theorem 3]{jozsa08}. In particular, they transform the $2n$ Majorana operators into a rotated set of operators which satisfy the same anti-commutation relations. 

A fermionic Gaussian unitary is generated from a quadratic Hamiltonian \cite{jozsa08},  
\begin{equation}
U_O = e^{iH}, \qquad 
H = -\frac{i}{4}\sum_{j,j'=1}^{2n}\alpha_{jj'}\gamma_j\gamma_{j'},
\end{equation}
where $\alpha \in \mathbb{R}^{2n\times 2n}$ is a real antisymmetric matrix and $O=e^{-\alpha} \in O(2n)$.

Under the action of $U_O$, an observable $\gamma_{S}$, with $S\in\mathcal{S}_{2k}$, transforms to
\begin{equation}\label{eq:unitary_prod}
\gamma_S^O:=U_O^{\dagger} \gamma_{S}U_O=\sum_{S'\in{\mathcal{S}_{2k}}}\text{det}(O_{S,S'})\gamma_{S'}\,,
\end{equation}
where $O_{S,S'}$ denotes a $2k \times 2k$ submatrix of $O$ with rows indexed by $S$ and columns by $S'$ \cite[Corollary 1]{mocherla23}. The class of fermionic Gaussian unitaries play an important role in various quantum computing applications \cite{valiant01,terhal02}, e.g., for performing quantum chemistry and many-body simulations \cite{kivlichan18,google20}, fermion sampling \cite{oszmaniec22,hangleiter23}, and classical shadows \cite{Zhao21,wan22,low22,ogorman22}.

It is often convenient to represent a fermionic system as a system of qubits via a fermion-to-qubit mapping. A transformation of this type maps the $2n$ Majorana operators of Eq. (\ref{eq:majoranas}) to multi-qubit Pauli operators while preserving their anti-commutativity relations. One such mapping is the Jordan-Wigner (JW) transformation \cite{jordan93}, in which the Majorana operators of mode $j\in [n]$ are represented as the $n$-qubit Pauli operators,
\begin{align}\label{eq:jw1}
\gamma_{2j-1}&=Z^{\otimes j-1}\otimes X\otimes\id^{\otimes n-j}\,,\\
\label{eq:jw2}
\gamma_{2j}&=Z^{\otimes j-1}\otimes Y\otimes\id^{\otimes n-j}\,,
\end{align}
where $X,Y$ and $Z$ are the qubit Pauli operators, and $\id$ the $2\times 2$ identity.

A fermionic Gaussian unitary in the JW representation is described by a matchgate circuit \cite{valiant01}. Matchgates are the class of two-qubit gates generated by the rotations $\exp(i\theta X\otimes X)$, $\exp(i\theta Z\otimes \id)$ and $\exp(i\theta \id\otimes Z)$. A matchgate circuit (arranged in a 1D architecture) is generated from the set of all nearest-neighbor matchgates, together with the Pauli operator $X$ on the $n$--th qubit.

To describe the behavior of the joint measurement as the system size increases we use the standard notation $\mathcal{O}$, $\Omega$ and $\Theta$ for asymptotic bounds. In particular, $f(n)=\Theta(g(n))$ implies $\exists\, c,C,n_0>0$ such that $c\cdot g(n)\leq f(n)\leq C \cdot g(n)$ for all $n>n_0$, while $\mathcal{O}$ and $\Omega$ are the upper and lower bounds, respectively.

\section{Fermionic joint measurements}\label{sec:jm}

\begin{figure}[!t]
    \centering
	\includegraphics[width=8cm]{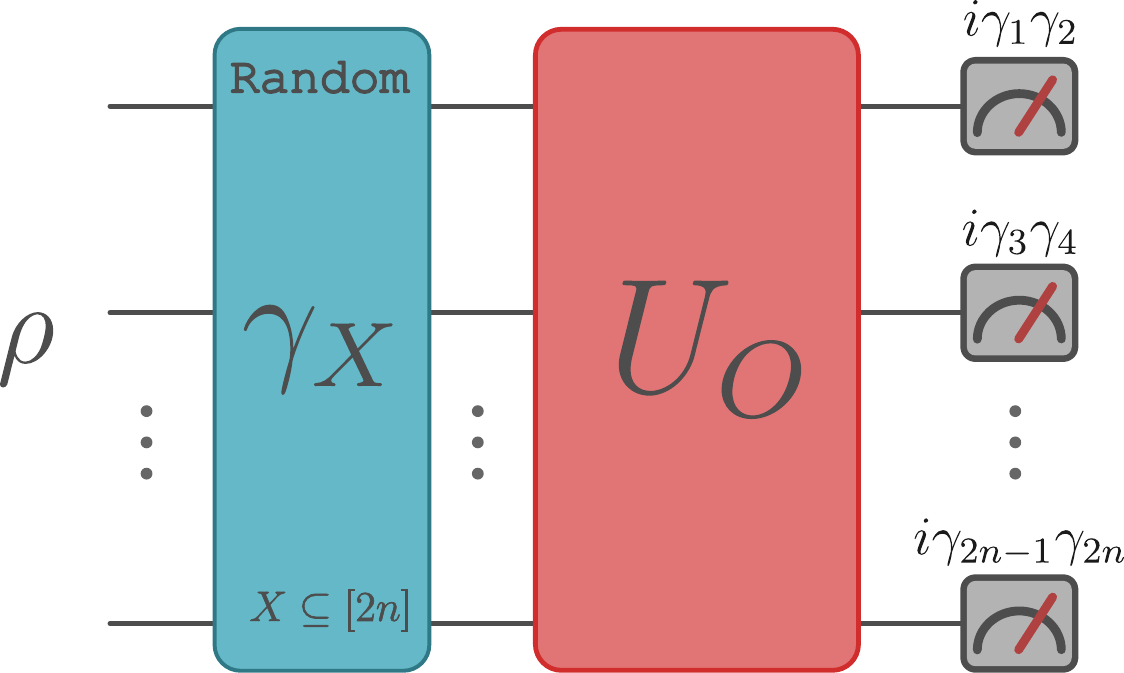}
	\caption{\label{fig:general_idea} Circuit implementation of the joint measurement $\G^O$, defined in Eq. (\ref{eq:jm_gen}), on an $n$-mode fermionic state. The state evolves under the action of a Majorana monomial $\gamma_X$, sampled uniformly at random from the set of all Majorana monomials, followed by a fermionic Gaussian unitary $U_O$. The $n$ measurements $\{\gamma_R\,|\,R\in\mathcal{D}_2\}$ are performed on the $n$ modes, with outcomes $q_{R}\in\{\pm 1\}$. In Sec. \ref{sec:nearly_optimal} we specify a choice of $U_O$ which optimizes the accuracy of the joint measurement.
In the qubit setting---under the Jordan-Wigner transformation---the same circuit is implemented by performing a random Pauli operator $P_i\in\{\id,X,Y,Z\}$ on each of the $n$ qubits, followed by a matchgate circuit and the computational basis measurement.
	}
\end{figure} 

We now introduce a general fermionic parent POVM which, after classical post-processing, yields all unsharp Majorana observables of even degree. First, we consider the joint measurability of the set of $2n$ degree--$1$ Majorana operators, previously studied in connection with the Clifford algebra \cite{kunjwal14,bluhm18}.

\subsection{Single Majoranas}\label{sec:single}

Consider the set of $2n$ Majorana measurements $\M_{j}(\pm)=\tfrac{1}{2}(\id\pm\eta_{j}\gamma_{j})$, with sharpness $\eta_{j}\in[0,1]$ for $j\in[2n]$. Joint measurability has previously been considered in both the uniform \cite{kunjwal14} and biased cases \cite{bluhm18}. Here, the uniform case satisfies $\eta_{j}=\eta_{j'}$ for all $j,j'\in[2n]$, while a biased measurement drops this constraint. By exploiting connections between joint measurability and the inclusion of free spectrahedra \cite{bluhm18}, it was shown that $\{\M_{j}\,|\,j\in[2n]\}$ are jointly measurable if and only if $\sum_{j=1}^{2n} \eta_{j}^2\leq 1$. Note that the criterion generalizes the joint measurability condition for qubit Pauli observables \cite{busch86,brougham07}, and applies for any set of pairwise anti-commuting unitary matrices which are Hermitian and idempotent (of which there are at most $2n+1$ in the Hilbert space of dimension $2^n$).

In the uniform case, an optimal joint measurement was presented in \cite{kunjwal14}. This generalizes to the biased parent POVM,
\begin{equation}\label{eq:jm_single}
\G(e_1,\ldots,e_{2n})=\frac{1}{2^{2n}}\left(\id+\sum_{j=1}^{2n} e_{j}\eta_{j}\gamma_{j}\right)\,,
\end{equation}
with outcome string $e_1,\ldots,e_{2n}\in\{\pm1\}$. By taking the marginals, for each $j\in[2n]$, we recover $\M_{j}(e_{j})$.

We now show that the optimal joint measurement, with $\sum_{j=1}^{2n} \eta_{j}^2=1$, is projectively simulable, i.e., it can be implemented by a randomization over projective measurements \cite{oszmaniec17}. First, a single Majorana operator, say $\gamma_1$, is rotated to a linear combination $\gamma_{1}^O=\sum_{j=1}^{2n} \eta_{j}\gamma_{j}$, by an orthogonal transformation $O\in O(2n)$, as described in Eq. (\ref{eq:unitary}). Then, by conjugating with a randomly sampled monomial $\gamma^{\dagger}_X$, $X\subseteq[2n]$, the operator is mapped to 
\begin{equation}
\gamma_1^O\mapsto \gamma^{\dagger}_X\cdot\gamma_{1}^O\cdot\gamma_X=\sum_{j=1}^{2n}x_{j}\eta_{j}\gamma_{j}\,,
\end{equation}
where $x_{j}=(-1)^{|X|-|X\cap \{j\}|}$ (cf. Eq. (\ref{eq:commutativity})). 
In particular, if $X$ is even,
\begin{equation}\label{eq:randomisation}
x_j=\begin{cases}
-1 &\mbox{if} \,\,\, j\in X \\
1 &\mbox{otherwise} \,,
 \end{cases}
\end{equation}
and if $X$ is odd, the signs in Eq. (\ref{eq:randomisation}) are reversed. It follows that, given any ${\bf x}=(x_1,\ldots, x_{2n})\in\{\pm 1\}^{2n}$, $\exists !\, X	\subseteq [2n]$, such that 
\begin{equation}
\gamma^{\dagger}_X\cdot (\gamma_1,\ldots,\gamma_{2n})\cdot \gamma_X=(x_1\gamma_1,\ldots,x_{2n}\gamma_{2n})\,.
\end{equation}
Consequently, implementing the projective measurement 
\begin{equation}\label{eq:degree-1-pvm}
\P^{O,X}(q)=\tfrac{1}{2}(\id+q \gamma^{\dagger}_X\gamma^O_{1}\gamma_X)\,,
\end{equation}
with $X\in 2^{[2n]}$ sampled uniformly at random from the power set $2^{[2n]}$, i.e., the set of all subsets of $[2n]$, is equivalent to performing the parent $\G$. The outcome $(q,X)$ from (\ref{eq:degree-1-pvm}) yields the string of $2n$ outcomes $e_j=q\cdot x_j$ for all $j\in[2n]$ in Eq. (\ref{eq:jm_single}).

\subsection{Degree--$2k$ Majoranas}\label{subsec:degree_k}

In the fermionic setting, we denote the computation basis measurements by the index set
\begin{equation}\label{eq:diag}
\mathcal{D}_{2}=\{\,\{2j-1,2j\}\,\,|\,j\in[n]\}\subset \mathcal{S}_2\,,
\end{equation}
and the corresponding POVMs,
\begin{equation}\label{eq:diag_k}
\P_R(q_R)=\tfrac{1}{2}(\id+q_R\gamma_R)\,,\quad R\in \mathcal{D}_2\,,
\end{equation}
with outcomes $q_R\in\{\pm 1\}$. Under the Jordan-Wigner transformation, $\{\P_R\,|\, R\in\mathcal{D}_2\}$ map to the standard computational basis measurement on $n$ qubits (cf. Eqs. (\ref{eq:jw1}) and (\ref{eq:jw2})). More generally, the set of degree--$2k$ Majorana observables that commute with the computational basis measurements are denoted by the index set
\begin{equation}
\mathcal{D}_{2k}=\{\cup_{i=1}^k R_i\,|\, R_i\in\mathcal{D}_{2}\,,R_i\neq R_{i'}\}\subset \mathcal{S}_{2k}\,.
\end{equation}

To simultaneously measure all unsharp versions of the degree--$2k$ Majorana monomials $\{\gamma_S\,|\,S\in\mathcal{S}_{2k}\}$, we first take an arbitrary $R\in\mathcal{D}_{2k}$, and apply a fermionic Gaussian unitary $U^{\dagger}_O$ to $\P_R$, as described in Eq. (\ref{eq:unitary_prod}), such that
\begin{equation}\label{eq:rotated_comp_basis}
\P^{O}_R(\pm):=U_O^{\dagger}\P_R(\pm)U_O=\tfrac{1}{2}(\id\pm\gamma_R^{O})\,,
\end{equation}
where
\begin{equation}\label{eq:rotated_prod}
\gamma^O_R:=\sum_{S\in\mathcal{S}_{2k}}\det(O_{R,S})\gamma_{S}\,.
\end{equation}
Then, conjugating with a Majorana monomial $\gamma^{\dagger}_X$, $X\subseteq [2n]$, transforms $\P^{O}_R$ to
\begin{eqnarray*}\label{eq:random_projective}
\P^{O,X}_R(\pm)&:=&\gamma^{\dagger}_X\cdot \P^{O}_R(\pm)\cdot \gamma_X=\tfrac{1}{2}(\id\pm\gamma_R^{O,X})\,,
\end{eqnarray*}
where
\begin{equation}\label{eq:random_rotated_prod}
\gamma_R^{O,X}:=\sum_{S\in\mathcal{S}_{2k}}x_S\det(O_{R,S})\gamma_{S}\,,
\end{equation}
and $x_S=(-1)^{|X|\cdot|S|-|S\cap X|}$. Finally, sampling $X\subseteq [2n]$ uniformly at random from the power set $2^{[2n]}$, i.e., the set of all subsets of $[2n]$, defines a POVM,
\begin{equation}\label{eq:jm}
\begin{aligned}
\G^O_{R}(q_R,X) &:= 2^{-2n}\P^{O,X}_R(q_R)\\
&= \frac{1}{2^{2n+1}}\Bigl(\id
   + q_R \sum_{\mathclap{S\in\mathcal{S}_{2k}}} x_S \det(O_{R,S})\,\gamma_{S}
   \Bigr)
\end{aligned}
\end{equation}
with outcomes $q_R\in\{\pm1\}$ and $X\in 2^{[2n]}$. 

More generally we can define a POVM that is not limited to just one choice of $R\in\mathcal{D}_{2k}$. In particular, by taking the product of $n$ POVMs $\{\G^O_{R}\,|\,R\in \mathcal{D}_2\}$, as defined in Eq. (\ref{eq:jm}) with $k=1$, we can construct a POVM $\G^O$, which simultaneously measures $\G^O_{R}$, for all $R\in \mathcal{D}_{2k}$ and $k=1,2,\ldots,n$. This holds due to the commutativity of the observables $\G^O_{R}$, which are obtained from the commuting set $\{\gamma_R \,|\, R \in \mathcal{D}_2\}$ via conjugation with fixed unitaries $U_O$ and $\gamma_X$. Thus, the  product in Eq.~(\ref{eq:jm_gen}) defines a valid joint POVM.

We summarize the above construction in the following proposition.

\begin{proposition}\label{prop:jm}
Suppose an $n$ mode fermionic state $\rho$ undergoes a unitary evolution described by the circuit in Fig. \ref{fig:general_idea}, followed by a computational basis measurement $\{\P_R\,|\,R\in\mathcal{D}_2\}$ on the $n$ modes, with outcome ${\bf q}\in\{\pm 1\}^{n}$. In the Heisenberg picture, this defines a POVM
\begin{equation}\label{eq:jm_gen}
\G^O({\bf q},X)=(2^{2n})^{n-1}\prod_{R\in\mathcal{D}_2} \G^O_{R}(q_{R},X)\,,
\end{equation}
measured on the initial state $\rho$, with $\G^O_{R}$ defined in Eq. (\ref{eq:jm}), and $X\in 2^{[2n]}$ a randomly sampled element from the power set of $[2n]$. Furthermore, $\G^O$ is projectively simulable.
\end{proposition}

It is straightforward to see that $\G^O$ is a parent of the POVM $\G^O_{R}$ defined in Eq. (\ref{eq:jm}), for any $R\in\mathcal{D}_{2k}$ (cf. Appendix \ref{A:pp}). Next, we show that it is also a parent POVM of any sufficiently unsharp even degree Majorana observable.

For a degree--$2k$ observable $\gamma_S$, where $S\in\mathcal{S}_{2k}$, we apply the following classical post-processing. First, we choose (at this stage arbitrarily) a set $R=\bigcup_{i=1}^{k}R_i\in\mathcal{D}_{2k}$ with $R_i\in\mathcal{D}_2$, and let $q_R=\prod_{i=1}^{k}q_{R_i}$. The set $R$ indicates which $k$ POVMs $\G^O_{R_i}$ we use to jointly measure $\gamma_S$, i.e, the POVM $G_R^O$ of Eq. (\ref{eq:jm}). We recover the measurement $\M_S(e_S)$ from the parent $\G^O$ via the post-processing of Eq. (\ref{eq:postprocessing}), with
\begin{equation}\label{eq:pp_gen}
D(e_{S}|S,R,X,{\bf q})=
\begin{cases}
1, & \!\!\text{if }\text{sgn}(\det(O_{R,S}))\,x_S q_R \\[-2pt]
   & \qquad\qquad\qquad\quad\,= e_S,\\
0, &\!\! \text{otherwise}.
\end{cases}
\end{equation}
where $x_S=(-1)^{|X|\cdot|S|-|S\cap X|}$. In particular, the outcome $e_S$ is obtained by the mapping,
\begin{equation}\label{eq:pp_map}
({\bf q},X)\mapsto \text{sgn}(\det(O_{R,S}))x_Sq_R\,.
\end{equation}
In Appendix \ref{A:pp} we show the following.
\begin{proposition}\label{prop:jm_pp}
Let $\G^O$ be the parent POVM defined in Eq. (\ref{eq:jm_gen}). For every $S\in\mathcal{S}_{2k}$, and an arbitrary $R\in\mathcal{D}_{2k}$, the post-processing defined in Eq. (\ref{eq:pp_gen}) recovers the degree--$2k$ measurement $\M_S$ of Eq. (\ref{eq:unsharp}), with sharpness
\begin{equation}\label{eq:minors}
\eta_{R,S}=|\det(O_{R,S})|\,.
\end{equation}
\end{proposition}

It follows that, for a specific $O\in O(2n)$, the choice of $R\in\mathcal{D}_{2k}$ affects how sharply the observable $\gamma_S$ is jointly measured.

\section{Asymptotically optimal joint measurements}\label{sec:nearly_optimal}

Our task is to find a suitable $O\in O(2n)$ for the parent $\G^O$ such that its relevant minors, i.e., Eq. (\ref{eq:minors}), are not too small. In particular, we wish to find $O$ such that, for every $S\in\mathcal{S}_{2k}$, the sharpness
\begin{equation}\label{eq:optimal_eta_S}
\eta_{S}=\max \{\eta_{R,S}\,|\,R\in\mathcal{D}_{2k}\}\,,
\end{equation}
is sufficiently large.

In a naive strategy, consider $O=H$, where $H$ is a $2n\times 2n$ real Hadamard matrix with elements $|h_{ij}|=1/\sqrt{2n}$ for all $i,j\in[2n]$. 
Suppose that we wish to jointly measure all degree--$2$ Majorana observables $\{\gamma_S\,|\,S\in\mathcal{S}_2\}$. The sharpness $\eta_{R,S}=|\det(O_{R,S})|$ associated with each observable is obtained from a $2\times 2$ minor of $O$. For each $S$ there are $n$ possible minors, depending on the choice of $R\in\mathcal{D}_2$. Since every $2\times 2$ minor of a $2n\times 2n$ Hadamard is either $0,$ or $\pm n^{-1}$, it follows that $\eta_{R,S}\in\{0,n^{-1}\}$.  By selecting the optimal $R\in\mathcal{D}_2$, i.e., favouring non-zero valued minors, we obtain a parent POVM that measures $\gamma_S$ with sharpness $\eta_S=n^{-1}$ for those observables where such a minor exists. Another strategy, for example, is to sample a Haar-distributed random orthogonal matrix $O$ and find the maximum minors.

Instead, to obtain sharper observables, our joint measurement $\G$ is implemented by a randomization over a set of $N$ POVMs, $\G^{(r)}$, $r=1,\ldots,N$. Each POVM is of the form described in Eq. (\ref{eq:jm_gen}) and defined by a matrix $O^{(r)}\in O(2n)$ which has a modified block diagonal structure. For example, the degree--$2$ parent is implemented by sampling, uniformly at random, one of two POVMs, $\G^{(1)}$ and $\G^{(2)}$. For a given $S\in\mathcal{S}_2$, at least one of the two POVMs measures $\gamma_S$ with sharpness $\eta_S=\Omega(n^{-1/2})$. If, in the worst case scenario, the sampled POVM contains no useful information about $\gamma_S$ (i.e., the corresponding sharpness parameter is $\eta_S=0$), the relevant marginal is the trivial observable $\M_S(e_S)=\id/2$ whose outcome can be drawn at random from a coin flip. After taking the randomization over both POVMs into account, the sharpness $\eta_S=\Omega(n^{-1/2})$ is reduced by a factor of $1/2$. More generally, provided at least one of the $N$ POVMs approximates the observable with sufficient sharpness, the randomization will reduce $\eta_S$ by a factor of at most $1/N$.

In the following theorem, we place a lower bound on the quantity $\eta_{2k}$ defined in Eq. (\ref{eq:min_sharpness}),  where $k$ is a constant independent of the system size $n$, i.e., $k=\mathcal{O}(1)$.
\begin{theorem}\label{thm:jm}
Assume $k=\mathcal{O}(1)$ and $N>4k$. Let $\G$ be a POVM implemented by a uniform randomisation over $N$ POVMs $\G^{(1)},\ldots,\G^{(N)}$, where $\G^{(r)}:= \G^{(O^{(r)})}$ is defined in Eq. \eqref{eq:jm_gen}. Then there exists a set of orthogonal matrices $O^{(1)},\ldots,O^{(N)}$ for which $\G$ jointly measures all degree--$2k$ observables $\{\gamma_S\,|\,S\in \mathcal{S}_{2k}\}$ of an $n$ mode fermionic system with sharpness
\begin{equation}\label{eq:eta_2k_lb}
\eta_{2k}=\Omega(n^{-k/2})\,.
\end{equation}
Furthermore, when $k = 1$, the result holds with only $N = 2$ orthogonal matrices.
\end{theorem}

In Sec. \ref{sec:ir_syk} we will show that the parent from Theorem \ref{thm:jm} is asymptotically optimal, i.e., $\eta_{2k}$ has the same asymptotic scaling as the incompatibility robustness.

The proof of Theorem~\ref{thm:jm} follows by constructing a suitable collection of orthogonal matrices, as summarised in the next two subsections. For degree--2 observables, we explicitly construct two orthogonal matrices $O^{(1)}$ and $O^{(2)}$ that yield a joint measurement achieving the asymptotic lower bound for sharpness of $\Omega(n^{-1/2})$. For degree--$2k$ observables, we provide a randomised construction of $O^{(1)},\ldots,O^{(N)}$, which yields a suitable set of orthogonal matrices (i.e., such that the joint measurement satisfies Eq. \eqref{eq:eta_2k_lb}) with high probability. In the rare event that a randomly sampled set fails, the construction can be repeated until success. Further details and proofs can be found in Appendices \ref{A:jm_quadratic}, \ref{A:arbitrary_n} and \ref{A:jm_quartic}.

\subsection{Degree--$2$ Majoranas}\label{sec:jm_pair}

\begin{figure}[!t]
    \centering
	\includegraphics[width=8cm]{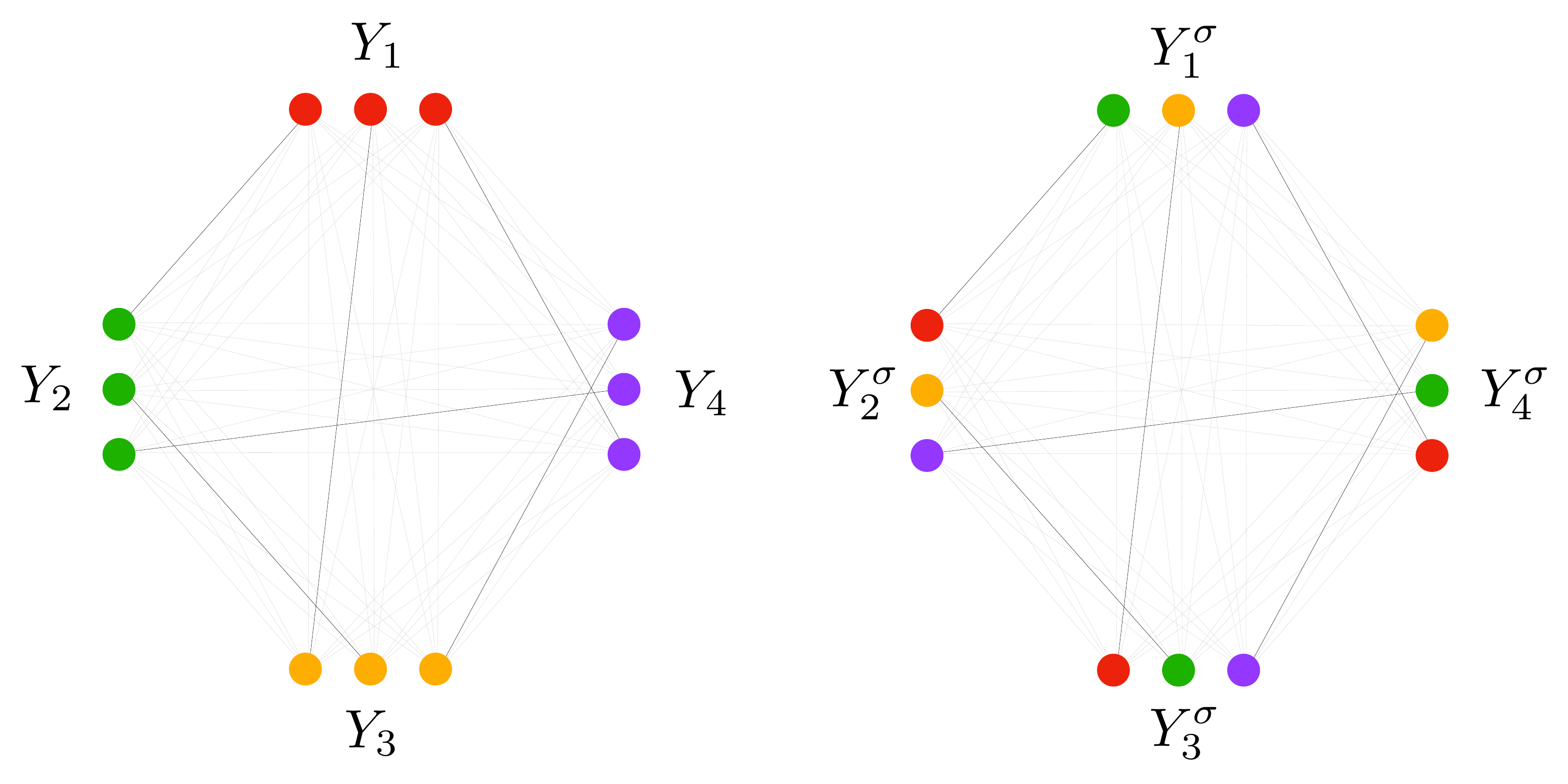}
	\caption{(Left.) A perfect matching $\mathcal{M}$ of the Turán graph $T(12,4)$, with 12 vertices partitioned into 4 (coloured) subsets $Y_1,\ldots, Y_4$, each containing 3 vertices. The sparse arrangement ensures that, given any two partition subsets, only one edge in the perfect matching connects them. (Right.) A second Turán graph is obtained from the first via a permutation $\sigma$ of the vertices, exchanging $v_1\leftrightarrow v_2$ if and only if $\{v_1,v_2\}\in\mathcal {M}$. Consequently, each partition subset $Y_1^{\sigma}\ldots, Y_4^{\sigma}$, has no identically coloured vertices.}
    \label{fig:perfect_matching}
\end{figure}

To specify $O^{(r)}$, we initially restrict the system size to $2n=\ell(\ell+1)$, with $\ell\in\mathbb{Z}^+$. 
We drop this assumption and generalize to an arbitrary $n$ mode system in Appendix \ref{A:arbitrary_n}. To ensure all quadratic Majorana observables are measured simultaneously with the required noise (Thm. \ref{thm:jm}), we randomize over two orthogonal rotations. We label each POVM $\G^{(j)}$, $j=1,2$, and the associated $2n\times 2n$ orthogonal matrices
\begin{equation}\label{eq:matrixQ}
O^{(1)}=P_{\pi} \cdot D\,\,\,\,\text{and}\,\,\,\, O^{(2)}=O^{(1)}\cdot P_{\sigma}\,,
\end{equation}
respectively, where $D=\bigoplus_{i=1}^{\ell+1}F^{(i)}$ is a block diagonal matrix, and $F^{(i)}\in O(\ell)$ is a lower-flat matrix (cf. Appendix \ref{A:jm_quadratic}). Lower-flat orthogonal matrices are variants of \emph{flat} orthogonal matrices, first introduced in \cite{jaming15}. The essential feature of an $\ell\times \ell$ lower-flat orthogonal matrix is that the magnitude of each matrix entry is lower bounded by $c\cdot \ell^{-1/2}$, for some positive constant $c$. One such example is a Hadamard matrix, although lower-flat orthogonal matrices exist for every order \cite{jaming15}. 

The permutation matrices $P_{\pi}$ and $P_{\sigma}$ act by permuting the rows and columns of $D$, respectively, via $\pi,\sigma\in S_{2n}$, where $S_{2n}$ is the symmetric group over $[2n]$. To define the permutations, we make use of perfect matchings of Turán graphs (see, e.g., Fig. \ref{fig:perfect_matching}). A \emph{perfect matching} $\mathcal{M}$ of a graph $G = (V, E)$ is a set of disjoint edges $\{v_1, v_2\} \in E$ such that every vertex $v \in V$ belongs to exactly one edge in $\mathcal{M}$. In particular, if $|V| = 2n$, then a perfect matching consists of exactly $n$ such edges. A \emph{Turán graph} $T(2n, \ell+1)$ is defined by partitioning the vertex set $V$ of size $2n$ into $\ell+1$ subsets of as equal size as possible, and connecting two vertices by an edge if and only if they belong to different subsets. An example of a perfect matching of a Turán graph $T(2n, \ell+1)$ is illustrated in Fig.~\ref{fig:perfect_matching} for the case $n = 6$, $\ell = 3$. For our joint measurement construction to work, we require the perfect matching to be \emph{sparsely arranged}. A sparsely arranged perfect matching $\mathcal{M}$ of a Turán graph is one in which every pair of partition subsets is connected by exactly one edge in $\mathcal{M}$.

Given $\mathcal{M}$, we choose a permutation $\pi \in S_{2n}$ that reorders the vertex labels such that the matching edges are mapped to pairs $\{2j-1,2j\}$, i.e., elements in $\mathcal{D}_2$. In particular, $\pi$ maps the standard pairings of $\mathcal{D}_2$ back to the original matched pairs, namely
\begin{equation}\label{eq:matching_perm}
\mathcal{M}=\{\,\{\pi^{-1}(2j-1),\pi^{-1}(2j)\}\,|\,j=1,\ldots,n\}\,.
\end{equation}

The second permutation $\sigma\in S_{2n}$ is chosen to exchange $v_1\leftrightarrow v_2$ if and only if $\{v_1,v_2\}\in\mathcal {M}$. Such a permutation maps one Turán graph to another by shuffling the elements into different partition subsets. In particular, each partition subset now contains vertices that were previously in different partition subsets (see Fig. \ref{fig:perfect_matching}, right).

To illustrate the structure of $O^{(1)}$ and $O^{(2)}$ in Eq. \eqref{eq:matrixQ}, we provide a pictorial representation for $n=6$ in Fig. \ref{fig:matrix_Q}. To achieve the scaling behaviour in Theorem \ref{thm:jm}, we require that for every $S\in\mathcal{S}_2$, there exists $R\in\mathcal{D}_2$ such that $\G^{(1)}$ or $\G^{(2)}$ jointly measures $\gamma_S$, with $\eta_{R,S}=|\det(O^{(r)}_{R,S})|=\Omega(\ell^{-1})$. Consider $O^{(1)}$ from Fig. \ref{fig:matrix_Q}, which we denote by $O=O^{(1)}$ for simplicity. For a pair of adjacent rows $R=\{2j-1,2j\}$, a $2\times 2$ submatrix $O_{R,S}$ is monomial when $S=\{j_1,j_2\}$ labels two columns containing elements of different colours. For example, if $R=\{1,2\}$, then $O_{R,S}$ is monomial if $j_1\in\{1,2,3\}$ and $j_2\in\{4,5,6\}$. The monomiality of $O_{R,S}$ ensures $\eta_{R,S}=\Omega(\ell^{-1})$, since its non-zero entries are elements from lower-flat matrices $F^{(i)}\in O(\ell)$, i.e., $\min_{u,v\in[\ell]}|f^{(i)}_{uv}|=\Omega(\ell^{-1/2})$ \cite{jaming15}. On inspection, $O$ yields the required sharpness for all $\gamma_S$, $S\in\mathcal{S}_2$, except when $S\subset\{1,2,3\},\{4,5,6\},\ldots,\{10,11,12\}$. In these exceptional cases, all relevant submatrices are singular. However, $O^{(2)}$ is constructed to ensure $\G^{(2)}$ jointly measures these remaining observables.

\begin{figure}[!t]
    \centering
	\includegraphics[width=8cm]{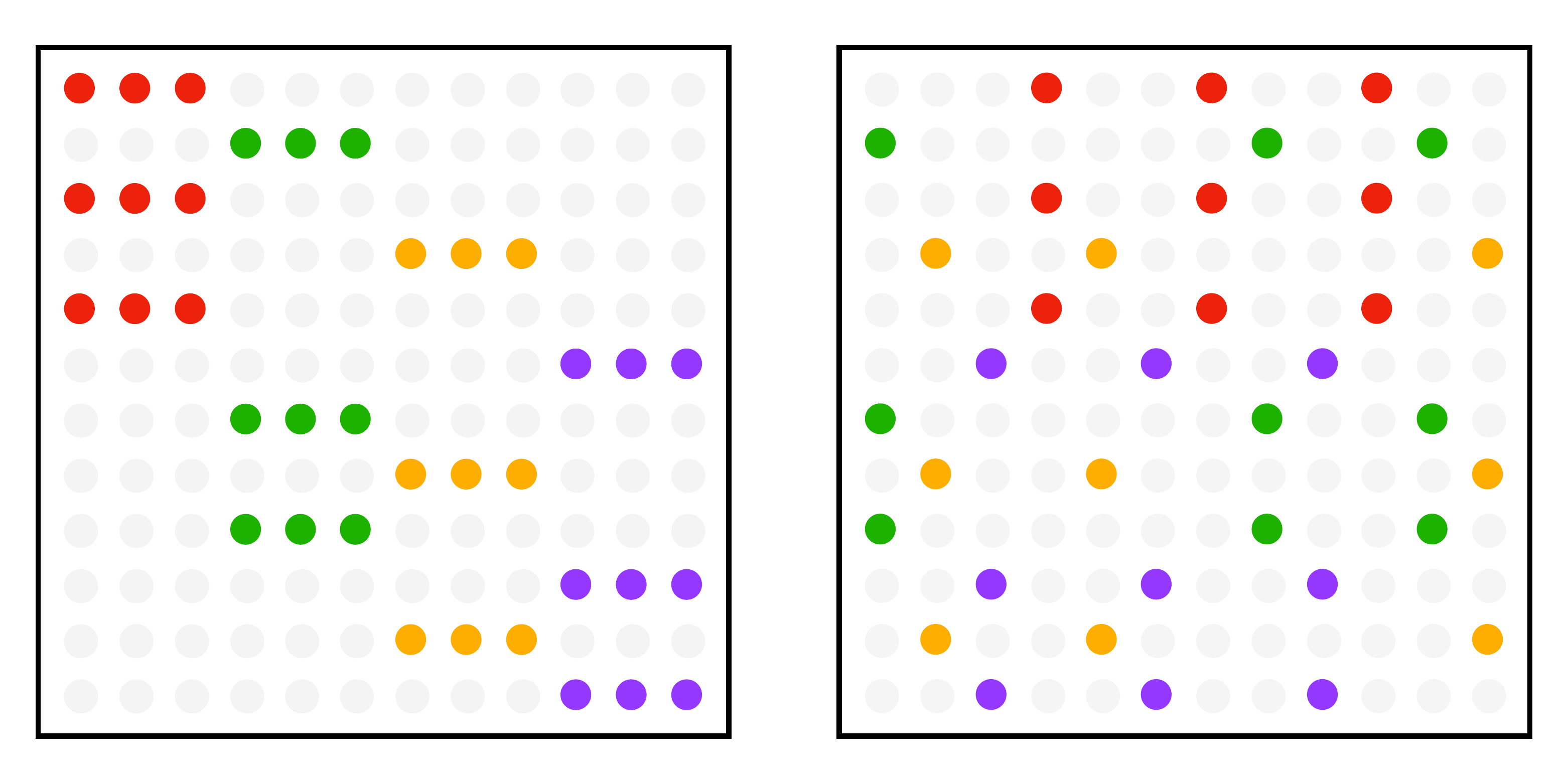}
	\caption{Pictorial representation, for $n=6$, of the orthogonal matrices $O^{(1)}$ (left) and $O^{(2)}$ (right) defining the POVMs $\G^{(1)}$ and $\G^{(2)}$, respectively. The block diagonal matrix $D$ consists of four $3\times3$ lower-flat orthogonal matrices $F^{(i)}$, each assigned a distinct colour, e.g., the elements of $F^{(1)}$ are coloured red. Uncoloured entries represent zeros. The matrix $O^{(1)}$ is obtained by a permutation of the rows of $D$, as defined by a perfect matching of the Turán graph $T(12,4)$. The matrix $O^{(2)}$ is obtained by a permutation of the columns of $O^{(1)}$, as defined by a mapping between the vertices of two Turán graphs (see Fig. \ref{fig:perfect_matching}).}
    \label{fig:matrix_Q}
\end{figure}

To illustrate the proof of Theorem \ref{thm:jm} (given in Appendices \ref{A:jm_quadratic} and \ref{A:arbitrary_n}), it is useful to rewrite the joint measurement of Eq. (\ref{eq:jm_gen}) in terms of a perfect matching $\mathcal{M}$. In particular, $\mathcal{M}$ generates a set of $n$ POVMs $\{\G^O_M\,|\,M\in\mathcal{M}\}$ which constitute the measurement,
\begin{equation}\label{eq:jm_matching}
\widetilde\G^{O}=(2^{2n})^{n-1}\prod_{M\in\mathcal{M}}\G^O_{M}\,,
\end{equation}
where $\G^O_M$ takes the form of Eq. (\ref{eq:jm}), replacing $R\in\mathcal{D}_2$ with an edge $M\in \mathcal{M}$. Applying the rotation $O^{(1)}=P_{\pi}\cdot D$ in Eq. (\ref{eq:jm_gen}) is now equivalent to applying $\widetilde O^{(1)}=D$ in Eq. (\ref{eq:jm_matching}), i.e., $\G^{P_{\pi}\cdot D}=\widetilde \G^D$, where $\pi\in S_{2n}$ is the permutation that defines the perfect matching (cf. Eq. \eqref{eq:matching_perm}).

The block diagonal matrix $D$, acting on the set of $2n$ Majorana operators, partitions them into $\ell+1$ subsets of cardinality $\ell$, with each subset rotated by a lower-flat matrix $F^{(i)}\in O(\ell)$, $i=1,\ldots,\ell+1$. In this setting, it is natural to describe the joint measurability properties of $\widetilde\G^D$ using a Turán graph $T(2n,\ell+1)$.
%The partition into subsets $Y_i$ (of cardinality $\ell$) follows from the block diagonal structure of $D$, which rotates the Majorana operators in $Y_i$ by a lower-flat matrix $F^{(i)}\in O(\ell)$ for $i=1,\ldots,\ell+1$. 
%Note that a Turán graph contains no edges within the same partition subset, avoiding the possibility of a singular submatrix, as described earlier.

If $M\in\mathcal{M}$ is an edge of a perfect matching of $T(2n,\ell+1)$ it connects two vertices from distinct subsets. It is simple to check (Prop. \ref{prop:edge}, Appendix \ref{A:jm_quadratic}) that the POVM $\G^{D}_M$ jointly measures any sufficiently unsharp degree--2 observable which consists of a Majorana operator from each of the two subsets connected by $M$, i.e., all observables labelled by the edges between both subsets. To maximise the number of observables that are jointly measured by the $n$ POVMs $\{\G^{D}_M\,|\,M\in\mathcal{M}\}$, we choose a perfect matching that is sparsely arranged. Consequently, the $n$ POVMs of the perfect matching simultaneously measure all (sufficiently unsharp) degree--$2$ observables labelled by the edges of $T(2n,\ell+1)$. To jointly measure the remaining observables, the POVM $\widetilde\G^{D\cdot P_{\sigma}}$ is described by the same Turán graph but with a permutation of the vertices into different subsets. 

For example, in Fig. \ref{fig:perfect_matching} (left), the partition sets of the Turán graph are assigned a unique colour and the sparsely arranged perfect matching ensures any pair of distinct colours are connected by an edge. Thus, $\widetilde\G^{D}$ of Eq. (\ref{eq:jm_matching}) jointly measures all quadratic observables whose constituents are  distinctly coloured. The perfect matching and its permutation $\pi\in S_{2n}$ (cf. Eq. (\ref{eq:matching_perm})) specifies the matrix $P_{\pi}$ which permutes the rows of $D$ to obtain $O^{(1)}$ in Fig. \ref{fig:matrix_Q}. In Fig. \ref{fig:perfect_matching} (right), the permutation $\sigma\in S_{2n}$, maps the vertices of the first Turán graph to the second, such that identically coloured vertices are sent to different partition subsets. This defines the matrix $P_{\sigma}$ used to permute the columns of $D$, ensuring that all unmeasured observables are recovered by $\widetilde\G^{D\cdot P_\sigma}$, or equivalent $\G^{(2)}$. A complete proof, with bounds on the sharpness, is provided in Appendix \ref{A:jm_quadratic}. A generalisation to Turán graphs $T(2n,\ell+1)$, with $n$ arbitrary, is presented in Appendix \ref{A:arbitrary_n}.

\subsection{Degree--$2k$ Majoranas}\label{sec:jm_quartic}

To jointly measure unsharp degree--$2k$ Majorana observables $\M_S$ with $\eta_{2k}=\Omega(n^{-k/2})$, we employ a similar joint measurement scheme that requires a randomisation over $N$ POVMs $\G^{(r)}:=\G^{O^{(r)}}$, $r=1,\ldots,N$, each defined by a $2n\times 2n$ orthogonal matrix $O^{(r)}$  (cf. Eq. \eqref{eq:jm_gen}). Unlike the degree--$2$ case, we do not give an explicit closed-form construction of the $N$ orthogonal matrices. Instead, $\{O^{(r)}\}_{r=1}^N$ is constructed via a randomised procedure that yields a suitable set with high probability.

In particular, each POVM $\G^{(r)}$ is defined via a modified version of the matrix $O=P_{\pi} \cdot D$ from the degree--$2$ strategy (cf. Eq. (\ref{eq:matrixQ})), namely,
\begin{equation}\label{eq:quartic_matrixQ}
O^{(r)}=P_{\pi} \cdot D \cdot P_{\sigma_r}\,,
\end{equation} 
with $P_{\sigma_r}$ a permutation matrix defined by a randomly sampled $\sigma_r\in S_{2n}$. In particular, $O^{(r)}$ is obtained by permuting the columns of $P_{\pi} \cdot D$ with a permutation sampled uniformly at random from $S_{2n}$.

We prove in Appendix \ref{A:jm_quartic} that $N>4k$ randomly sampled permutations $\sigma_r$ are sufficient to ensure (with high probability) that all $\{\gamma_S\,|\,S\in\mathcal{S}_{2k}\}$ are jointly measured by $\G$ with $\eta_{2k}$ satisfying the desired bound. Explicitly, for the set $\{O^{(r)}\}_{r=1}^N$, we show that given any $S=\{j_1,\ldots,j_{2k}\}$, there exists $R\in\mathcal{D}_{2k}$, and at least one $r\in[N]$, such that $O^{(r)}_{R,S}$ is monomial. Since the determinant of the monomial matrix is determined by the product of $2k$ entries, each from an $\ell\times\ell$ lower-flat matrix, it follows that $\eta_{S,R}=\Omega(n^{-k/2})$.

To simplify the proof we again consider the perfect matching description of the joint measurement in Eq. (\ref{eq:jm_matching}), with each POVM $\widetilde\G^{(r)}$ defined by a random partition of the vertices of a Turán graph into $\ell+1$ subsets, described by $\sigma_r\in S_{2n}$. The associated rotation is given by $\widetilde O^{(r)}=D\cdot P_{\sigma_r}$ and the perfect matching $\mathcal{M}$ is again sparsely arranged. An observable $\gamma_S$ is jointly measured by $\widetilde\G^{(r)}$, with sufficient sharpness, if each element of $S\in\mathcal{S}_{2k}$ labels a vertex from a \emph{different} partition subset of the Turán graph, i.e., the $2k$ elements of $S$ lie in $2k$ disjoint partition subsets. In Prop. \ref{prop:partitions} (Appendix \ref{A:jm_quartic}) we calculate the probability of failing to jointly measure $\gamma_S$ from a single partition, and show (in Cor. \ref{cor:partitions}) that after $N>4k$ partitions of $T(2n,\ell+1)$, the probability of failing to jointly measure all observables approaches zero as $n$ increases.  

\section{Incompatibility robustness and the SYK model}\label{sec:ir_syk}

We now investigate the incompatibility robustness of degree--$k$ Majorana observables, initially dropping the restriction to even degree observables. In particular, for the measurements $\M_S(\pm)=\frac{1}{2}(\id\pm\eta\gamma_S)$, we evaluate the quantity, 
\begin{equation}\label{eq:inc_rob_def}
\eta_k^*=\sup\{\eta\,|\,\M_S \,\,\text{are compatible}\,\, \forall S\in\mathcal{S}_k\}\,,
\end{equation}
which, equivalently, is the minimal noise $(1-\eta)$ needed to ensure the existence of a parent POVM. Note that $\eta_{2k}^*$ provides an upper bound on $\eta_{2k}$ in Eq. (\ref{eq:min_sharpness}) for any parent POVM $\G$ since $\eta_{2k}^*=\max_{\G} \eta_{2k}$.

While incompatibility robustness is fully characterized for pairs of qubit observables \cite{busch86,yu10} as well as pairs of mutually unbiased bases \cite{carmeli19,kiukas22}, semidefinite programming is usually required for more general measurement assemblages \cite{cavalcanti17,designolle19a,designolle19b}. However, this approach quickly becomes infeasible for large collections of measurements due to the number of variables involved. Nevertheless, for sets of observables satisfying certain symmetry conditions, finding the incompatibility robustness via an SDP reduces to evaluating the maximum eigenvalue of a particular class of operators \cite{nguyen20}. In particular, given a measurement assemblage and a symmetry group $G$ which permutes the outcomes while maintaining the bundle structure (see Appendix \ref{Asec:symmetries}), the assemblage is required to be uniformly and rigidly symmetric for such a simplification to apply. Uniformity holds if each outcome is related to every other outcome by a symmetry transformation. Rigidity holds if the stabilizer group of each outcome commutes with a single projection operator and its complement. In particular, a unitary representation $U:G\rightarrow U(d)$, restricted to the stabilizer group, has exactly two irreducible subrepresentations. In Appendix \ref{Asec:symmetries}, we show that the assemblage of degree--$k$ Majorana observables (with $k$ even or odd) is uniformly and rigidly symmetric with respect to the braid group. In the spinor representation (see Appendix \ref{Asec:braids}), a braiding transformation acts, via conjugation, on the Majorana operators as,
\begin{equation}\label{eq:action_on_maj}
B_{i,j}\gamma_{\ell} B_{i,j}^{\dagger}=
\begin{cases}
 \gamma_{\ell} &\mbox{if} \,\,\,{\ell}\notin \{i,j\}\,, \\
 \gamma_j &\mbox{if}\,\,\, {\ell}=i \,,\\
- \gamma_i &\mbox{if}\,\,\, {\ell}=j \,.
 \end{cases}
\end{equation}

The symmetry properties ensure that calculating the incompatibility robustness (for each $k$) reduces to evaluating the operator norm, i.e., finding the maximum eigenvalue, of 
\begin{equation}\label{eq:syk}
H_k=\sum_{S\in\mathcal{S}_{k}}e_{S}\gamma_S\,,
\end{equation}
where $e_S\in\{\pm 1\}$ for all $S\in\mathcal{S}_{k}$.
In Appendix \ref{Asec:syk} we show that the incompatibility robustness can be expressed as follows.
\begin{lemma}\label{lem:robustness}
The set $\{\gamma_S\,|\,S\in \mathcal{S}_{k}\}$ of degree--$k$ Majorana observables of an $n$ mode fermionic system has incompatibility robustness,
\begin{equation}\label{eq:inc_rob}
\eta_k^*= \frac{1}{|\mathcal{S}_{k}|}\max_{e_S\in\{\pm 1\}\forall S\in\mathcal{S}_k}{\no{H_k}}_{\infty}\,.
\end{equation}
\end{lemma}
Interestingly, the spectrum of the operator $H_k$ has been extensively studied in relation to the Sachdev-Ye-Kitaev (SYK) model \cite{sachdev93,kitaev15,feng19,hastings21,herasymenko22}, which is of independent interest in condensed matter physics and string theory.
We note that an upper bound on the incompatibility robustness for binary observables can be derived without symmetry considerations \cite{kunjwal14}, which for Majorana observables is equivalent to the right-hand side of Eq. (\ref{eq:inc_rob}).

We now invoke several results on the spectral behavior of $H_k$ to derive bounds (sometimes tight) on the incompatibility robustness, as summarized in the next theorem.
\begin{theorem}\label{thm:inc_rob}
Let $\{\gamma_S\,|\,S\in \mathcal{S}_{2k}\}$ denote the set of degree--$2k$ Majorana observables of an $n$ mode fermionic system.
\begin{itemize}
\item[(i)] For $k=1$, the incompatibility robustness satisfies $\eta_2^*\leq \frac{1}{\sqrt{2n-1}}$, and is saturated if and only if a $2n\times 2n$ (real) skew-Hadamard matrix exists. Furthermore, the joint measurement in Sec. \ref{sec:jm_pair} is asymptotically optimal, i.e., $\eta_2^*= \Theta(n^{-1/2})$.
\item[(ii)] For $k\leq 5$, and all $n$ sufficiently large, the incompatibility robustness satisfies $\eta_{2k}^*\leq \sqrt{\binom{n}{k}/\binom{2n}{2k}}$. Furthermore, the joint measurement in Sec. \ref{sec:jm_quartic} is asymptotically optimal, i.e., $\eta_{2k}^*= \Theta(n^{-k/2})$.
\end{itemize}
\end{theorem}
We sketch the proofs of parts $(i)$ and $(ii)$ in the next two subsections (with further details in Appendix \ref{A:IR}).

\subsection{Degree--$2$ Majoranas}\label{sec:pairs}

The coefficients of $H_2=i\sum_{i<j}e_{ij}\gamma_i\gamma_j$ define a unique $2n\times 2n$ real antisymmetric matrix $E=-E^{T}$, with entries $e_{ij}=-e_{ji}$ for $i<j$. As described in Appendix \ref{A:IR}, the operator $H_2$ can be mapped via an orthogonal matrix $O\in O(2n)$ to a free fermion Hamiltonian, $H_2=i\sum_{j=1}^n\lambda_j\,\widetilde \gamma_{2j-1}\widetilde \gamma_{2j}$, where $\lambda_j\in\mathbb{R}$ and $\pm i\lambda_j$ are the eigenvalues of $E$.

Since $H_2$ is now a sum of commuting projection operators (with eigenvalues $\pm 1$), the eigenvalues of $H_2$ have the form $\sum_j \pm \lambda_j$. Rewriting in terms of fermionic operators it follows that the eigenvalues of $H_2$ are given by $\lambda_{\bf m}=\sum_{j=1}^{n}\lambda_j(-1)^{m_j}$, where ${\bf m}=(m_1,m_2,\ldots,m_n)$ with $m_j\in\{0,1\}$. The maximum eigenvalue of $E$ is achieved when $m_j=1$ for $\lambda_j<0$ and $m_j=0$ for $\lambda_j\geq 0$ such that ${\no {H_2}}_{\infty}=\sum_{j=1}^{n}|\lambda_j|$. We now find a bound on the maximum eigenvalue of any matrix from the set of $2n\times 2n$ matrices $\mathcal{E}_{2n}=\{E\,|\,E=-E^{T}, \,\,e_{ij}\in\{\pm 1\}, \forall i\neq j\}$. The following result is proved in Appendix \ref{A:IR} using the theory of tournament matrices \cite{Wesp02,Ito17}.

\begin{lemma}\label{lem:hboundpair}
Let $E\in\mathcal{E}_{2n}$, with eigenvalues $\pm i\lambda_j$. Then, $\sum_{j=1}^{n}|\lambda_j|\leq n\sqrt{2n-1}$. The bound is tight if and only if there exists $E\in\mathcal{E}_{2n}$ such that $E+\id$ is a $2n\times 2n$ (real) Hadamard matrix.
\end{lemma}

It follows directly from Lemma \ref{lem:hboundpair} that $\max_{E\in\mathcal{E}_{2n}}{\no {H_2}}_{\infty}\leq n\sqrt{2n-1}$. Applying Eq. (\ref{eq:inc_rob}) with $|\mathcal{S}_2|=n(2n-1)$, we arrive at the upper bound provided in part $(i)$ of Theorem \ref{thm:inc_rob}. Note that the bound cannot be saturated for all system sizes. An $n\times n$ real Hadamard matrix does not exist if $n$ is not a multiple of 4 (except when $n=2$). 
However, Hadamard's conjecture suggests they exist in all other dimensions \cite{paley33}, which is confirmed for $n<668$ \cite{kharaghani05,bruzda}. A review of skew-Hadamard matrices and their constructions is given in \cite{koukouvinos08}; the first instance in which their existence is open is $n=276$.

Together with Theorem \ref{thm:jm}, which lower bounds $\eta_{2}$ via the joint measurement described in Sec. \ref{sec:jm_pair}, we conclude $\eta_2^*=\Theta(n^{-1/2})$.

\subsection{Degree--$2k$ Majoranas}\label{sec:quad}

In relation to the SYK model, most attention has been directed towards bounding the maximum eigenvalue of a random degree--$4$ Hamiltonian of the form (\ref{eq:syk}), with coefficients $e_S$ chosen to be independent random Gaussians, such that $\sum_{S\in\mathcal{S}_k}e_S^2=1$ in expectation. While physical arguments suggest the maximum eigenvalue of any such operator drawn randomly from the degree--$4$ SYK model satisfies $\Theta(\sqrt{n})$, with high probability, a mathematical proof was only recently provided \cite{hastings21}. 

The proof relies on a general upper bound, given in [Theorem 1.5, \cite{hastings21}], for operators of the form (\ref{eq:syk}) with no constraints on the coefficients except for $\sum_{S}e_S^2=1$. We restate this result below, with a rescaling of the coefficients.

\begin{lemma}\label{thm:odonnell}
For an $n$ mode fermionic system, let $H_{2k}=\sum_{S\in\mathcal{S}_{2k}}e_{S}\gamma_S$, with $e_S\in\mathbb{R}$ and $\sum_{S}e_S^2=\binom{2n}{2k}$. If $k\leq 5$, then
\begin{equation}\label{eq:odonnell}
{\no{H_{2k}}}_{\infty}\leq \sqrt{\binom{2n}{2k}\binom{n}{k}}\,,
\end{equation}
for all $n$ sufficiently large.
\end{lemma}

The proof, which we do not reproduce here, is found in \cite{hastings21}. The argument relies on a graph theoretic approach, where the Majorana monomials of the Hamiltonian are represented as vertices of a graph, and the anti-commutation relations are described by its edges. Evaluating the Lovász number of a special class of Kneser-type graphs yields the desired bound.

Since Lemma \ref{thm:odonnell} includes Hamiltonians with coefficients limited to $e_S\in\{\pm 1\}$, it can be directly applied to bound the incompatibility robustness (of Lemma \ref{lem:robustness}) for all degree--$2k$ Majorana observables with $k\leq 5$. In particular, it follows from Eqs. (\ref{eq:inc_rob}) and (\ref{eq:odonnell}), together with $|\mathcal{S}_{2k}|=\binom{2n}{2k}$, that $\eta_{2k}^*\leq \sqrt{\binom{n}{k}/\binom{2n}{2k}}$. With the joint measurement from Theorem \ref{thm:jm} we obtain $\eta_{2k}^*=\Theta(n^{-k/2})$, as stated in part $(ii)$ of Theorem \ref{thm:inc_rob}.

\section{Estimating fermionic observables}\label{sec:est}

We now employ our joint measurement for two estimation tasks: fermionic partial tomography and estimating energies of Hamiltonians. Note that due to fermionic parity conservation, most physical observables involve linear combinations of even degree Majorana operators. Approximating $k$--body reduced density matrices ($k$--RDMs) of a quantum state $\rho$ requires estimating $\tr{\gamma_S\rho}$ for each degree--$2k$ observable $\gamma_S$  \cite{Zhao21,Jiang20}. For instance, knowledge of $k$--RDMs, with $k\leq 4$, has useful applications in error mitigation, modeling electrons in condensed matter systems and excited-state calculations \cite{takeshita20,wang21,peterson13,parrish19}. Approximating energies of chemical Hamiltonians, e.g., to find binding energies of molecules with $2$--body interactions, requires estimating linear combinations of degree--$2$ and degree--$4$ Majorana observables \cite{huggins21,wan22}.

\subsection{Local observables}\label{sec:local}

The joint measurement $\G^O$ in Eq. (\ref{eq:jm_gen}) allows us to simultaneously estimate the expectation values $\tr{\gamma_S\rho}$ of all degree--$2k$ Majorana observables via the single shot estimators $\hat\gamma_{S}=\eta_S^{-1}e_S$. Here,  $\eta_S$ is defined in Eq. (\ref{eq:optimal_eta_S}) and $e_S\in\{\pm1\}$ is the outcome of the unsharp Majorana observable $\M_S$ obtained from $\G^O$ via the mapping of Eq. (\ref{eq:pp_map}), i.e., $e_{S}=\text{sgn}(O_{R,S})x_Sq_R$. It is easy to check that $\mathbb{E}[\hat\gamma_{S}]=\eta_S^{-1}\mathbb{E}[e_S]=\tr{\gamma_S\rho}$, where the expectation value $\mathbb{E}$ is over the outcome statistics of the POVM $\G^O$ on the state $\rho$.

For a fixed $O\in\mathcal{O}(2n)$, let $\hat\gamma_{S}^{(j)}$ denote the estimator of $\tr{\gamma_S\rho}$ from the $j$--th implementation of the joint measurement $\G^O$. We wish to determine how many copies $L_{2k}$ of the state $\rho$ are needed to bound the error of $\hat\Gamma_{S}=\frac{1}{L_{2k}}\sum_{j=1}^{L_{2k}} \hat\gamma_{S}^{(j)}$, for all $S\in\mathcal{S}_{2k}$, such that $|\tr{\gamma_{S}\rho}-\hat \Gamma_{S}|< \epsilon$, with probability at least $1-\delta$. Since the single shot estimators are binary, Hoeffding's inequality provides an effective way to bound the sample complexity, and avoids the less efficient median-of-means approach \cite{hoeffding63}. In particular, as $\hat\gamma_{S}^{(j)}\in\{\pm \eta_S^{-1}\}$, then
\begin{equation}
\text{Pr}\left(|\tr{\gamma_S\rho}-\hat\Gamma_S|<\epsilon\right)>1-2e^{-L_{2k}(\epsilon\eta_S)^2/2}\,,
\end{equation}
for each $S\in \mathcal{S}_{2k}$. If we require each probability failure to be no more than $\delta/|\mathcal{S}_{2k}|$ and take the union bound over all $|\mathcal{S}_{2k}|$ events, we succeed with probability at least $1-\delta$ by setting $2e^{-L_{2k}(\epsilon\eta_{2k})^2/2}=\delta/|\mathcal{S}_{2k}|$. Applying Theorem \ref{thm:jm}, the number of samples $L_{2k}$ required to obtain the confidence interval $1-\delta$ is given below.

\begin{proposition}\label{prop:sample_complexity}
Let $\{\gamma_S\,|\,S\in\mathcal{S}_{2k}\}$ denote the set of degree--$2k$ Majorana observables of an $n$ mode fermionic system. For $k=\mathcal{O}(1)$, the set of degree--$2k$ observables can be simultaneously estimated via the joint measurement of Eq. (\ref{eq:jm_gen}), up to additive error $\epsilon$ with probability at least $1-\delta$, with sample complexity
\begin{equation}
L_{2k}=\mathcal{O}\left(\frac{2n^{k}}{\epsilon^2}\log\left(\frac{2|\mathcal{S}_{2k}|}{\delta}\right)\right)\,.
\end{equation}
\end{proposition}

%For odd degree Majorana operators the joint measurement estimation scheme shares the same sample complexity scaling as the quadratic ($k=1$) case.

\subsection{Hamiltonians}\label{sec:hamiltonians}

We define a $2q$--local Hamiltonian, which consists of Majorana monomials up to degree--$2q$, as 
\begin{equation}\label{eq:ham}
H=\sum_{k=1}^q\sum_{S\in \mathcal{S}_{2k}}\alpha_{S}\gamma_{S}\,,
\end{equation}
with $\alpha_S\in\mathbb{R}$. We construct an unbiased estimator of $\tr{H\rho}$ from the outcomes of a single joint measurement $\G^O$ defined in Eq. (\ref{eq:jm_gen}), with arbitrary $O\in O(2n)$. For simplicity, we assume the joint measurement strategy requires only a single $O$, but the calculations generalize straightforwardly to larger sets. Furthermore, we assume the estimator of an observable $\gamma_S$ is obtained from the post-processing of $\G^O$ (see Eq. (\ref{eq:pp_gen})) via the choice $R(S)\in\mathcal{D}_{2k}$ such that $\eta_S=|\det(O_{R(S),S})|$, as described in Eq. (\ref{eq:optimal_eta_S}). We will use $R=R(S)$ as shorthand. The estimator is given by 
\begin{equation}\label{eq:ham_est}
\hat H=\sum_{k=1}^q\sum_{S\in\mathcal{S}_{2k}}\alpha_{S}\eta_S^{-1}e_{S}\,,
\end{equation}
where $e_{S}\in\{\pm 1\}$. It follows that $\mathbb{E}[\hat H]=\sum_{S}\alpha_{S}\eta_S^{-1}\mathbb{E}[e_{S}]=\sum_{S}\alpha_{S}\tr{\gamma_{S}\rho}=\tr{H\rho}$. The variance of the estimator can be found by evaluating $\mathbb{E}[e_{S}e_{S'}]$ for all pairs of observables $\gamma_S$ and $\gamma_{S'}$, of possibly different degrees. This involves calculating the two-observable marginals $\M_{S,S'}(e_{S},e_{S'})$ from the joint measurement $\G^O$, as shown in Appendix \ref{A:2-obs-marginals}.
\begin{proposition}\label{cor:var}
The variance of the estimator $\hat H$ of a $2q$--local Hamiltonian $H$ is given by
\begin{multline*}
\var[\hat H]=\sum_{\substack{S,S'\\S\neq S'}}f(S,S')\frac{\alpha_{S}\alpha_{S'}}{\nu_S\nu_{S'}}\nu_{S\triangle S'}\emph{tr}[\gamma_{S\triangle S'}\rho]\\
+\sum_{S}\alpha^2_{S}|\nu_S|^{-2}-(\tr{H\rho})^2\,,
%+\sum_{\substack{R,R'\\ R\neq R'}}\alpha_{R}\alpha_{R'}\emph{tr}[\gamma_R\gamma_{R'}\rho]\\-(\emph{tr}[H\rho])^2\,,
\end{multline*}
where $S\triangle S'$ denotes the symmetric difference, $\nu_S=\det(O_{R,S})$, $\nu_{S'}=\det(O_{R',S'})$ and  $\nu_{S\triangle S'}=\det(O_{R\triangle R' ,S\triangle S'})$, with
\begin{equation}\label{eq:sym_fun}
f(S,S')=
\begin{cases}
1 &\mbox{if} \,\,\,|S\triangle S'|=|R\triangle R'| \\
0 &\mbox{otherwise}\,.
 \end{cases}
\end{equation}
\end{proposition}

To obtain a bound on the sample complexity, one can apply the median-of-means approach, which gives a simple classical post-processing strategy to reduce the effect of estimation errors \cite{huang20,mcnulty22}. 
The bound is controlled by the variance, which for the above strategy can be optimized via two approaches. First, the orthogonal matrix $O\in O(2n)$ appearing in the joint measurement $\G^O$ can be chosen to minimize $\nu_{S\triangle S'}/(\nu_S\nu_{S'})$. 
In a more general strategy, one can optimize over the classical post-processing, without restricting to a single $R$ for each $S$.

\section{A comparison to other estimation strategies}\label{sec:comparison}

\subsection{Joint measurements via fermion-to-qubit mappings}\label{sec:f-to-q}

Given that, for an $n$ mode fermionic system, every Majorana operator (of odd or even degree) can be expressed as an $n$--qubit Pauli operator, a possible measurement approach is to consider a joint measurement of Pauli observables \cite{busch86,brougham07}. As shown in \cite{mcnulty22}, we can measure any sufficiently unsharp $n$--qubit Pauli operator $P=P_1\otimes \ldots\otimes P_n$, with $P_i\in\{\id,X,Y,Z\}$, via a local parent POVM on each qubit, with effects
\begin{equation}\label{eq:jm_qubit}
\F(x_i,y_i,z_i)=\frac{1}{8}\left(\id+\frac{1}{\sqrt{3}}(x_iX+y_iY+z_iZ)\right)\,,
\end{equation}
where $x_i,y_i,z_i\in\{\pm 1\}$ label the $i$--th qubit outcome string. Furthermore, this measurement can be implemented by a randomization over projective measurements. To obtain the unsharp measurement,
\begin{equation}\label{eq:noisy_pauli}
\M_P(e_P)=\tfrac{1}{2}(\id+e_P\eta_PP)\,,
\end{equation}
with outcome $e_P\in\{\pm 1\}$, we perform the classical post-processing,
\begin{equation}\label{eq:pp_qubit}
D(e_P|P,x_1,\ldots,z_n)=\begin{cases}
1 &\mbox{if} \,\,\,e_P=\mu(P) \\
0 &\mbox{if} \,\,\,e_P=-\mu(P) \,,
 \end{cases}
\end{equation}
where $\mu(P):=\prod_{i\in\text{supp}(P)}\mu_i(P)$ is the product of the relevant local outcomes, with $\mu_i(P)=x_i,y_i,z_i$ if $P_i=X,Y,Z$, respectively, and $\text{supp}(P)=\{i\,|\,P_i\neq\id\}$. After post-processing of the joint measurement outcomes it is easy to verify that the noise in Eq. (\ref{eq:noisy_pauli}) satisfies
\begin{equation}\label{eq:pauli_noise}
\eta_P=3^{-w(P)/2}\,,
\end{equation}
where $w(P)=|\text{supp}(P)|$ is the weight of the Pauli operator.

The choice of fermion-to-qubit mapping plays a crucial role in determining how sharply the observables can be simultaneously measured. A parent measurement implementing sharper Majorana observables requires a mapping that minimizes the weight of the Pauli operator (in Eq. (\ref{eq:pauli_noise})) associated with each $\gamma_j$. The Jordan-Wigner transformation, as described in Eqs. (\ref{eq:jw1}) and (\ref{eq:jw2}), acts non-trivially on $\mathcal{O}(n)$ qubits. Consequently, degree--$k$ observables are jointly measured with sharpness exponential in $n$.

In fact, the measurement strategy provided in \cite{Jiang20} can be interpreted as the joint measurement of $\F$ on each qubit, together with a fermion-to-qubit transformation defined on ternary trees \cite{vlasov19}. The ternary trees mapping aims to minimize the average weight of the Pauli operators and maps each $\gamma_j$ to an operator $P$ with $w(P)=\lceil\log_3(2n+1)\rceil$. Thus, any degree--$k$ Majorana observable is transformed to a Pauli operator of weight at most $ k\lceil\log_3(2n+1)\rceil$. It follows that all degree--$k$ Majorana observables can be jointly measured with
\begin{equation}
\eta_k\lesssim 3^{-\frac{k}{2}\log_3(2n+1)}=(2n+1)^{-\frac{k}{2}}\,,
\end{equation}
where $\eta_k$ is defined in Eq. (\ref{eq:min_sharpness}).

Approximating the $k$--body reduced density matrices (RDMs) of a fermionic quantum state requires estimating the expectation values of all degree--$2k$ Majorana observables (cf. Sec. \ref{sec:est}). This is achieved, to the desired $\epsilon$ precision, by implementing $\mathcal{O}(\eta_{2k}^{-2}\epsilon^{-2})$ repetitions of the joint measurement. Since $\eta_{2k}=\mathcal{O}(n^{-k})$ via ternary tree mappings, $\mathcal{O}(n^{2k}\epsilon^{-2})$ repetitions of the joint measurement are required. For example, $\mathcal{O}(n^{4}\epsilon^{-2})$ repetitions are required to estimate 2-RDMs. This contradicts the erroneous claim in \cite{Jiang20} that a joint measurement with the ternary tree mapping requires only $\mathcal{O}(n^{k}\epsilon^{-2})$ repetitions to estimate $k$--RDMs.

\subsection{Fermionic classical shadows}\label{sec:cs}

Classical shadows---a randomized measurement protocol that yields a classical representation of an unknown quantum state \cite{huang20}---provides a method to simultaneously estimate non-commuting observables. Originally applied to systems of qubits, with unitaries sampled from the local or global $n$--qubit Clifford group, the protocol has since been extended in several ways to fermionic systems \cite{Zhao21,wan22,low22,ogorman22}. In one approach, a fermionic shadow is constructed by sampling unitaries from the intersection of the fermionic Gaussian and Clifford unitaries \cite{Zhao21}.

Several connections have been made between joint measurements and classical shadows \cite{mcnulty22}. In particular, for estimating energies of Hamiltonians composed of linear combinations of $n$--qubit Pauli observables, a single shot estimator based on the outcomes of a joint measurement provides the same performance guarantees as locally biased classical shadows \cite{hadfield20}. In light of Prop. \ref{prop:sample_complexity}, estimating degree--$k$ Majorana observables via the joint measurement defined in Eq. (\ref{eq:jm_gen}) shares the same asymptotic bounds on the sample complexity as fermionic shadows \cite{Zhao21,wan22}. However, in contrast to classical shadows which sample from the intersection of fermionic Gaussian and Clifford unitaries, implementing the joint measurement requires sampling from approximately $4k$ fermionic Gaussian unitaries and the set of Majorana monomials (or Pauli operators).

Other connections also appear between the two measurement strategies. For example, classical shadows can be used to define joint measurements and, vice versa, the outcomes of a joint measurement, together with classical post-processing, result in a classical shadow of the quantum state \cite{mcnulty22}. The former relation provides a sufficient condition for the compatibility of an arbitrary set of measurements. In Appendix \ref{A:CS}, we apply the fermionic shadow of \cite{Zhao21} to derive a joint measurability condition for the set of degree--$2k$ Majorana observables. In particular, we show that the observables $\{\gamma_S\,|\,S\in\mathcal{S}_{2k}\}$ are compatible when $\eta_{2k}\leq \binom{n}{k}/\binom{2n}{2k}$, therefore placing a lower bound on the incompatibility robustness $\eta_{2k}^*=\Omega(n^{-k})$. This coincides with the bound provided by the joint measurement in Sec. \ref{sec:f-to-q} using a ternary tree fermion-to-qubit mapping. Clearly, this provides only a sufficient condition for joint measurability since the parent POVM in Theorem \ref{thm:jm} together with the upper bounds from Theorem \ref{thm:inc_rob}, imply $\eta_{2k}^*=\Theta(n^{-k/2})$ for $k\leq 5$.

The fermionic shadow norm $\no{\gamma_S}_{\text{FGU}}$, derived in \cite{Zhao21}, is a useful quantity that provides a (state independent) bound on the variance of the classical shadow estimators of $\tr{\gamma_S\rho}$, for all $S\in\mathcal{S}_{2k}$. In particular, the squared shadow norm is given by $\no{\gamma_S}_{\text{FGU}}^2=\binom{2n}{2k}/\binom{n}{k}$. It follows from Theorem \ref{thm:inc_rob} that $\eta_{2k}^*\leq \no{\gamma_S}^{-1}_{\text{FGU}}$ for $k\leq 5$. If $k=1$, and the number of fermionic modes $n$ agrees with the existence of a $2n \times 2n$ skew-Hadamard matrix, then $\eta_{2}^*= \no{\gamma_S}^{-1}_{\text{FGU}}$. In the following conjecture, we suggest this bound holds more generally.

\begin{conjecture}\label{conjecture}
Let $\{\gamma_S\,|\,S\in \mathcal{S}_{2k}\}$ denote the set of degree--$2k$ Majorana operators of an $n$ mode fermionic system, with $k=\mathcal{O}(1)$. The incompatibility robustness satisfies $\eta_{2k}^*\leq\sqrt{\binom{n}{k}/\binom{2n}{2k}}=\no{\gamma_S}^{-1}_{\text{FGU}}$. Thus, the joint measurement in Sec. \ref{sec:nearly_optimal} is asymptotically optimal, and $\eta_{2k}^*=\Theta(n^{-k/2})$.
\end{conjecture}

Interestingly, due to Lemma \ref{lem:robustness}, the conjecture places a bound on the maximum eigenvalue of $H_{2k}$, i.e., $\no{H_{2k}}_{\infty}\leq\sqrt{\binom{2n}{2k}\binom{n}{k}}$. Similar bounds have been considered when optimizing fermionic Hamiltonians, e.g., \cite{feng19,hastings21,herasymenko22}. For $k\leq 5$, the conjectured bound is confirmed in Theorem \ref{thm:inc_rob}. We also note that the joint measurement in Sec. \ref{sec:f-to-q} places a lower bound on the maximum eigenvalue, namely, $\max_{e_S\in\{\pm 1\}}\no{H_{k}}_{\infty}=\Omega(n^{k/2})$, for all even and odd $k=\mathcal{O}(1)$.

\

\section{Concluding remarks}\label{sec:conclusion}

We have studied the joint measurability of fermionic observables composed of products of anti-commuting Majorana operators. Our main results are summarised in Theorems \ref{thm:jm} and \ref{thm:inc_rob}. We have introduced a parent POVM which simultaneously measures, for $k=\mathcal{O}(1)$, all degree--$2k$ Majorana observables $\{\gamma_S\,|\,S\in\mathcal{S}_{2k}\}$ with sharpness $\eta_{2k}=\Theta(n^{-k/2})$. By characterizing the incompatibility properties of Majorana observables, we found that the joint measurement is asymptotically optimal for $k\leq 5$. Our measurement is well-suited for various estimation tasks and can be implemented via randomized projective measurements.

There are several open questions to explore. While we provide asymptotically tight bounds on the incompatibility robustness when $k\leq 5$, a general upper bound (Conjecture \ref{conjecture}) is left open. A full characterization of the incompatibility robustness, as well as an optimal joint measurement for every system size, also remains open. Furthermore, it would be interesting to explore an extension of the fermionic joint measurement strategy (Sec. \ref{sec:jm}) from even-degree Majorana operators to odd-degree observables. While the odd case is covered by jointly measuring local Pauli observables under a fermion-to-qubit mapping (Sec. \ref{sec:f-to-q}), the optimal joint measurement remains unknown (except when $k=1$).

From a practical perspective, it would be useful to compare our joint measurement strategy to other estimation methods, e.g., for approximating energies of fermionic Hamiltonians. One potential advantage lies in the ability to tailor the joint measurement to a particular Hamiltonian, and optimize the variance of the estimator to reduce the sample complexity. Similar techniques have been developed for classical shadows \cite{hadfield20} and joint measurements \cite{mcnulty22} on qubit systems. For this purpose, it would also be useful to explore feasible implementations of the protocol with simple circuits, low circuit depth, and minimal reconfigurations.

%In particular, for $k=\mathcal{O}(1)$, we described a parent POVM which jointly measures all degree--$k$ Majorana observables with sharpness $\eta_S=\Omega(n^{-k/4})$. Furthermore, we showed that the joint measurement is asymptotically optimal, at least for $k\leq 10$, by evaluating the minimal noise necessary to destroy incompatibility. Our method relied exploited the symmetry properties of the Majorana observables, described by the braid group, to connect the incompatibility robustness to the SYK model. Based on the spectral properties of the SYK model, we showed that incompatibility robustness satisfies $\eta^*=\Theta(n^{-k/4})$.

\section*{Acknowledgments} 
We thank Marcin Kotowski and Ingo Roth for insightful discussions, as well as Matthew Hastings and Ryan O'Donnell for clarifying aspects of their work. The authors acknowledge financial support from the TEAM-NET project cofinanced by the EU within the Smart Growth Operational Program (Contract No. POIR.04.04.00-00-17C1/18-00). D.M. acknowledges support from PNRR MUR Project No. PE0000023-NQSTI.

\bibliographystyle{quantum}
\bibliography{references_dan}

%\newpage
\onecolumngrid
\appendix

\section{Classical post-processing of the joint measurement}\label{A:pp}

In this section, we present some basic properties of the POVM $\G^O$ defined in Eq. (\ref{eq:jm_gen}) and prove Prop. \ref{prop:jm_pp} of Sec. \ref{subsec:degree_k}. Recall that for any $S\in\mathcal{S}_{2k}$ and $X\subseteq [2n]$, we denote $x_S=(-1)^{|X|\cdot|S|-|X\cap S|}$, and for any $R=\cup R_i\in\mathcal{D}_{2k}$ with $R_i\in\mathcal{D}_{2}$, we denote $q_R=\prod_{i=1}^k q_{R_i}\in\{\pm 1\}$.
\begin{proposition}\label{Aprop:jm}
The POVM $\G^O$ defined in Eq. (\ref{eq:jm_gen}) is equivalent to
\begin{equation}\label{eqA:jm_gen_equiv}
\G^O({\bf q},X)=\frac{1}{2^{3n}}\left(\id+\sum_{k=1}^n \sum_{R\in\mathcal{D}_{2k}}q_R\sum_{S\in\mathcal{S}_{2k}}x_S\det(O_{R,S})\gamma_{S}\right)\,,
\end{equation}
with outcomes ${\bf q}=(q_{12},q_{34},\ldots,q_{2n-1,2n})\in\{\pm 1\}^{n}$ and $X\in 2^{[2n]}$.
\end{proposition}
\begin{proof}
Let $R_i=\{2j_i-1,2j_i\}\in\mathcal{D}_2$ with $j_i<j_{i'}$ for $i<i'$. The POVM defined in Eq. (\ref{eq:jm_gen}) can be expressed as
\begin{align}
\G^O({\bf q},X)
&=(2^{2n})^{n-1}\!\!\prod_{R\in\mathcal{D}_2} \G^O_{R}(q_{R},X) \label{A:jm_gen}\\[2pt]
&=\frac{1}{2^{3n}}\!\!\prod_{R\in\mathcal{D}_2}\!\left(\id+q_R\gamma_{R}^{O,X}\right)\label{A:jm_gen1}\\[2pt]
&=\frac{1}{2^{3n}}\!\Bigl(
   \id
   + \sum_{R_1\in\mathcal{D}_2} q_{R_1}\gamma_{R_1}^{O,X}
   + \sum_{R_1,R_2\in\mathcal{D}_2} q_{R_1}q_{R_2}\gamma_{R_1}^{O,X}\gamma_{R_2}^{O,X} \notag\\[-2pt]
&\qquad\qquad\qquad\qquad\qquad\quad\,\,
   + \sum_{R_1,R_2,R_3\in\mathcal{D}_2} q_{R_1}q_{R_2}q_{R_3}\gamma_{R_1}^{O,X}\gamma_{R_2}^{O,X}\gamma_{R_3}^{O,X}
   + \cdots
   \Bigr)\notag\\[2pt]
&=\frac{1}{2^{3n}}\!\left(
   \id
   + \sum_{R\in\mathcal{D}_2} q_{R}\gamma_{R}^{O,X}
   + \sum_{R\in\mathcal{D}_4} q_{R}\gamma_{R}^{O,X}
   + \sum_{R\in\mathcal{D}_6} q_{R}\gamma_{R}^{O,X}
   + \cdots
   \right)\label{A:jm_gen2}\\[2pt]
&=\frac{1}{2^{3n}}\!\left(
   \id
   + \sum_{k=1}^n\sum_{R\in\mathcal{D}_{2k}} q_R\gamma_R^{O,X}
   \right)\label{A:jm_gen3}\\[2pt]
&=\frac{1}{2^{3n}}\!\left(
   \id
   + \sum_{k=1}^n \sum_{R\in\mathcal{D}_{2k}} q_R
     \sum_{S\in\mathcal{S}_{2k}} x_S\det(O_{R,S})\gamma_{S}
   \right).\label{A:jm_gen4}
\end{align}
with $\gamma_R^{O,X}$ defined in Eq. (\ref{eq:random_rotated_prod}). In Eq. (\ref{A:jm_gen2}) we have used the fact that, for $R=\cup_{i=1}^k R_i\in\mathcal{D}_{2k}$ and $R_i\in\mathcal{D}_2$,
\begin{equation}
\gamma_R^O=U_O^{\dagger}\gamma_RU_O=U_O^{\dagger}\left(\prod_{i=1}^k\gamma_{R_i}\right)U_O=\prod_{i=1}^k \left(U_O^{\dagger}\gamma_{R_i}U_O\right)=\prod_{i=1}^k\gamma_{R_i}^O\,,
\end{equation}
and similarly, $\gamma_R^{O,X}=\prod_{i=1}^k\gamma^{O,X}_{R_i}$.
\end{proof}
Next, we show that the POVM can be rewritten by modifying the outcome space by replacing $X\subseteq [2n]$ with a string ${\bf x}=(x_1,\ldots,x_{2n})\in\{\pm 1\}^{2n}$.
\begin{proposition}\label{Aprop:jm_equiv}
The POVM $\G^O$ defined in Eq. (\ref{eqA:jm_gen_equiv}) can be expressed as
\begin{equation}\label{eqA:equiv_jm}
G^O({\bf q},{\bf x})=\frac{1}{2^{3n}}\left(\id+\sum_{k=1}^n \sum_{R\in\mathcal{D}_{2k}}q_R\sum_{S\in\mathcal{S}_{2k}}\left(\prod_{i\in S}x_i\right)\det(O_{R,S})\gamma_{S}\right)\,,
\end{equation}
with outcomes ${\bf q}=(q_{12},q_{34},\ldots,q_{2n-1,2n})\in\{\pm 1\}^{n}$ and ${\bf x}=(x_1,x_2\ldots,x_{2n})\in\{\pm 1\}^{2n}$.
\end{proposition}

\begin{proof}
Recall (cf. Eq. (\ref{eq:commutativity})),
\begin{equation}\label{Aeq:majorana_conjugation}
\gamma^{\dagger}_X\cdot\gamma_{j}\cdot\gamma_X=x_j\gamma_j\,,\qquad j=1,\ldots,2n,
\end{equation}
and
\begin{equation}
\gamma^{\dagger}_X\cdot\gamma_{S}\cdot\gamma_X=x_S\gamma_S\,,\quad S\in2^{[2n]}\,,
\end{equation}
where $x_{j}=(-1)^{|X|-|X\cap \{j\}|}$ and $x_S=(-1)^{|S|\cdot |X|-|S\cap X|}$. Thus, we can treat $x_S=\left(\prod_{j\in S}x_j\right)$ in $\G^O$ as a product of the relevant integers from the string ${\bf x}=(x_1,\ldots,x_{2n})\in\{\pm 1\}^{2n}$.

It is straightforward to check that each $X\in 2^{[2n]}$ gives rise to a unique string ${\bf x}=(x_1,\ldots,x_{2n})\in\{\pm 1\}^{2n}$ via Eqs. (\ref{Aeq:majorana_conjugation}). In particular, if the $-1$ entries of ${\bf x}$ are labelled by the indices $L=\{\alpha_1,\ldots,\alpha_{|L|}\}\subseteq [2n]$ and $|L|$ is even, then Eqs. (\ref{Aeq:majorana_conjugation}) is satisfied if $X=L$. On the other hand, if $|L|$ is odd then $X=[2n]\setminus L$. It follows that we can replace the outcomes $X\in 2^{[2n]}$ with ${\bf x}\in\{\pm 1\}^{2n}$ in the joint measurement $\G^O$.
\end{proof}

\begin{proposition}\label{Aprop:pp_R}
For any $R=\cup_{i=1}^k R_i\in\mathcal{D}_{2k}$ and $k\in [n]$, the POVM $\G^O$ in Eq. (\ref{eqA:jm_gen_equiv}) is a joint measurement of 
\begin{equation}\label{Aeq:jm_R}
\G^O_R(q_R,{\bf x})=\frac{1}{2^{2n+1}}\left(\id+q_R\sum_{S\in\mathcal{S}_{2k}}\left(\prod_{i\in S}x_i\right)\det(O_{R,S})\gamma_{S}\right)\,,
\end{equation}
with $q_R=\prod_{i=1}^k q_{R_i}\in\{\pm 1\}$ and ${\bf x}\in \{\pm 1\}^{2n}$.
\end{proposition}

\begin{proof}
This follows by taking the marginals, $\sum_{\{q_T|T\in\mathcal{D}_{2k}\}\setminus \{q_R\}}\G^O({\bf q},{\bf x})=\G_R^O(q_R,{\bf x})$.
\end{proof}

\begin{proposition}\label{Aprop:pp_gen}
For every $S\in\mathcal{S}_{2k}$ and $R\in\mathcal{D}_{2k}$, the POVM $\G^O_R$ in Eq. (\ref{Aeq:jm_R}) with classical post-processing,
\begin{equation}\label{eqA:pp}
D(e_{S}|S,q_R,{\bf x})=\begin{cases}
1 &\mbox{if} \,\,\,\text{sgn}(\det(O_{R,S}))q_R\left(\prod_{i\in S}x_i\right)=e_S \\
0 &\mbox{otherwise} \,,
 \end{cases}
\end{equation}
jointly measures the degree--$2k$ measurement $\M_S(e_S)=\frac{1}{2}(\id+e_S\eta_{R,S}\gamma_S)$ with sharpness $\eta_{R,S}=|\det(O_{R,S})|$.
\end{proposition}

\begin{proof}
Let $\tau=\text{sgn}(\det(O_{R,S}))$ and $x_S=\prod_{i\in S}x_i$. To obtain the unsharp measurement $\M_S$,  with outcome $e_S\in\{\pm1\}$, we evaluate
\begin{eqnarray}\label{A:cpp}
\M_S(e_S)&=&\sum_{q_R\in\{\pm 1\}}\sum_{{\bf x}\in\{\pm 1\}^{2n}}D(e_S|S,q_R,{\bf x})\G^O_R(q_R,{\bf x})\\
&=&\frac{1}{2^{2n+1}}\sum_{\substack{q_R,\,{\bf x}\,:\\e_S=\tau x_Sq_R}}\left(\id+q_R\sum_{T\in\mathcal{S}_{2k}}\left(\prod_{i\in T}x_i\right)\det(O_{R,T})\gamma_{T}\right)\\
&=&\frac{1}{2}\left(\id+q_R\left(\prod_{i\in S}x_i\right)\det(O_{R,S})\gamma_{S}\right)\\
&=&\frac{1}{2}\left(\id+\frac{e_S}{\tau}\det(O_{R,S})\gamma_{S}\right)\\
&=&\frac{1}{2}\left(\id+e_S|\det(O_{R,S})|\gamma_{S}\right)\,.
\end{eqnarray}
\end{proof}

Following Props. \ref{Aprop:jm}, \ref{Aprop:jm_equiv}, \ref{Aprop:pp_R} and \ref{Aprop:pp_gen}, we arrive at Prop. \ref{prop:jm_pp} in Sec. \ref{subsec:degree_k}.

\section{Joint measurement of degree--$2$ Majoranas}\label{A:jm_quadratic}

\subsection{Flat orthogonal matrices}

For simplicity, the joint measurability scheme described in Sec. \ref{sec:jm_pair} is initially restricted to system size $2n=\ell(\ell+1)$, where $\ell\in\mathbb{Z}^+$. To achieve the lower bound in Theorem \ref{thm:jm}, the parent POVM relies on the existence of the following orthogonal matrix.

\begin{definition}
An $\ell\times\ell$ orthogonal matrix $F\in O(\ell)$ is a lower-flat matrix if its entries satisfy
\begin{equation}\label{eq:lower-flat}
\min_{i,j\in[\ell]}|f_{ij}|=\Omega(\ell^{-1/2})\,.
\end{equation}
\end{definition}

A lower-flat matrix is a less restrictive version of a flat orthogonal matrix \cite{jaming15}. A flat orthogonal matrix additionally requires $\max_{i,j\in[\ell]}|f_{ij}|=\mathcal{O}(\ell^{-1/2})$ such that all entries have absolute value close to $1/\sqrt{\ell}$. Clearly, any Hadamard matrix is a (lower-)flat matrix. While it is conjectured that flat orthogonal matrices exist for all $\ell$, only specific constructions are known. For instance, if a Hadamard of order $\ell+1$ exists, a flat orthogonal matrix of order $\ell$ exists. Such examples occur when $\ell=p^m$, with $p=4k-1$ prime and $m$ odd (e.g. $\ell=3,7,11,19,\ldots$). 

Fortunately, one can construct a lower-flat orthogonal matrix for any $\ell\in\mathbb{Z}^+$ \cite{jaming15}. An example is given by the $N\times N$ matrix ($N=2^m+q$ and $q<2^m$),
\begin{equation}\label{eq:flat}
F=\left( \begin{array}{c|c|c}
           H_{s,s} & \frac{1}{\sqrt{2}}H_{s,q} & -\frac{1}{\sqrt{2}}H_{s,q} \\ \hline
           \frac{1}{\sqrt{2}}H_{q,s} & \frac{1}{2}(H_{q,q}+\id) &  \frac{1}{2}(-H_{q,q}+\id)   \\ \hline
           -\frac{1}{\sqrt{2}}H_{q,s} & \frac{1}{2}(-H_{q,q}+\id)  & \frac{1}{2}(H_{q,q}+\id) \\
         \end{array} \right)\,,
\end{equation}
where $H$ is any Hadamard matrix of order $2^m$, which we assign the block structure,
\begin{equation}
H=\left( \begin{array}{c|c}
           H_{s,s} & H_{s,q} \\ \hline
           H_{q,s} & H_{q,q}   \\
         \end{array} \right)\,.
\end{equation}
 Here, $H_{s,q}$ denotes an $s\times q$ block matrix with $s=2^m-q$. Orthogonality of the lower-flat matrix follows from the decomposition $F=V^T\widetilde F V$, where 
\begin{equation}
    \widetilde F=\left( \begin{array}{c|c}
           H & 0 \\ \hline
           0 & I_q  
         \end{array} \right)\, \qquad \text{and} \qquad 
         V=\left( \begin{array}{c|c|c}
           I_s & 0 & 0 \\ \hline
           0 & \frac{1}{\sqrt{2}}I_q &  -\frac{1}{\sqrt{2}}I_q   \\ \hline
           0 & \frac{1}{\sqrt{2}}I_q  & \frac{1}{\sqrt{2}}I_q
         \end{array} \right)\,. 
\end{equation}
Note that for $m>1$, the entries of $F$ in Eq. (\ref{eq:flat}) satisfy $\frac{1}{2\sqrt{2^m}}\leq|f_{ij}|\leq \frac{1}{2}+\frac{1}{\sqrt{2^m}}$, and therefore $|f_{ij}|\geq\frac{1}{2\sqrt{N}}$.

\subsection{Fermionic Gaussian unitaries and perfect matchings}\label{A:perfect_matchings}

As described in Sec. \ref{sec:jm_pair}, to jointly measure all noisy degree--$2$ Majorana observables, we randomised over two POVMs $\G^{(1)}$ and $\G^{(2)}$, as defined in Eq. (\ref{eq:jm_gen}), which are implemented by the orthogonal matrices $O^{(1)}=P_{\pi} \cdot D$ and $O^{(2)}=P_{\pi} \cdot D\cdot P_{\sigma}$, respectively. The construction of each matrix, and the proof of Theorem \ref{thm:jm}, is most easily understood by considering perfect matchings in Turán graphs, where the vertices correspond to degree--1 Majorana operators. We initially consider a Turán graph $T(2n,\ell+1)$ with its $2n=\ell(\ell+1)$ vertices partitioned into $\ell+1$ subsets, $Y_{1}=\{1,\ldots, \ell\},\ldots, Y_{\ell+1}=\{\ell^2+1,\ldots, \ell(\ell+1)\}$, such that $|Y_i|=\ell$ for all $i\in[\ell+1]$. It will be useful to assign each subset of vertices a distinct color (e.g., see Fig. \ref{fig:perfect_matching}, left).

\begin{figure}[!t]\label{fig:permutation_main}
  \centering
\begin{minipage}{0.49\columnwidth}
%\begin{flushleft}
 \centering
\includegraphics[scale=0.15]{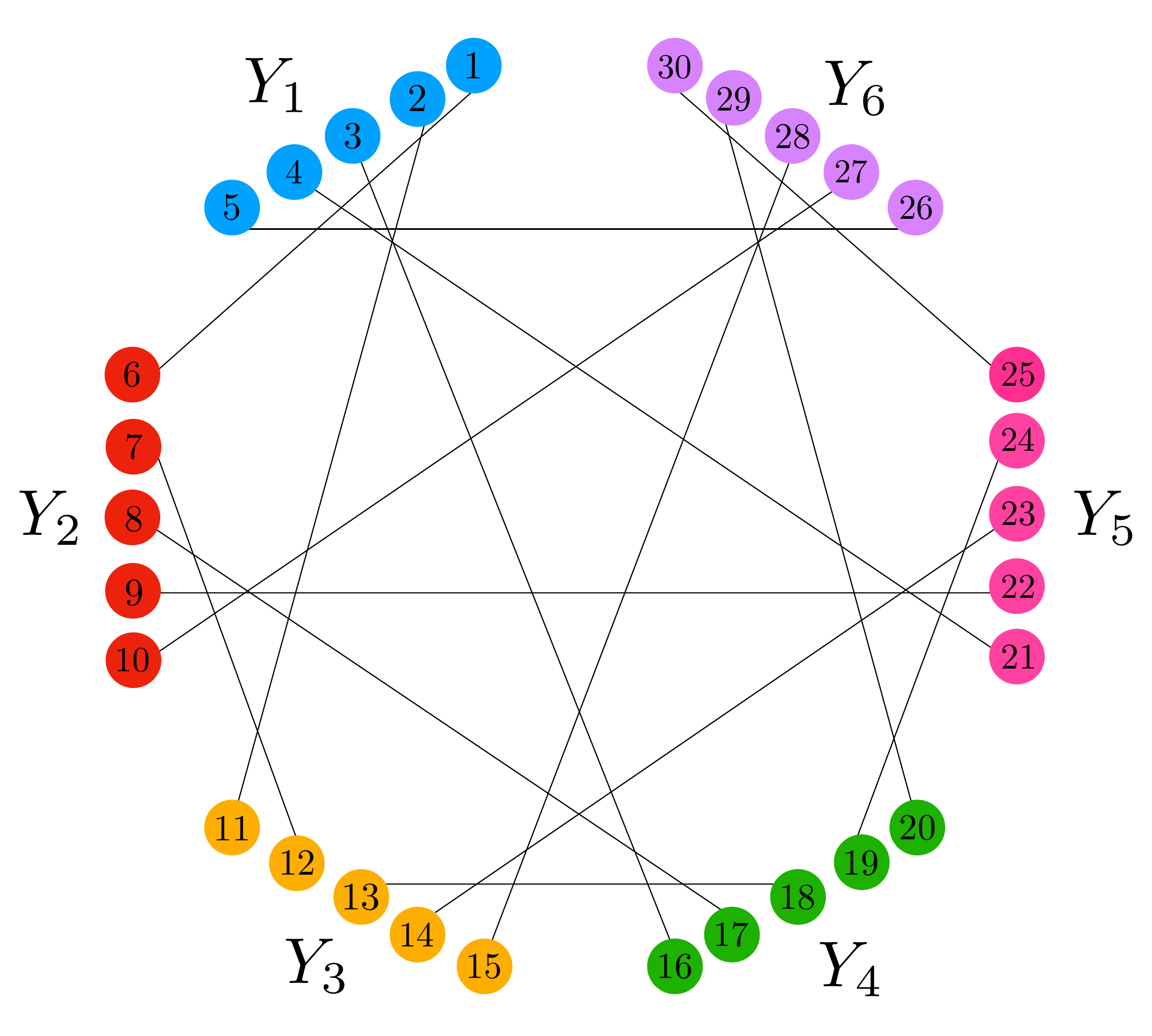}
%\end{flushleft}
\end{minipage}
\begin{minipage}{0.49\columnwidth}
%\begin{flushleft}
\centering
\includegraphics[scale=0.15]{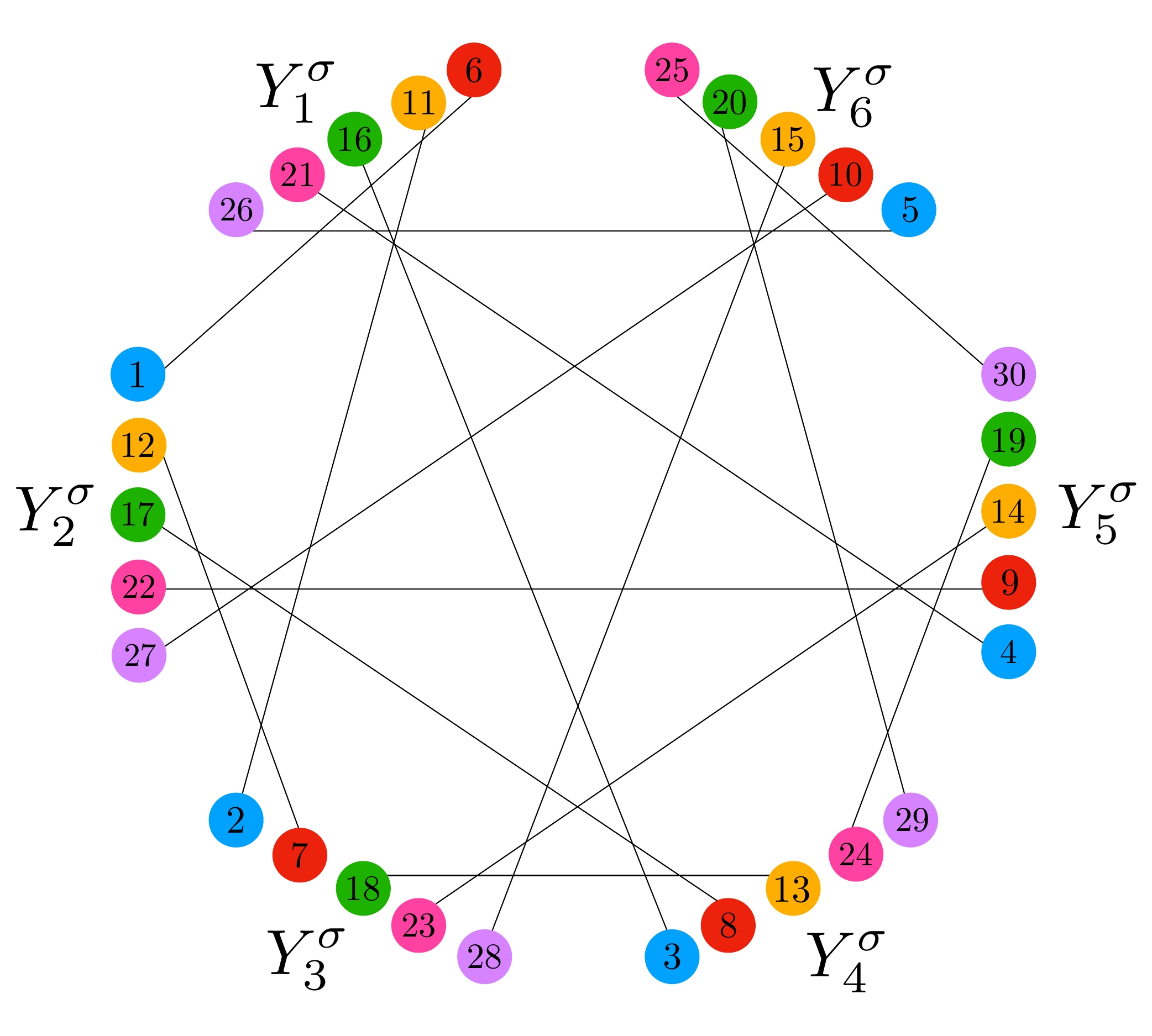}
%\end{flushleft}
\end{minipage}
\caption{\label{fig:perm_A}(Left.) An example, for $n=15$, of a perfect matching $\mathcal{M}$ of the Turán graph $T(2n,\ell+1)$, with $\ell=5$. The vertices (representing single Majorana operators) are partitioned into $6$ disjoint subsets $Y_i$ of cardinality $5$. The six partition subsets are assigned distinct colors and the sparsely arranged perfect matching (see Def. \ref{def:sparse}) ensures the existence of an edge between every two distinct colors. (Right.) A Turán graph $T(30,6)$ obtained by permuting vertices connected by an edge in the original Turán graph, i.e., $v_1\leftrightarrow v_2$ if $\{v_1,v_2\}\in\mathcal{M}$. This ensures each partition subset $Y_i^{\sigma}$ contains no identically colored vertices.}
\end{figure}
The perfect matching we construct is required to satisfy the additional property of being sparsely arranged, which we define as follows.
\begin{definition}\label{def:sparse}
    A perfect matching $\mathcal{M}$ of a Turán graph is \emph{sparsely arranged} if, for every pair of distinct partition subsets $Y_i$ and $Y_{i'}$, there exists an edge $\{v_1,v_2\}\in \mathcal{M}$, with $v_1\in Y_i$ and $v_2\in Y_{i'}$.
\end{definition}
 A sparsely arranged perfect matching for $n=15$ is illustrated in Fig. \ref{fig:perm_A}, and generalized to $2n=\ell(\ell+1)$ vertices in Fig. \ref{fig:perm_B}.

Given a perfect matching $\mathcal{M}$ of a Turán graph $T(2n,\ell+1)$, we define a joint measurement,
\begin{equation}\label{Aeq:jm_matching}
\widetilde\G^O=(2^{2n})^{n-1}\prod_{M\in\mathcal{M}}\G^O_{M}\,,
\end{equation}
where
\begin{equation}\label{eq:povm_edge}
\G^{O}_M(q_M,X)=\frac{1}{2^{2n}}\left(\frac{\id+q_M \gamma^{O,X}_M}{2}\right)\,,
\end{equation}
and $\gamma^{O,X}_M:=\sum_Sx_S\cdot\det(O_{M,S})\cdot\gamma_S$, in analogy with Eqs. (\ref{eq:random_rotated_prod}) and (\ref{eq:jm}). Now, applying the matrix $O$ is equivalent to performing $P_{\pi}\cdot O$ in the original setting of Eq (\ref{eq:jm_gen}), i.e., $\G^{P_{\pi}\cdot O}=\widetilde\G^O$, where $\pi\in S_{2n}$ defines the perfect matching of Eq. (\ref{eq:matching_perm}) and maps elements of $\mathcal{M}$ to $\mathcal{D}_2$.

\begin{proposition}\label{prop:edge}
Suppose $\G^O_M$ is the POVM in Eq. (\ref{eq:povm_edge}) generated from the edge $M=\{v_1,v_2\}\in \mathcal{M}$ of the Turán graph $T(2n,\ell+1)$ connecting $Y_i$ and $Y_{i'}$. Let $O=D=F^{(1)}\oplus\cdots\oplus F^{(\ell+1)}$ be the block diagonal matrix of lower-flat matrices $F^{(i)}\in O(\ell)$, as defined in Eq. (\ref{eq:flat}). Then, $\G^O_M$ jointly measures the set of observables $\{i\gamma_{j_1}\gamma_{j_2} \,|\, j_1\in Y_i,j_2\in Y_{i'}\}$ with sharpness $\eta_{j_1,j_2}=|\det(D_{M,\{j_1,j_2\}})|=\Omega(\ell^{-1})$.
\end{proposition}

This can be easily checked by performing the same classical post-processing calculation described in Appendix \ref{A:pp}. Note that $\eta_{j_1,j_2}=|\det(D_{M,\{j_1,j_2\}})|=|d_{v_1j_1}d_{v_2j_2}|$, where $d_{\mu_1j_1}$ and $d_{v_2j_2}$ are matrix elements of $D$ that come from lower-flat orthogonal matrices $F^{(i)}$ and $F^{(i')}$, respectively. Thus, the lower bound on the sharpness follows directly from Eq. (\ref{eq:lower-flat}).

If $\mathcal{M}$ is sparsely arranged, it follows from Prop. \ref{prop:edge} that $\widetilde\G^{D}$, as defined in Eq. (\ref{Aeq:jm_matching}), jointly measures the set of observables $\{\gamma_S\,|\, S\,\, \text{is an edge of} \,\,T(2n,\ell+1)\}$ with $\eta_S=\Omega(\ell^{-1})$, i.e., all quadratic monomials whose constituent vertices are distinctly colored.

To jointly measure observables composed of operators from the same partition subset (i.e., of the same color), we define a second POVM based on a new partition of vertices of the Turán graph. This partition is described by a permutation $\sigma\in S_{2n}$ which permutes elements from the original partition subsets $Y_1,\ldots,Y_{\ell+1}$ to $Y_1^{\sigma},\ldots,Y_{\ell+1}^{\sigma}$, where $\sigma$ exchanges vertices $v_1\leftrightarrow v_2$ if and only if  $\{v_1,v_2\}\in\mathcal{M}$. This guarantees that each partition subset $Y_i^{\sigma}$ contains no identically colored vertices (see, e.g., Fig. \ref{fig:perm_A}, right). In this case, the POVM is defined by the matrix $D\cdot P_{\sigma}$, where $P_{\sigma}$ permutes the columns of $D$ via $\sigma\in S_{2n}$. From the POVM $\widetilde\G^{D\cdot P_{\sigma}}$, it follows via Prop. \ref{prop:edge} that we can jointly measure (with sufficient sharpness) all quadratic observables composed of Majorana operators of the same color. In conclusion, after randomizing over both $\G^{D}$ and $\G^{D\cdot P_{\sigma}}$, the full set of observables is measured with sharpness satisfying the required asymptotic scaling.

\begin{figure}[!t] \label{fig:permutation_main2}
  \centering
\begin{minipage}{0.49\columnwidth}
 \centering
%\begin{flushleft}
\includegraphics[scale=0.25]{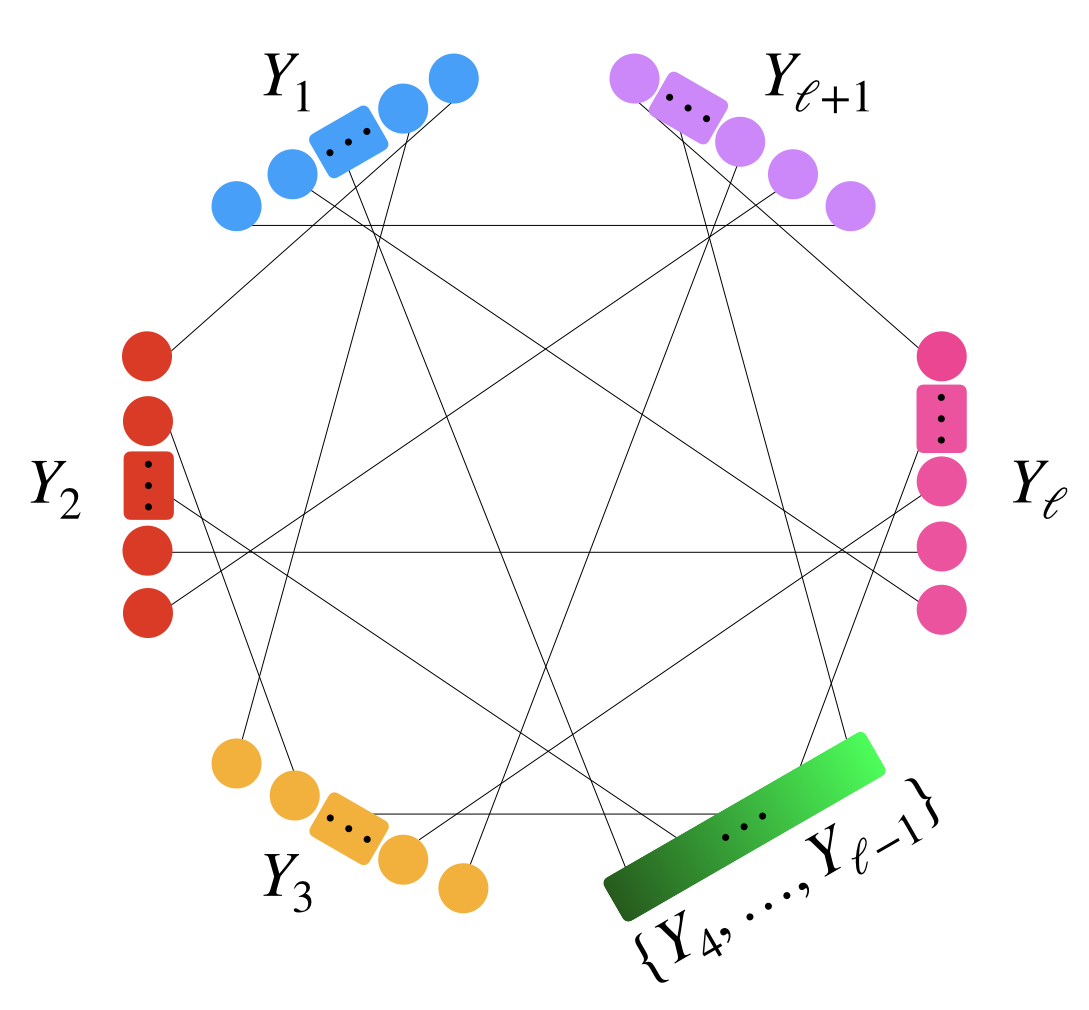}
%\end{flushleft}
\end{minipage}
\caption{\label{fig:perm_B}
A perfect matching of the graph $T(2n,\ell+1)$, with its $2n$ vertices partitioned into disjoint subsets $Y_i$, $i=1,\ldots, \ell+1$, each of cardinality $\ell$. First, each vertex in $Y_1$ is connected to a vertex of a different subset via the edges $(1,\ell+1), (2,2\ell+1),\ldots,(\ell,\ell^2+1)$, ensuring that $Y_1$ is connected to every other partition subset. This strategy repeats for the remaining subsets, i.e., vertices in $Y_2$ are paired to vertices of different subsets by the edges $(\ell+2,2\ell+2),(\ell+3,3\ell+2),\ldots,(2\ell,\ell^2+2)$. This continues to the penultimate subset $Y_{\ell}$ that has only one vertex left unpaired, and is joined to the remaining vertex of $Y_{\ell+1}$ by the edge $(\ell^2,\ell(\ell+1))$. We note that the perfect matching is used to define the permutation $\pi$ in Eq. (\ref{A:perm}), required to construct $O^{(1)}\in O(2n)$ (cf. Eq. (\ref{eq:matrixQ})). The permutation is defined such that each edge is mapped to a unique element in $\mathcal{D}_2$, e.g., $(1,\ell+1)\rightarrow (1,2)$.}
\end{figure}

In the circuit representation of the measurement (see Fig. \ref{fig:general_idea} and Eq. (\ref{eq:jm_gen})), after performing a fermionic Gaussian unitary, we simultaneously measure the $n$ commuting observables $\{\gamma_R\,|\,R\in \mathcal{D}_2\}$ as opposed to $\{\gamma_M\,|\,M\in \mathcal{M}\}$. This difference is incorporated by the permutation defined by the sparsely arranged perfect matching $\mathcal{M}$ of Fig. \ref{fig:perm_B} (cf. Eq. (\ref{eq:matching_perm})), namely,
\begin{equation}\label{A:perm}
\pi=\left( \begin{array}{ccccccc}
                      1 & \ell+1 & 2 & 2\ell+1 & \ldots & \ell^2-1  & 2n\\
                      1 & 2 & 3 & 4 & \ldots & 2n-1 & 2n
         \end{array} \right)\,.
\end{equation}
In particular, the measurements $\widetilde\G^D$ and $\widetilde\G^{D\cdot P_{\sigma}}$ defined in Eq. (\ref{Aeq:jm_matching}) are equivalent to the measurements $\G^{P_{\pi}\cdot D}$ and $\G^{P_{\pi}\cdot D\cdot P_{\sigma}}$ defined in Eq. (\ref{eq:jm_gen}), respectively. Written explicitly, the permutation $\sigma$, which exchanges vertices $v_1\leftrightarrow v_2$ if and only if  $\{v_1,v_2\}\in\mathcal{M}$ (e.g., Fig. \ref{fig:perm_A}, right) is defined, in cyclic form, as
\begin{equation}\label{eq:sigma2}
\sigma=(1,\ell+1)(2,2\ell+1)\ldots(\ell^2-1,\ell(\ell+1))\,.
\end{equation}

\subsection{Example: \texorpdfstring{$n=3$}{n=3}}
We consider the case of $n=3$ with $\ell=2$. To jointly measure all quadratic Majorana observables we require two POVMs of the form defined in Eq. (\ref{eq:jm_gen}),
\begin{equation*}
\G^{(r)}({\bf q},X)=4^6\G^{(r)}_{12}(q_{12},X)\cdot \G^{(r)}_{34}(q_{34},X)\cdot \G^{(r)}_{56}(q_{56},X)\,,
\end{equation*}
where, e.g.,
\begin{equation*}
\G^{(r)}_{12}(q_{12},X)=\frac{1}{2^6}\left(\frac{\id+q_{12}\sum_{S\in\mathcal{S}_2} x_{S}\cdot\det(O^{(r)}_{12,S})\cdot\gamma_{S}}{2}\right)\,,
\end{equation*}
and $r=1,2$. To construct the orthogonal matrices $O^{(r)}$, let $D=F\oplus F\oplus F$, where
\begin{equation}
F=\frac{1}{\sqrt{2}}\left( \begin{array}{cc}
           1 & 1 \\
           1 & -1
         \end{array} \right)\,
\end{equation}
is a Hadamard matrix (also a lower-flat matrix), and $P_{\pi}$, with $\pi=(1)(23)(45)(6)$, is the matrix which permutes rows $2\leftrightarrow3$ and $4\leftrightarrow5$ of $D$, as described by the perfect matching of the Turán graph $T(6,3)$ in Fig. \ref{fig:perm_C}.

The first POVM $\G^{(1)}$ is therefore specified by
\begin{equation}
O^{(1)}=P_{\pi}\cdot D=\frac{1}{\sqrt{2}}\left( \begin{array}{cccccc}
           1 & 1 & 0 & 0 & 0 & 0 \\
           0 & 0 & 1 & -1 & 0 & 0 \\
           1 & -1 & 0 & 0 & 0 & 0 \\
           0 & 0 & 0 & 0 & 1 & 1 \\
           0 & 0 & 1 & 1 & 0 & 0 \\
           0 & 0 & 0 & 0 & 1 & -1
         \end{array} \right)\,.
\end{equation}
Applying the post-processing of Eq. (\ref{eq:pp_gen}) to $\G^{(1)}_{12}$ we recover noisy versions of $\gamma_S$ only when $\det(O^{(1)}_{12,S})\neq0$. This holds when $S\in\{\{1,3\}, \{1,4\},\{2,3\},\{2,4\}\}$, in which case $|\det(O^{(1)}_{12,S})|=1/2$. Likewise, post-processing of $\G^{(1)}_{34}$ and $\G^{(1)}_{56}$ yields the noisy observables associated with $S\in\{\{1,5\}, \{1,6\},\{2,5\},\{2,6\}\}$ and $S\in\{\{3,5\}, \{3,6\},\{4,5\},\{4,6\}\}$, respectively, with $\eta_{S}=1/2$, as defined in Eq. (\ref{eq:optimal_eta_S}). These cover all edges of the Turán graph $T(6,3)$ in Fig. \ref{fig:perm_C} (left).

\begin{figure}[t]\label{fig:figure_n_3}
 \centering
\begin{minipage}{0.49\columnwidth}
%\begin{flushleft}
 \centering
\includegraphics[scale=0.10]{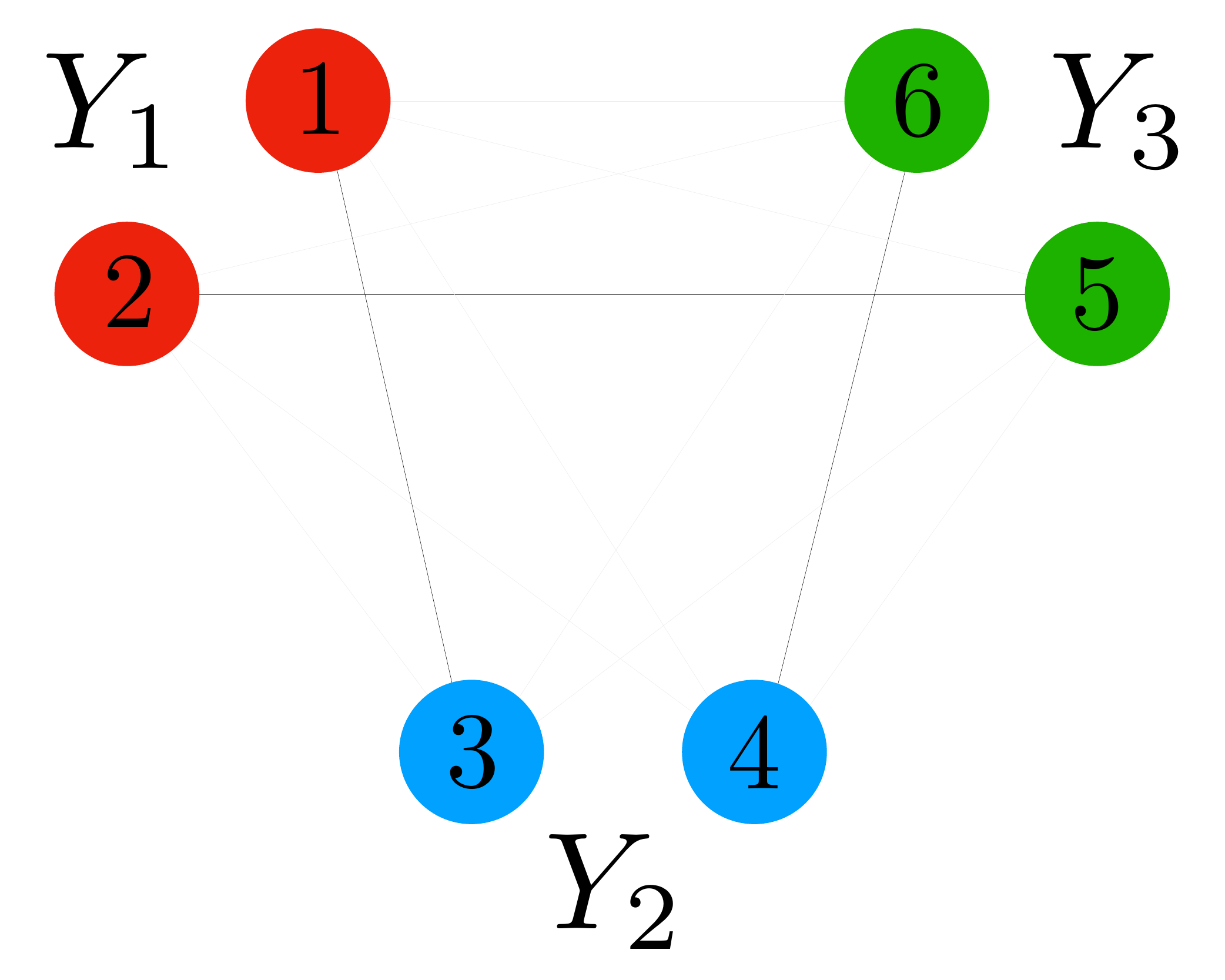}
%\end{flushleft}
\end{minipage}
\begin{minipage}{0.49\columnwidth}
%\begin{flushleft}
 \centering
\includegraphics[scale=0.10]{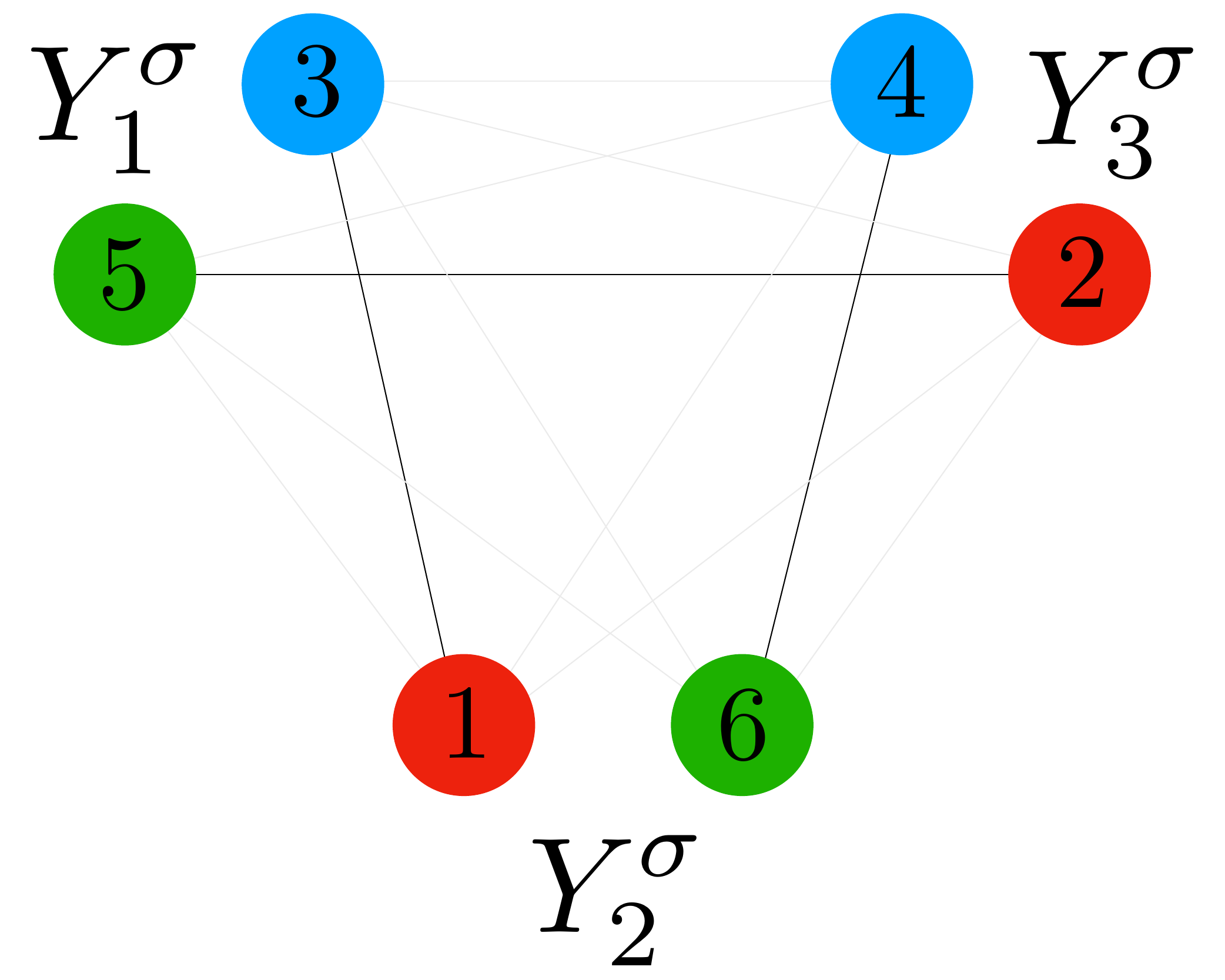}
%\end{flushleft}
\end{minipage}
\caption{\label{fig:perm_C}A perfect matching of the Turán graph $T(6,3)$, used to define the permutation matrices $P_{\pi}$ and $P_{\sigma}$. The three edges $(1,3)$, $(2,5)$ and $(4,6)$ of the perfect matching (left) are permuted to $(1,2)$, $(3,4)$ and $(5,6)$, respectively, such that $\pi=(1)(23)(45)(6)$. The second Turán graph (right) is obtained by swapping vertices connected by an edge in the first perfect matching, i.e., $\sigma=(13)(25)(46)$.}
\end{figure}

The missing observables, i.e., $S\in\{\{1,2\},\{3,4\},\{5,6\}\}$, are recovered from the second POVM $\G^{(2)}$ with
\begin{equation}
O^{(2)}=P_{\pi}\cdot D\cdot P_{\sigma}\,,
\end{equation}
where $\sigma=(13)(25)(46)$ such that $P_{\sigma}$ permutes the columns $1\leftrightarrow3$, $2\leftrightarrow5$ and $4\leftrightarrow6$ of $D$ (see Fig. \ref{fig:perm_C}, right). Explicitly, we implement
\begin{equation}
O^{(2)}=\frac{1}{\sqrt{2}}\left( \begin{array}{cccccc}
           0 & 0 & 1 & 0 & 1 & 0 \\           
           1 & 0 & 0 & 0 & 0 & 1 \\
           0 & 0 & 1 & 0 & -1 & 0 \\           
           0 & 1 & 0 & 1 & 0 & 0 \\
           1 & 0 & 0 & 0 & 0 & -1 \\
           0 & 1 & 0 & -1 & 0 & 0
         \end{array} \right)\,,
\end{equation}
and perform post-processing of each POVM $\G^{(2)}_{12}$, $\G^{(2)}_{34}$ and $\G^{(2)}_{56}$. From $\G^{(2)}_{12}$ we obtain the four noisy quadratic observables $S\in\{\{1,3\}, \{1,5\},\{3,6\},\{5,6\}\}$. Similarly, for $\G^{(2)}_{34}$ and $\G^{(2)}_{56}$ we obtain the observables $S\in\{\{2,3\}, \{2,5\},\{3,4\},\{4,5\}\}$, and $S\in\{\{1,2\}, \{1,4\},\{2,6\},\{4,6\}\}$, respectively. In each case, $\eta_{S}=1/2$.

\section{Joint measurement for arbitrary $n$ modes}\label{A:arbitrary_n}

To prove that Theorem \ref{thm:jm} (i) applies for an arbitrary $n$ mode fermionic system, we now generalize the construction of $O^{(1)}$ and $O^{(2)}$ defined in Eq. (\ref{eq:matrixQ}) and described in Appendix \ref{A:jm_quadratic}. Suppose $2n=\ell(\ell+1)+t$, where $0\leq t<2(\ell+1)$. We adapt the joint measurement by modifying the subsets of vertices on which the diagonal block matrix acts. In particular, consider the Turán graph $T(2n, \ell+1)$, with the $2n$ vertices partitioned into $\ell+1$ subsets $Y_i$ of cardinality $\ell+\ell_i$, where $0\leq \ell_i\leq 2$ for each $i\in [\ell+1]$. The partition is chosen such that $Y_1=\{1,2,\ldots, \ell+\ell_1\}$, $Y_2=\{\ell+\ell_1+1,\ldots,2\ell+\ell_1+\ell_2 \}$, etc., and the extra $t$ vertices are distributed as evenly as possible among the subsets. Thus, we define a new block diagonal matrix $D=\bigoplus_{i=1}^{\ell+1} F^{(i)}$, where $ F^{(i)}$ is a $(\ell+\ell_i)\times (\ell+\ell_i)$ lower-flat orthogonal matrix.

We consider two cases: (a) $0\leq t\leq (\ell+1)$, and (b) $(\ell+1)<t<2(\ell+1)$. First, suppose (a) holds such that each subset $Y_i$ contains at most $\ell+1$ vertices. The matrix $O^{(1)}=P_{\pi}\cdot D$ relies on a permutation $\pi$ of a perfect matching $\mathcal{M}$ of $T(2n, \ell+1)$ which is \emph{almost} sparsely arranged. $\mathcal{M}$ is partially constructed by ignoring the ``extra'' vertices and following the strategy described in Fig. \ref{fig:perm_B}. Edges can then be added in an arbitrary fashion to complete the perfect matching, leaving some of the partition subsets connected twice. The matrix $O^{(2)}=P_{\pi}\cdot D \cdot P_{\sigma}$ results from a permutation $\sigma$ (in analogy with Fig. \ref{fig:perm_A}, right) which exchanges vertices $v_1\leftrightarrow v_2$ iff $\{v_1,v_2\}\in \mathcal{M}$. This leads to a partition $Y_i^{\sigma}, i=1,\ldots,\ell+1$, in which no two vertices from the same set are mapped via $\sigma$ to a single subset $Y_{i}^{\sigma}$.

This is no longer possible if (b) holds since the cardinality of $Y_i$ can be larger than the number of partition subsets. Instead, whenever $|Y_i|=\ell+2$, we modify the perfect matching to ensure the two ``extra'' vertices $v_1,v_2\in Y_i$ are connected by an edge in $\mathcal{M}$, i.e., $M=\{v_1,v_2\}\in \mathcal{M}$. The corresponding POVM $\G^D_M$, defined in Eq. (\ref{eq:povm_edge}), is not guaranteed to measure all sufficiently unsharp pairs $i\gamma_{j_1}\gamma_{j_2}$, with $j_1,j_2\in Y_i$. In particular, the relevant submatrix $D_{M,\{j_1,j_2\}}$ (from rows $v_1$ and $v_2$) is not monomial, and the determinant could be zero (or too small). Nevertheless, due to the orthogonality of $F^{(i)}$, there exists at least one pair $\{j_1,j_2\}\subset Y_i$ such that $i\gamma_{j_1}\gamma_{j_2}$ is jointly measured by $\G^D_{M}$ with sufficient sharpness.  In particular, for orthogonality of $F$ to hold (dropping the index $i$ for simplicity), given a pair of rows $(v_1,v_2)$ there exists a pair of columns $(j_1,j_2)$ such that exactly three of the four matrix entries of
\begin{equation}
F_{\{v_1,v_2\},\{j_1,j_2\}}=\left( \begin{array}{cc}
           F_{v_1,j_1} & F_{v_1,j_2} \\
           F_{v_2,j_1} & F_{v_2,j_2}
         \end{array} \right)\,
\end{equation}
have the same sign. In this case, $|\det(F_{\{v_1,v_2\},\{j_1,j_2\}})|=|F_{v_1,j_1}F_{v_2,j_2}|+|F_{v_1,j_2}F_{v_2,j_1}|=\Omega(\ell^{-1})$, guaranteeing that $i\gamma_{j_1}\gamma_{j_2}$ is jointly measured with sufficient sharpness. Therefore, it is not necessary to ensure the vertices of the given pair $j_1$ and $j_2$ lie in disjoint subsets in the second partitioning. Instead, $\sigma$ is chosen such that the only vertices of $Y_i$ that remain together in a single partition subset are those connected by an edge in $\mathcal{M}$. Following a similar argument presented in Appendix \ref{A:jm_quadratic}, we conclude all degree--$2$ Majorana observables of an $n$ mode fermionic system are jointly measured with the desired sharpness.

\section{Joint measurement of degree--$2k$ Majoranas}\label{A:jm_quartic}

To simultaneously measure all degree--$2k$ Majorana observables (with the noise specified in Thm. \ref{thm:jm}) we randomize over a collection on parent POVMs $\G^{O}$ defined in Eq. (\ref{Aeq:jm_matching}). For simplicity, we work with POVMs $\widetilde\G^O$, described by Eq. (\ref{eq:jm_matching}) and a perfect matching $\mathcal{M}$ of a Turán graph $T(2n,\ell+1)$. We initially assume $2n=\ell(\ell+1)$, ensuring that the corresponding Turán graph has a vertex set that partitions into $\ell+1$ subsets, each of cardinality $\ell$. Furthermore, we assume $\mathcal{M}$ is sparsely arranged (cf. Sec. \ref{sec:jm_pair}).

Let $\mathcal{Y}$ denote the set of all partitions $Y=(Y_1,\ldots,Y_{\ell+1})\in\mathcal{Y}$ of the Turán graph $T(2n,\ell+1)$, i.e., all partitions of its vertex set $[2n]=\{1,2,\ldots,2n\}$, into $\ell+1$ subsets $Y_i\in \binom{[2n]}{\ell}$, such that $Y_i\cap Y_{i'}=\varnothing$. In the following definition, we associate a joint measurement $\widetilde\G^O$ with each partition of a Turán graph.

\begin{definition}
We say that the joint measurement $\widetilde\G^O$  (in Eq. (\ref{Aeq:jm_matching})) is \emph{defined by the partition} $Y=(Y_1,\ldots,Y_{\ell+1})\in \mathcal{Y}$ of a Turán graph $T(2n,\ell+1)$ if 
$O=D\cdot P_{\sigma}$, where $\sigma\in S_{2n}$ is the permutation associated with the partition $Y$.
\end{definition}

Here, $D=\bigoplus_{i=1}^{\ell+1}F^{(i)}$ is the usual block diagonal matrix of lower-flat matrices $F^{(i)}\in O(\ell)$, and $\sigma$ is the permutation which maps the partition $\{1,\ldots,\ell\},\ldots,\{\ell^2+1,\ldots,\ell(\ell+1)\}$ to the subsets $Y_1,\ldots Y_{\ell+1}$.

\begin{definition}
The set $S=\{j_1,j_2,\ldots,j_{2k}\}\in\mathcal{S}_{2k}$ is \emph{generated} from a partition $Y=(Y_1,\ldots,Y_{\ell+1})\in\mathcal{Y}$ if for every pair $\{j,j'\}\subseteq S$, then $\{j,j'\}\nsubseteq Y_1,\ldots,Y_{\ell+1}$, i.e., no two elements of $S$ are in the same partition subset.
\end{definition}

\begin{lemma}\label{lem:jm_2k}
Let $\widetilde\G^O$ be the joint measurement defined by the partition $Y\in\mathcal{Y}$ of a Turán graph $T(2n,\ell+1)$. If $S\in \mathcal{S}_{2k}$ is generated from $Y$, then $\widetilde\G^O$ jointly measures $\gamma_S$ with $\eta_S=\Omega(n^{-k/2})$
\end{lemma}

\begin{proof}
Let $\mathcal{M}$ be the sparsely arranged perfect matching (described in Fig. \ref{fig:perm_B} of Appendix \ref{A:jm_quadratic}) of the Turán graph $T(2n,\ell+1)$ and suppose $M_1,\ldots,M_{k}\in \mathcal{M}$ are the $k$ distinct edges connecting the $2k$ distinct partition subsets $Y_{i_1},Y_{i_2},\ldots,Y_{i_{2k}}$. For any $S=\{j_1,j_2,\ldots,j_{2k}\}\in\mathcal{S}_{2k}$, it follows from Prop. \ref{prop:jm_pp} that classical post-processing of the joint measurement yields the unsharp observable $\M_S(\pm)=\frac{1}{2}(\id\pm|\det(O_{M,S})|\gamma_S)$, where $M=\cup_{i=1}^k M_i$. If $S$ is generated from $Y$, the $2k\times 2k$ submatrix $O_{M,S}$ is monomial (i.e., each row and column has exactly one non-zero element). Therefore, the determinant is the product of $2k$ entries from a set of $\ell\times \ell$ lower-flat matrices. Since each entry satisfies Eq. (\ref{eq:lower-flat}), we have $\eta_S=\Omega(n^{-k/2})$.
\end{proof}

It follows that we can jointly measure the subset of unsharp degree--$2k$ observables $\gamma_S$ for which the $2k$ vertices of $S$ belong to $2k$ distinct partition subsets. However, we are not guaranteed to jointly measure observables whose $2k$ vertices belong to fewer than $2k$ distinct subsets. In this case, after classical post-processing, the relevant $2k\times 2k$ submatrix of $O$ is no longer monomial, and a non-zero determinant is not assured. We therefore implement $N$ POVMs $\widetilde\G^{(r)}$, $r=1,\ldots,N$, each defined by a random partition from $\mathcal{Y}$. In the following, we show that $N>4k$ random partitions generate the full set $\mathcal{S}_{2k}$, with high probability. Thus, we arrive at the conclusion described in Theorem \ref{thm:jm}.

\begin{proposition}\label{prop:partitions}
Let $2n=\ell(\ell+1)$ and $2k\leq\ell+1$. The probability of failing to generate $S\in\mathcal{S}_{2k}$ from a partition $Y=(Y_1,\ldots,Y_{\ell+1})\in\mathcal{Y}$ of $[2n]$, chosen uniformly at random from the set of all partitions $\mathcal{Y}$, is given by
\begin{equation}\label{Aeq:prob_fail}
P(S)=1-\ell^{2k}\binom{\ell+1}{2k}\binom{2n}{2k}^{-1}\,.
\end{equation}
Furthermore, for $k=\mathcal{O}(1)$, the failure probability satisfies $P(S)=\mathcal{O}(\ell^{-1})$.
\end{proposition}

\begin{proof}
Let $\mathcal{S}_{2k}(Y)$ denote the set of all $2k$--element subsets of $[2n]$ that are generated from $Y$, i.e.,
\begin{equation*}
\mathcal{S}_{2k}(Y)=\{\{j_1,j_2,\ldots,j_{2k}\} \,|\,\{j_{u},j_{u'}\}\nsubseteq Y_i\,,\forall\, u,u',i \}\,.
\end{equation*}
The probability of successfully generating the string $\{1,2,\ldots,2k\}$ from a partition $Y\in\mathcal{Y}$, chosen uniformly at random, is given by
\begin{equation}\label{Aeq:success}
\text{Prob}(\{1,2,\ldots,2k\}\in \mathcal{S}_{2k}(Y))=\frac{|\mathcal{P}(1,2,\ldots,2k)|}{|\mathcal{Y}|}\,,
\end{equation}
where
\begin{equation*}
\mathcal{P}(1,2,\ldots,2k)=\{Y\in\mathcal{Y}\,|\,\{1,2,\ldots,2k\}\in \mathcal{S}_{2k}(Y)\}\,,
\end{equation*}
is the set of all partitions which generate $\{1,2,\ldots,2k\}$.
To calculate this probability we first count the number of elements in the set of pairings
\begin{equation}
\mathcal{Q}=\{(S,Y)\,|\,S\in \mathcal{S}_{2k}(Y)\,,\,Y\in\mathcal{Y}\}\,.
\end{equation}
Using the invariance $|\mathcal{P}(1,2,\ldots,2k)|=|\mathcal{P}(S)|$ for all $S\in \binom{[2n]}{2k}$, and $|\mathcal{S}_{2k}|=\binom{2n}{2k}$, we have
\begin{equation}\label{Aeq:Q1}
|\mathcal{Q}|=|\mathcal{P}(1,2,\ldots,2k)|\,\binom{2n}{2k}\,.
\end{equation}
Alternatively, we can count the number of all partitions $\mathcal{Y}$, together with number of sets $S$ generated by each partition, i.e., $\mathcal{X}(Y)=\{S\in \mathcal{S}_{2k}(Y)\}$, such that 
\begin{equation}\label{Aeq:Q2}
|\mathcal{Q}|=|\mathcal{Y}|\cdot|\mathcal{X}(Y)|\,.
\end{equation}

Equating the two expressions of Eqs. (\ref{Aeq:Q1}) and (\ref{Aeq:Q2}), allows us to rewrite the success probability of Eq. (\ref{Aeq:success}) as $|\mathcal{X}(Y)|\binom{2n}{2k}^{-1}$. To calculate $|\mathcal{X}(Y)|$ note that there are $\ell+1$ subsets to choose the $2k$ elements from, and for every subset, we can choose any of the $\ell$ elements, hence
\begin{equation}
|\mathcal{X}(Y)|=\binom{\ell+1}{2k}\ell^{2k}\,.
\end{equation}
The probability of failing to generate a $2k$--element subset $S\in\mathcal{S}_{2k}$ is, therefore,
\begin{equation}
P(S):=\text{Prob}\left(S\notin \mathcal{S}_{2k}(Y)\right)=1-|\mathcal{X}(Y)|\binom{2n}{2k}^{-1}=1-\ell^{2k}\binom{\ell+1}{2k}\binom{2n}{2k}^{-1}\,,
\end{equation}
as stated in the proposition.

To determine the asymptotic behavior of the failure probability for $k=\mathcal{O}(1)$, let $L=\ell(\ell+1)$ such that
\begin{eqnarray}
\ell^{2k}\binom{\ell+1}{2k}{\binom{L}{2k}}^{-1}&=&\frac{(\ell+1)\ell(\ell-1)\cdots(\ell-2k+2)\ell^{2k}}{L(L-1)(L-2)\cdots(L-2k+1)}\,.
\end{eqnarray}
Thus,
\begin{eqnarray}
P(S)&=&\frac{L(L-1)(L-2)\cdots(L-2k+1)-(\ell+1)\ell(\ell-1)\cdots(\ell-2k+2)\ell^{2k}}{L(L-1)(L-2)\cdots(L-2k+1)}\\
&=&\frac{(L-1)(L-2)\cdots(L-2k+1)-(\ell-1)\cdots(\ell-2k+2)\ell^{2k}}{(L-1)(L-2)\cdots(L-2k+1)}=:\frac{\mathcal{N}}{\mathcal{D}}\,.
\end{eqnarray}
The numerator $\mathcal{N}$ can be expanded as
\begin{eqnarray}
\mathcal{N}&=&L^{2k-1}+\mathcal{O}(L^{2k-2})-\ell^{4k-2}+\alpha\ell^{4k-3}+\mathcal{O}(\ell^{4k-4})=\alpha\ell^{4k-3}+\mathcal{O}(\ell^{4k-4})\,,
\end{eqnarray}
where $\alpha\in\mathbb{R}$. The coefficient $\alpha$ can be found from the simple expansion,
\begin{equation}
(\ell-1)\cdots(\ell-2k+2)=\ell^{2k-2}-\tfrac{1}{2}(2k-1)(2k-2)\ell^{2k-3}+\mathcal{O}(\ell^{2k-4})\,,
\end{equation}
such that $\alpha=(2k-1)(2k-2)/2$. Similarly, for the denominator,
\begin{eqnarray}
\mathcal{D}&=L^{2k-1}+\mathcal{O}(L^{2k})=\ell^{4k-2}+\mathcal{O}(\ell^{4k-4})\,.
\end{eqnarray}
Therefore, for $k=\mathcal{O}(1)$, we have $P(S)=\mathcal{O}(\ell^{-1})$, completing the proof of Prop. \ref{prop:partitions}.
\end{proof}

Now suppose we sample $N$ independent partitions of $\mathcal{Y}$. The probability of failing to generate $S\in\mathcal{S}_{2k}$ from at least one of the $N$ partitions is $P(S|N)=P(S)^{N}$. Taking the union bound over all $2k$--element subsets we find that the probability of failing to generate \emph{all} elements of $\mathcal{S}_{2k}$ from $N$ partitions is, 
\begin{equation}
P_N:=P\left(\bigcup_{S\in\mathcal{S}_{2k}}S\,\big|\,N\right)\leq\binom{2n}{2k}P(S)^N\,.
\end{equation}
This leads to the following corollary of Prop. \ref{prop:partitions}.
\begin{corollary}\label{cor:partitions}
Suppose $2n=\ell(\ell+1)$ and $2k\leq\ell+1$, with $k=\mathcal{O}(1)$. Given $N$ partitions sampled uniformly at random from $\mathcal{Y}$, the probability $P_N$ of failing to generate all elements of $\mathcal{S}_{2k}$ satisfies,
\begin{equation}
P_N=\mathcal{O}(\ell^{4k-N})\,.
\end{equation}
\end{corollary}

For $k=2$, we observe that since $P_N=\mathcal{O}(\ell^{8-N})$, taking $N>8$ partitions ensures the probability of not generating all subsets becomes sufficiently small for large $n$. Recall that the $r$--th partition defines a POVM $\widetilde\G^{(r)}$, and every $S\in\mathcal{S}_{2k}$ generated from the partition implies $\gamma_S$ can be jointly measured by $\widetilde\G^{(r)}$ with $\eta_S=\Omega(n^{-k/2})$. Thus, Cor. \ref{cor:partitions} implies that we can jointly measure, with high probability, the full set of observables via a randomization over $N>4k$ POVMs.

Finally, note that up to a minor modification of Eq. (\ref{Aeq:prob_fail}), we can extend Prop. \ref{prop:partitions} and Cor. \ref{cor:partitions} to the more general case of $2n=\ell(\ell+1)+t$, where $0\leq t<2(\ell+1)$. Here, we add at most two elements to each of the $\ell+1$ subsets $Y_j$ of $Y$ to ensure all $2n$ integers are included in the partition. In the extremal case, $t=2\ell$, ($t$ must be even) we add an extra two integers to each subset $Y_i$ such that $|Y_i|=(\ell+2)$ for all $i=1,\ldots,\ell$. For simplicity, we can also add two ``fictitious'' elements to the remaining subset, meaning that we generate more terms than necessary. A simple computation reveals the same scaling behavior of $P_N$. Thus, each POVM $\widetilde\G^{(r)}$ is defined by a random partition $Y$ of $[2n]$ into $\ell+1$ subsets, as equally sized as possible. The perfect matching $\mathcal{M}$ is partially constructed by first ignoring the ``extra'' $t$ vertices and following the strategy described in Fig. \ref{fig:perm_B}. Edges are then arbitrarily added to complete the perfect matching. We conclude the proof of Theorem \ref{thm:jm} by observing that each POVM $\widetilde\G^{(r)}$ satisfies Lemma \ref{lem:jm_2k}, and, with high probability, we can jointly measure the full set of observables $\{\gamma_S\,|\,S\in\mathcal{S}_{2k}\}$ after $N>4k$ partitions.

\section{Symmetries of Majorana observables}\label{Asec:symmetries}

In this appendix, by studying the symmetries of the degree--$k$ Majorana observables described by the braid group, we show that the measurement assemblage $\{\gamma_S\,|\,S\in\mathcal{S}_k\}$ is uniformly and rigidly symmetric (see Definitions \ref{def:uniformity} and \ref{def:rigidity}). Importantly, these properties result in a simplified expression for the incompatibility robustness, which we discuss in Appendix \ref{Asec:syk}. 

Using terminology from \cite{nguyen20}, it is useful to describe the degree--$k$ measurement assemblage as a \emph{bundle} of outcomes, defined by the triple $(\meo_k,\pi,\mathcal{S}_k)$ consisting of the set $\mathcal{S}_k$ of Majorana measurements, the set $\meo_k=\{e_S\,|\,e_S=\pm1\,,S\in\mathcal{S}_k\}$ of all possible outcomes, and a map $\pi:\meo_k\rightarrow\mathcal{S}_k$ such that $\pi(e_S)=S$. Furthermore, $\pi^{-1}(S)$ is called the \emph{fibre} over $S$ and denotes the set of outcomes of $\M_S$, i.e., $\sum_{e_S\in\pi^{-1}(S)}\M_S(e_S)=\id$, for each $S$. See Fig. \ref{fig:bundle-section} for an illustration.

We now introduce a symmetry group $G$ of the measurement assemblage acting on the bundle, together with a unitary representation $U:G\rightarrow U(d)$. The symmetry group $G$ permutes the outcomes of the bundle in a way that is compatible with the assignment of the outcomes to the measurements $\M_S$:
\begin{equation}
g[\pi(e_S)]=\pi[g(e_S)] \,\,\,\text{for all} \,\,\,g\in G \,\,\,\text{and}\,\,\,e_S\in\meo_k\,.
\end{equation}
For outcomes related by a symmetry transformation $g\in G$, their associated POVM effects are related by the unitary $U_g$:
\begin{equation}
\M_S(e_S)=U_g\M_{g^{-1}[S]}({g^{-1}[e_S]})U^{-1}_g \,\,\,\text{for all} \,\,\,g\in G \,\,\,\text{and}\,\,\,e_S\in\meo_k\,.
\end{equation}
To describe the symmetries of the measurement assemblage we consider the braid group introduced below.

\begin{figure}
    \centering
    \includegraphics[width=0.4\textwidth]{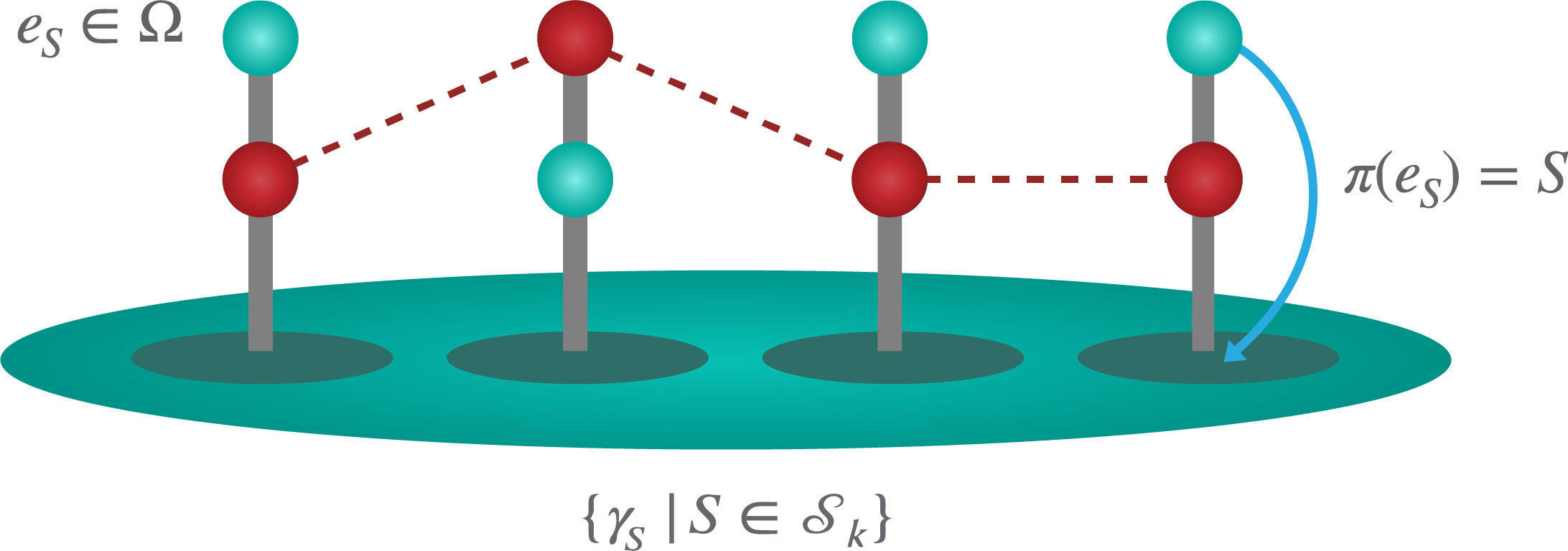}
    \caption{An illustration of the assemblage of Majorana observables $\{\gamma_S\,|\,S\in\mathcal{S}_k\}$, viewed as a bundle of outcomes $(\meo_k,\pi,\mathcal{S}_k)$. The set $\Omega_k$ consists of all possible measurement outcomes (colored balls), and the map $\pi:\meo_k\rightarrow\mathcal{S}_k$ is defined such that $\pi(e_S)=S$, where $S$ (shaded ellipse) labels the measurement $\M_S$. A fiber over $S$  (balls attached to a grey vertical line) corresponds to the outcomes of $\M_S$. A section $\ve$ of $\Omega_k$ (defined in Appendix \ref{Asec:syk}) is a particular outcome string containing one outcome for every $S\in\mathcal{S}_k$. An example of a section is illustrated by the subset of red-colored balls.}
    \label{fig:bundle-section}
\end{figure}

\subsection{The braid group}\label{Asec:braids}

We are interested in studying the symmetries of the measurement assemblage described by the braid group $\mathcal{B}_{n}$. Consider a braid as a collection of $n$ strands with their ends aligned in two rows of points (one above the other). The standard (or Artin) generators of the braid group are $\{b_i,b_i^{-1}\}$, $i=1,\ldots,n-1$, which act by crossing the $i$--th strand over the $(i+1)$--th strand (see Fig. \ref{fig:braid_group}). The generators $b_i$ and $b_i^{-1}$ are the inverse of each other and satisfy the relations,
\begin{eqnarray}
b_ib_j&=&b_jb_i\,\,\,\text{for}\,\,\,|i-j|>1\,,\label{eq:braid_generators1}\\
b_ib_{i+1}b_i&=&b_{i+1}b_ib_{i+1}\,\,\,\text{for}\,\,\,i=1\ldots,n-2\,.\label{eq:braid_generators2}
\end{eqnarray}
The first relation shows that two non-neighboring braids commute, while the second is analogous to the third Reidemeister move in knot theory. Together they are sufficient to fully describe the braid group \cite{artin47}. Multiplication in the group corresponds to the sequential application of the braid operations. 
%The braid group is an extensions of the symmetric group $S_{n}$, which satisfies Eqs. (\ref{eq:braid_generators1})--(\ref{eq:braid_generators2}), and additionally $b_i^2=1$ for all $i$.

\begin{figure}[h]
    \centering
    \includegraphics[width=0.4\textwidth]{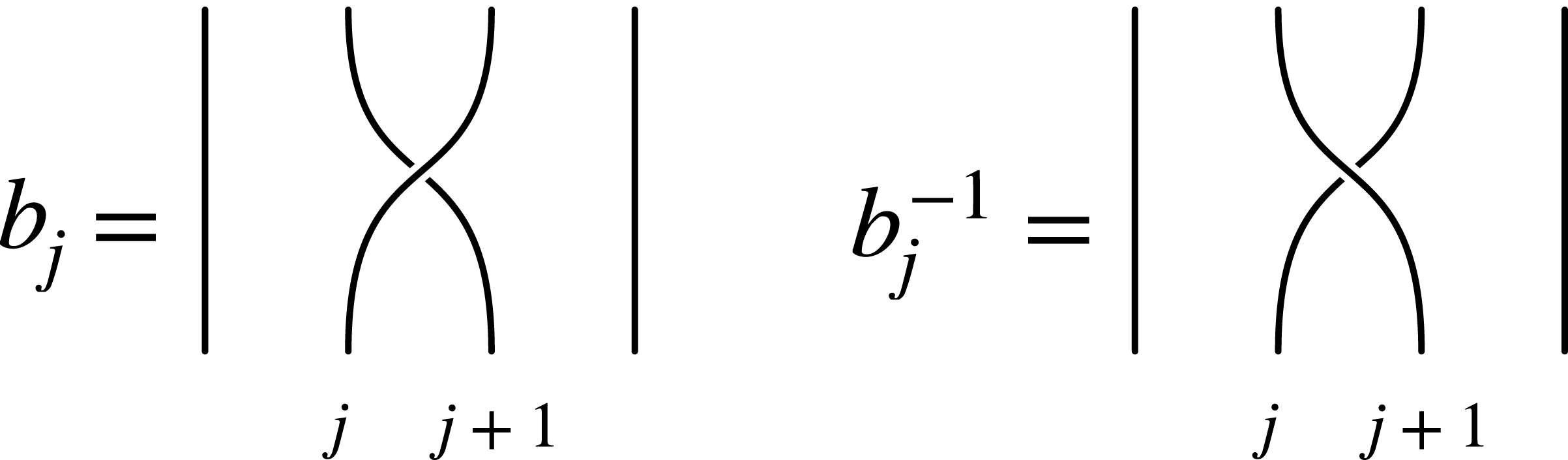}
    \caption{The generators of the braid group. The braid element $b_j$ acts on the $n$ strands by crossing the $i$--th strand over the $(i+1)$--th strand.}
    \label{fig:braid_group}
\end{figure}

Let $\varphi:\mathcal{B}_{2n}\rightarrow U(2^n)$ denote the spinor representation of the braid group, with the associated generators
\begin{equation}\label{Aeq:braids}
B_i:=\varphi(b_i)=\exp(-\frac{\pi}{2}\gamma_i \gamma_{i+1})=\frac{1}{\sqrt{2}}(\id-\gamma_i \gamma_{i+1})\,.
\end{equation}
These nearest-neighbor exchange operators generate the non-local operators,
\begin{eqnarray}
B_{i,j}&=&\exp(-\frac{\pi}{4}\gamma_i \gamma_j)\\
&=&B_{j-1}\hdots B_{i+1}B_iB^{\dagger}_{i+1}\hdots B^{\dagger}_{j-1}\,.
\end{eqnarray}
A braiding transformation acts, via conjugation, on the Majorana operators as,
\begin{equation}\label{eq:action_on_maj}
B_{i,j}\gamma_{\ell} B_{i,j}^{\dagger}=
\begin{cases}
 \gamma_{\ell} &\mbox{if} \,\,\,\ell\notin \{i,j\}\,, \\
 \gamma_j &\mbox{if}\,\,\, \ell=i \,,\\
- \gamma_i &\mbox{if}\,\,\, \ell=j \,.
 \end{cases}
\end{equation}
For a degree--$k$ Majorana operator $\gamma_S$, the braid transformations satisfy,
\begin{equation}\label{eq:action_on_maj_prod}
B_{i,j}\gamma_S B_{i,j}^{\dagger}=
\begin{cases}
 \gamma_S &\mbox{if} \,\,\, i,j\notin S \,\,\text{or}\,\,i,j\in S\,,\\
\gamma_{S}\gamma_i\gamma_j &\mbox{otherwise}\,.
 \end{cases}
\end{equation}
We will also consider the group generated by the $2n$ elementary braid operators $B_i$ and their inverses, defined in Eq. (\ref{Aeq:braids}), together with a single Majorana operator $\gamma_{j}$, with $j\in[2n]$. While the Braiding transformations obey Eq. (\ref{eq:action_on_maj})--(\ref{eq:action_on_maj_prod}), the single Majorana transformations satisfy $\gamma_{j}\gamma_S\gamma_{j}^{\dagger}=\pm\gamma_S$, with the sign depending on the parity of $|S|$ and whether $j\in S$.

\subsection{Uniformity and rigidity} % (fold)
\label{sec:uniformity_and_rigidity}

\begin{definition}[Uniformity]\label{def:uniformity}
The set of degree--$k$ Majorana measurements, denoted by $\mathcal{S}_k$, is \emph{uniformly symmetric} if any outcome $e_S\in\meo_k$ can be related to any other outcome $e_{S'}\in\Omega_k$ by a symmetry transformation $g\in G$, i.e., for any pair $S,S'\in\mathcal{S}_k$, there exists $g\in G$ such that $g(e_S)=e_{S'}$.
\end{definition}

\begin{proposition}[Braiding recipe]\label{prop:braiding_recipe}
    Consider an $n$ mode fermionic system.
    For any two Majorana operators $\gamma_S,\gamma_{S'}$, where $S, S'\in \mathcal{S}_k$, there exists a set of braidings that maps $\gamma_S\mapsto\gamma_{S'}$.
    In particular, the set of braidings $\{B_{I_1,J_1},\ldots,B_{I_{|I|},J_{|I|}}\}$, where $I:=\{I_1,\ldots,I_{|I|}\} = S\backslash S'$ and $J := \{J_1,\ldots,J_{|I|}\}=S'\backslash S$,
    acts such that
    \begin{equation}
        \gamma_{S'} = \pm \left(\prod_{i=1}^{|I|}B_{I_i,J_i}\right)\gamma_S \left(\prod_{j=1}^{|I|}B_{I_j,J_j}\right)^\dagger.
    \end{equation}
\end{proposition}
\begin{proof} 
This follows directly from the transformation described in Eq. (\ref{eq:action_on_maj_prod}).
\end{proof}
As an example, consider the two Majorana observables with indices $S = \{1,2,3,4\}$ and $S' = \{2,4,6,8\}$.
A set of braids that can map $\gamma_1\gamma_2\gamma_3\gamma_4\mapsto\gamma_2\gamma_4\gamma_6\gamma_8$, up to signs, is $\{B_{1,6},B_{3,8}\}$.
Then, applying these transformations to $\gamma_1\gamma_2\gamma_3\gamma_4$ give us
\begin{equation}
    B_{1,6}B_{3,8}\cdot \gamma_1\gamma_2\gamma_3\gamma_4\cdot B_{3,8}^ \dagger B_{1,6}^{\dagger} = -\gamma_2\gamma_4\gamma_6\gamma_8.
\end{equation}

\begin{theorem}[Uniformity of Majorana observables]\label{thm:uniformity}
    For an n mode fermionic system, the set of degree--$k$ Majorana measurements is uniformly symmetric under braid transformations.
\end{theorem}

\begin{proof}
    For $k=1$, any two outcomes $e_i$, $e_j$, where $i,j\in[2n]$, are related by the braid transformation $B_{i,j}$ (cf. Eq. (\ref{eq:action_on_maj})). For $i<j$, the braid $B_{i,j}$ applied to the measurement $\M_i(e_i)$ gives
    \begin{equation}
        \M_j(e_j) =
        \begin{cases}
            B_{i,j} \M_i(e_i) B_{i,j}^\dagger, & \text{ if } e_j = e_i\\
            B_{i,j}B_{i,j}B_{i,j} \M_i(e_i) B_{i,j}^\dagger B_{i,j}^\dagger B_{i,j}^\dagger , & \text{ if } e_j = -e_i.
        \end{cases}
    \end{equation}
    A similar transformation holds for $i>j$ by applying the braiding $B_{j,i}$. 

    For $k>1$, Prop.~\ref{prop:braiding_recipe} gives us a recipe to map any degree--$k$ Majorana operators to any other degree--$k$ operator, up to signs.
    Applying this transformation to the measurement $\M_S(e_S)$, results in
    \begin{equation}
        \left(\prod_{i=1}^{|I|}B_{I_i,J_i}\right)\M_S(e_S) \left(\prod_{j=1}^{|I|}B_{I_j,J_j}\right)^\dagger = \M_{S'}(\pm e_S),
    \end{equation}
    where the sign $\pm$ depends on the sequence of braidings.
    To obtain the opposite outcome sign, we can apply a braid $B_{i,j}$ twice, where $i\in S'$ and $j\notin S'$. Then,
    \begin{equation}
        B_{i,j}B_{i,j}\M_{S'}(\pm e_S)B_{i,j}^\dagger B_{i,j}^\dagger = \M_{S'}(\mp e_S).
    \end{equation}
    Consequently, the set degree--$k$ Majorana measurements are uniformly symmetric.
\end{proof}

We now show that the degree--$k$ measurement assemblages are rigidly symmetric, as defined below. First, we prove rigidity of the Majorana observables via transformations from the braid group, assuming fixed parity of the system. Then, we prove rigidity without restricting the parity, by considering braid and single Majorana transformations. Note that the parity is measured by the operator $Q:=\gamma_{[2n]}=i^n\gamma_1\gamma_2\cdots\gamma_{2n}$.

\begin{definition}[Rigidity]\label{def:rigidity}
Consider the symmetry group $G$, and a unitary representation, $U:G\to \text{U}(2N)$. Let $G[e_S]=\{g\in G :g(e_S)=e_S\}$ be the stabilizer group of outcome $e_S\in\meo_k$, and $U(G[e_S])=\{U_g :g\in G[e_S]\}$ the corresponding set of unitaries. The set of degree--$k$ Majorana  measurements is \emph{rigidly symmetric} if, for each outcome $e_S\in\meo_k$, the commutant of $U(G[e_S])$ is given by $\text{Span}(\id,\Pi_{e_S})$, where $\Pi_{e_S}$ is a projection operator.
\end{definition}

We note that degree--$k$ Majorana observables are uniformly symmetric (cf. Theorem \ref{thm:uniformity}) and therefore all outcomes can be treated equivalently. We therefore consider the stabiliser group $G[S]=\{g\in G :g(S)=S\}$ of a single observable $\gamma_S$. Rigidity holds if the commutant of $U(G[S]) = \{ U_g \,|\, U_g \gamma_S U_g^\dagger = \gamma_S, \,  U_g\in U(G)\}$ is given by $\text{Span}(\id,\gamma_S)$.

\begin{theorem}[Rigidity with fixed parity]\label{th:rigid_parity}
    For an $n$ mode fermionic system with fixed parity, the set of degree--$k$ observables $\{\gamma_S\,|\, S\in\mathcal{S}_k \}$ is rigidly symmetric under braid transformations.
\end{theorem}
\begin{proof}
    Consider the observable $\gamma_S$ and its stabilizer group $U(G[S])$ under braid transformations.
    We know that $[\gamma_S,B_{i,j}] = 0$, if the pair $i,j\in S$ or if $i,j\notin S$.
    Consequently, a subset of the stabilizer group $U(G[S])$ is $\mathcal{P}_S = \{B_{i,j}\,| \, i,j\in S \text{ or } i,j\notin S  \}$.
    To prove rigidity it is enough to consider at first only the commutant of $\mathcal{P}_S$.
    Let $\gamma_Y$ be a Majorana observable, such that, $[\gamma_Y,B_{i,j}] = 0,\, \forall B_{i,j}\in \mathcal{P}_S$.
    This restricts $Y$ to equal $S$, its complement or $Y=[2n]$, i.e. $Y\in\{S,[2n]\backslash S,[2n]\}$.
    Thus, the commutant of $\mathcal{P}_S$ is given by $\text{Span}(\id,\gamma_S, Q\gamma_S,Q)$, where $Q$ is the parity operator.
    However, since the fermionic system is restricted to the fixed-parity sector, the parity operator $Q$ and the operator $\gamma_S$ are equivalent to $\id$ and $Q\gamma_S$, respectively.
    Hence, the commutant of $U(G[S])$ is $\text{Span}(\id,\gamma_S)$.
\end{proof}

\begin{theorem}[Rigidity without fixed parity]
    For an $n$ mode fermionic system, the set of degree--$k$ observables $\{\gamma_S\,|\, S\in\mathcal{S}_k \}$ is rigidly symmetric under braid and single Majorana transformations.
    \end{theorem}
\begin{proof}
    Consider the observable $\gamma_S$ and its stabilizer group $U(G[S])$ under braid and single Majorana transformations.
    In the proof of Theorem~\ref{th:rigid_parity} we showed that the commutant of $\mathcal{P}_S$ is $\text{Span}(\id,\gamma_S, Q\gamma_S, Q)$, where $\mathcal{P}_S\subset U(G[S])$ and the stabilizer group contains only braid transformations.
    Additionally, let us consider the stabilizer group which includes single Majorana transformations. First, we assume $|S|$ is even and let $\mathcal{O}_S = \{\gamma_{j}\,|\, j\in [2n]\backslash S\}\subset U(G[S])$.
    The observable $Q\gamma_S = \gamma_{[2n]\backslash S}$ does not commute with any unitary in $\mathcal{O}_S$, and the same can be said by the parity operator $Q$.
    Therefore, any operator $\gamma_Y$ that commutes with all unitaries in $\mathcal{P}_S$ and $\mathcal{O}_S$ must belong to $\text{Span}(\id,\gamma_S)$.
    Hence, the commutant of $U(G[S])$ is $\text{Span}(\id,\gamma_S)$ and the measurement assemblage is rigidly symmetric. If, on the other hand, $|S|$ is odd, we consider $\mathcal{O}_S = \{\gamma_{j}\,|\, j\in S\}\subset U(G[S])$ and the same argument follows.
\end{proof}

\section{Incompatibility robustness and the SYK model}\label{Asec:syk}

A \emph{section} $\ve:\mathcal{S}_k\rightarrow\Omega_k$ of the bundle $(\Omega_k,\pi, \mathcal{S}_k)$, as illustrated in Fig. \ref{fig:bundle-section}, is a map that assigns to every $S\in\mathcal{S}_k$ a particular outcome $e_S\in\{\pm 1\}$. In particular, $\pi[\ve(S)]=S\in\mathcal{S}_k$. We denote $\ve(\mathcal{S}_k):=\{\ve(S):S\in\mathcal{S}_k\}$ as the associated outcome string, which we will often abbreviate to $\ve$. We define $\Xi_k$ as the set of all sections of the bundle.

At some noise threshold, the measurements $\M^{\eta}_S(\pm)=\frac{1}{2}(\id\pm \eta\gamma_S)$ for all $S\in\mathcal{S}_k$, become jointly measurable, and their POVM effects constitute the marginals of a parent measurement $\G$ with outcomes in $\Xi_k$, i.e.,
\begin{equation}\label{Aeq:jm_sections}
\M^{\eta}_S(e_S)=\sum_{\ve \in\Xi_k}\delta_{\ve[\pi(e_S)],e_S}\G(\ve)\,.
\end{equation}
Note that this condition is equivalent (for finite outcome observables) to the definition of joint measurability provided in Eq. (\ref{eq:postprocessing}), as shown in \cite{ali09}. The minimum amount of noise required to reach this threshold quantifies the incompatibility of the set of measurements and is known as the incompatibility robustness,
\begin{equation}\label{Aeq:SDP}
\eta_k^*=\max\{\eta>0\,|\,\M^{\eta}_S \,\,\text{are compatible}\,\,  \forall S\in\mathcal{S}_k\}\,.
\end{equation}
This quantity can be solved by semidefinite programming (SDP), whereby Eq. (\ref{Aeq:SDP}) is optimized over the variables $\eta$ and Hermitian operators $\G(\ve)$, such that Eq. (\ref{Aeq:jm_sections}) is satisfied and $\G(\ve)\geq 0$ for all $\ve \in \Xi_k$. In general, the number of variables (matrices) in the optimization equals $|\Xi_k|$, which scales exponentially in the system size. The number of variables can be reduced significantly due to the symmetries from $G$ carrying over to the space of sections \cite{nguyen20}. While this simplifies the problem considerably, a simpler expression can be derived by studying the equivalent dual formulation:

\begin{align}\label{Aeq:SDP}
\eta_k^*&=\min_{X_{e_S}} \quad1+\sum_{e_S\in\Omega_k}\tr{X_{e_S}\M_S(e_S)}\nonumber\\
&\phantom{==} \text{s.t.}\quad 1+\sum_{e_S\in\Omega_k}\tr{X_{e_S}\M_S(e_S)}\geq\frac{1}{d}\sum_{e_S\in\Omega_k}\tr{\M_S(e_S)}\tr{X_{e_S}}\,,\nonumber\\
&\phantom{==}\qquad\sum_{e_S\in\Omega_k}\delta_{\ve[\pi(e_S)],e_S}X_{e_S}\geq 0\quad \forall \ve\in\Xi_k\,,
\end{align}
where the variables $X_{e_S}$ are $d\times d$ Hermitian matrices, with $d=2^n$. Importantly, the symmetries of the measurement assemblage also apply to the dual variables, i.e. $X_{e_S}=U_gX_{g^{-1}(e_S)} U_g^{-1}$ for all $g\in G$ and $e_S\in\Omega_k$. Consequently, for uniformly and rigidly symmetric assemblages, $X_{e_S}$ commutes with the unitaries of the stabiliser group $U(G[e_S])$, and therefore have the form $X_{e_S}=a\id+b\M_S(e_S)$. Furthermore, the uniformity of the assemblage ensures $a$ and $b$ are independent of the outcomes. As described in \cite{nguyen20}, it follows from the constraints of Eq. (\ref{Aeq:SDP}) that the inequalities are saturated for $a=2\mu_k/(d|\mathcal{S}_k|^2)$ and $b=-2/(d|\mathcal{S}_k|)$,
where
\begin{equation}\label{eq:norm}
\mu_k=\max_{\ve\in\Xi_k}{\no{\sum_{S\in\mathcal{S}_k} \M_S(\ve[S])}}_{\infty}\,,
\end{equation}
is the maximum eigenvalue of $\sum_S \M_S(\ve[S])$ among all sections $\ve\in\Xi_k$.
%Since $|\Omega_k|=2|\mathcal{S}_k|$ and $d = 2^n$, we have $Z=d|\mathcal{S}_k|/2$. 
We can rewrite Eq. (\ref{eq:norm}) as
\begin{equation}
\mu_k=\frac{1}{2}\left(|\mathcal{S}_k|+\max_{\ve\in\Xi_k}{\no{H_k(\ve)}}_{\infty}\right)\,,
\end{equation}
where $H_k(\ve)=\sum_{S\in\mathcal{S}_{k}}e_{S}\gamma_S.$ After substituting $X_{e_S}=a\id+b\M(e_S)$ into the expression of $\eta_k^*$ from Eq. (\ref{Aeq:SDP}), we arrive at
\begin{equation}\label{Aeq:robustness}
\eta_k^*=\frac{2}{|\mathcal{S}_k|}\left(\mu_k-\frac{|\mathcal{S}_k|^2}{|\Omega_k|}\right)= |\mathcal{S}_{k}|^{-1}\max_{\ve\in\Xi_k}{\no{H_k(\ve)}}_{\infty}\,,
\end{equation}
where $|\Omega_k|=2|\mathcal{S}_{k}|$, in agreement with Lemma \ref{lem:robustness} of the main text.

Interestingly, the operator $H_k(\ve)$ is related to the Sachdev-Ye-Kitaev (SYK) model of degree--$k$ \cite{sachdev93,kitaev15,feng19,hastings21,herasymenko22}, although it is more common to assume the coefficients $e_S$ are Gaussian random variables. The spectral behavior of the SYK model has been studied extensively, and recent results by Hastings and O'Donnell \cite{hastings21} can be used to bound the incompatibility robustness $\eta^*_{2k}$, for $k\leq 5$ (see Lemma \ref{thm:odonnell} and Theorem \ref{thm:inc_rob}). In Appendix \ref{A:IR} we find a tight upper bound for $\eta^*_2$ by evaluating the maximum eigenvalue of $H_2(\ve)$ using the theory of tournament matrices.

\section{Incompatibility robustness: Quadratic Majoranas}\label{A:IR}

The incompatibility robustness of the assemblage of degree--$2$ Majorana observables can be found (via Lemma \ref{lem:robustness}) by evaluating the maximum eigenvalue of the class of operators
\begin{equation}
 H_2=i\sum_{i<j}e_{ij}\gamma_{i}\gamma_{j}=\frac{i}{2}\sum_{i,j=1}^{2n}e_{ij}\gamma_{i}\gamma_{j}\,,
 \end{equation} 
where $e_{ij}\in\{\pm 1\}$, $e_{ij}=-e_{ji}$, and $e_{ii}=0$. Every Hamiltonian of this form corresponds to a unique real $2n\times 2n$ \emph{antisymmetric} matrix $E=-E^{T}$, with off-diagonal entries $\pm1$. Such matrices are known as \emph{tournament} matrices \cite{mccarthy96,Wesp02}, and are commonly defined in terms of a round-robin tournament in which every pair of competitors meet once.
\begin{definition}\label{def:tournament}
Let $\mathcal{E}_{2n}$ denote the set of $2n\times 2n$ tournament matrices. A matrix $E\in\mathcal{E}_{2n}$ is a tournament matrix if its entries satisfy
\begin{equation}\label{eq:tournament}
e_{ij}=
\begin{cases}
1 &\mbox{if player $i$ defeats player $j$} \\
-1 &\mbox{if player $j$ defeats player $i$} \\
 0 &\mbox{if $i=j$} \,.
 \end{cases}
\end{equation}
\end{definition}

Since $E\in\mathcal{E}_{2n}$ is a real antisymmetric matrix, it can be written as $E=O^TBO$ where $O\in O(2n)$ and $B$ is a block diagonal matrix with $2\times 2$ blocks
\begin{equation}
B_{j}=\left( \begin{array}{cc}
           0 & \lambda_j \\
           -\lambda_j & 0
         \end{array} \right)\,,
\end{equation}
where $\lambda_j\in\mathbb{R}$ and $\pm i\lambda_j$ are the eigenvalues of $E$. We can define a new (rotated) set of $2n$ Majorana operators $\widetilde \gamma_j=\sum_k o_{jk}\gamma_k$, which anti-commute and are Hermitian. After some simple algebra, the operator $H_2$ can be written as
\begin{eqnarray}
H_2&=&\frac{i}{2}\sum_{i,j=1}^{2n}e_{ij}\gamma_{i}\gamma_{j}=\frac{i}{2}\sum_{i,j,k,\ell=1}^{2n}o_{ki}b_{k\ell}o_{\ell j}\gamma_i\gamma_j\\
&=&\frac{i}{2}\sum_{k,\ell=1}^{2n}b_{k\ell}\widetilde\gamma_{k}\widetilde\gamma_{\ell}=\frac{i}{2}\sum_{j=1}^{n}\lambda_j(\widetilde \gamma_{2j-1}\widetilde \gamma_{2j}-\widetilde \gamma_{2j}\widetilde \gamma_{2j-1})\\
&=&i\sum_{j=1}^n\lambda_j\,\widetilde \gamma_{2j-1}\widetilde \gamma_{2j}\,.
\end{eqnarray}
As $H_2$ is now a sum of commuting projection operators (with eigenvalues $\pm 1$), the eigenvalues of $H_2$ are of the form $\sum_j \pm \lambda_j$. Rewriting in terms of the fermionic operators which satisfy $\widetilde a_j^{\dagger}\widetilde a_j=\tfrac{1}{2}(\id-i\widetilde \gamma_{2j-1}\widetilde \gamma_{2j})$, we have
\begin{equation}
H_2=\sum_{j=1}^n\lambda_j(\id-2\widetilde a_{j}^{\dagger}\widetilde a_{j})\,.
\end{equation}
From $\widetilde a_j^{\dagger}\widetilde a_j\ket{m_j}=m_j\ket{m_j}$, with $m_j=0,1$, it follows that the eigenvalues of $H_2$ are given by
 \begin{equation}
\lambda_{\bf m}=\sum_{j=1}^{n}\lambda_j(-1)^{m_j}\,,
\end{equation}
where ${\bf m}=(m_1,m_2,\ldots,m_n)$ with $m_j\in\{0,1\}$. The maximum eigenvalue is reached at ${\bf m}^*$ where $m^*_j=1$ if $\lambda_j<0$ and $m^*_j=0$ if $\lambda_j\geq 0$ such that
\begin{eqnarray}
{\no {H_2}}_{\infty}=\sum_{j=1}^{n}|\lambda_j|\,.
\end{eqnarray}
The next result places a bound on the maximum eigenvalue of $H_2$, using the theory of tournament matrices \cite{Wesp02,Adiga10}.

\begin{lemma}\label{A:lem:hboundpair}
Let $E\in\mathcal{E}_{2n}$, as defined in Eq. (\ref{eq:tournament}), with eigenvalues $\pm i\lambda_j$. Then,
%For $H_s$ defined in Eq. (\ref{eq:maj_pairs}), we have
\begin{equation}\label{eq:skew_energy}
\sum_{j=1}^{n}|\lambda_j|\leq n\sqrt{2n-1}\,.
\end{equation}
The bound is tight if and only if there exists $E\in\mathcal{E}_{2n}$ such that $E+\id$ is a $2n\times 2n$ (real) Hadamard matrix.
\end{lemma}
\begin{proof}
We follow a proof used to upper bound the skew energy of digraphs and tournament matrices \cite{Adiga10}. Suppose $E\in\mathcal{E}_{2n}$ has eigenvalues $\theta_k$, $k=1,\ldots,2n$. By Cauchy-Schwarz,
\begin{eqnarray}\label{eq:CS}
\left(\sum_{k=1}^{2n}|\theta_k|\right)^2&\leq& 2n\sum_{k=1}^{2n}|\theta_k|^2=2n\sum_{k=1}^{2n}(-\theta_k^2)\\
&=&2n\tr{-S^2}\,.
\end{eqnarray}
For any $E\in\mathcal{E}_{2n}$, we have $E^2=-EE^T$. Furthermore, each row vector of $E$ has squared norm $2n-1$, hence $(E^2)_{kk}=-(2n-1)$. It follows that $\tr{-E^2}=2n(2n-1)$ and, therefore,
\begin{equation}\label{eq:skew_bound}
\sum_{k=1}^{2n}|\theta_k|\leq 2n\sqrt{2n-1}\,.
\end{equation}
Collecting the eigenvalues into pairs $\pm i\lambda_j$ and counting one from each pair, we arrive at Eq. (\ref{eq:skew_energy}).
%Since ${\no {H_s}}_{\infty}= \sum_{j=1}^n|\lambda_j|$, where $\pm i\lambda_j$ are the eigenvalues of $S\in\mathcal{S}_{2n}$, then Eq. (\ref{eq:skew_energy}) follows.

The Cauchy-Schwarz inequality in Eq. (\ref{eq:CS}) is saturated if and only if $|\theta_k|^2=\alpha$ for all $k$. Therefore, Eq. (\ref{eq:skew_bound}) is saturated if and only $|\theta_k|=\sqrt{2n-1}$ for $k=1,\ldots,2n$. As $E$ is a normal matrix, there exists a unitary $U$ such that $U^{\dagger}EU=D=\text{diag}(\theta_1,\ldots,\theta_{2n})$. Taking the conjugate transpose we obtain $U^{\dagger}E^TU=D^*$, and therefore $EE^T=(2n-1)\id$. It follows that $H=E+\id$ satisfies $HH^{T}=2n\id$, and therefore $H$ is a $2n\times 2n$ Hadamard matrix.
\end{proof}

Thus, we have ${\no {H_2}}_{\infty}\leq n\sqrt{2n-1}$ for all $E\in\mathcal{E}_{2n}$. Recalling Eq. (\ref{eq:inc_rob}) with $|\mathcal{S}_2|=\binom{2n}{2}=n(2n-1)$, we arrive at the bound on the incompatibility robustness $\eta_2^*\leq\frac{1}{\sqrt{2n-1}}$ for any $n$, as given in part $(i)$ of Theorem \ref{thm:inc_rob}.

The simplest example occurs when $n=2$, where there exists a tournament matrix
\begin{equation}
E=\left( \begin{array}{cccc}
           0 & 1 & 1 & 1\\
           -1 &  0 & 1 &  -1\\
           -1 &  -1 & 0 &  1\\
           -1 &  1 & -1 &  0
         \end{array} \right)\in\mathcal{E}_4\,,
\end{equation}
such that $E+\id$ is (skew)-Hadamard. It follows from Theorem \ref{thm:inc_rob} that the degree--2 incompatibility robustness for the two mode system is given by $\eta_2^*=\frac{1}{\sqrt{3}}$.

\section{Connections to classical shadows}\label{A:CS}

The classical shadow protocol yields a classical representation of an unknown quantum state via a randomized measurement strategy sampling from a unitary ensemble $\mathcal{U}$  \cite{huang20}. After classical post-processing of the outcomes of a computational basis measurement, one obtains an operator, known as the classical shadow, which is not positive semidefinite but yields the desired state in expectation. This provides a single-shot estimator for multiple properties of the quantum system, including for simultaneous estimation of expectation values for non-commuting observables.

The original formulation of classical shadows, based on sampling from ensembles of random Clifford $n$--qubit unitaries or single-qubit Clifford unitaries, has been adapted to fermionic systems \cite{Zhao21,wan22,low22,ogorman22}. In one instance \cite{Zhao21}, a fermionic classical shadow is constructed from sampling $U_O\in \mathcal{U}_{\text{FGU}}$ from the ensemble of fermionic Gaussian unitaries which are also $n$--qubit Clifford unitaries (i.e., the braiding group). For an $n$ mode fermionic system, a unitary $U_O\in\mathcal{U}_{\text{FGU}}$ and outcome $z\in\{\pm 1\}^n$ of the computational basis measurement, the fermionic classical shadow is given by
\begin{equation}\label{eq:fermion_cs}
\hat\rho_z^O=\frac{1}{2^n}\left(\id+\sum_{j=1}^n\lambda_{2j}^{-1}\sum_{R\in\mathcal{D}_{2j}}\bra{z}\gamma_R\ket{z}\sum_{S\in\mathcal{S}_{2j}}\det(O_{R,S})\gamma_{S}\right)\,,
\end{equation}
where $\lambda_{2j}=\binom{n}{j}/\binom{2n}{2j}$, and $O\in \text{Alt}(2n)$, with $\text{Alt}(2n)$ the alternating group of $2n\times 2n$ permutation matrices of determinant one. The shadow provides a single shot estimator $\hat\gamma_S=\tr{\gamma_S\hat\rho_z^O}$ for the expectation value of any even degree Majorana observable. The variance of this estimator, which controls the sample complexity, is bounded by the shadow norm $\no{\gamma_S}_{\text{FGU}}=\sqrt{\binom{2n}{2k}/\binom{n}{k}}$.

In recent work, several connections were made between joint measurements and classical shadows \cite{mcnulty22}. For instance, a classical shadow can be used to define a joint measurement and to provide a sufficient condition for the compatibility of an arbitrary set of measurements. We now apply the fermionic classical shadow of Eq. (\ref{eq:fermion_cs}) to derive a joint measurability condition on the set of degree--$2k$ Majorana observables. 

The randomized measurement procedure of classical shadows can be described by a single POVM $\G(z,U_O)=\frac{1}{|\mathcal{U}_{\text{FGU}}|}U_O^{\dagger}\kb{z}{z}U_O$, where $U_O\in\mathcal{U}_{\text{FGU}}$ and $z\in\{\pm1\}^n$. For an arbitrary POVM $\M(e)$, let $q(e|\M,z,U_O)=\tr{\M(e)\hat\rho_{z}^O}$, which, in expectation, yields the outcome statistics of $\M(e)$. To derive a sufficient condition for joint measurability of a set of degree--$2k$ Majorana observables, we require that $q(\pm|S,z,U_O)$ is a classical post-processing function for all unsharp measurements $\M_S(\pm)=\frac{1}{2}(\id\pm\eta_{2k}\gamma_S)$, with $S\in\mathcal{S}_{2k}$. This imposes the condition $\tr{\M_S(\pm)\hat \rho_z^O}\geq 0$ for all $S\in\mathcal{S}_{2k}$, which simplifies to $1\pm\eta_{2k}\tr{\gamma_S\hat \rho_z^O}\geq 0$. From Eq. (\ref{eq:fermion_cs}) we have,
\begin{align*}
\tr{\gamma_S\hat \rho_z^O}&=\frac{1}{2^n}\left(\tr{\gamma_S}+\sum_{j=1}^n\lambda_{2j}^{-1}\sum_{R\in\mathcal{D}_{2j}}\bra{z}\gamma_R\ket{z}\sum_{S'\in\mathcal{S}_{2j}}\det(O_{R,S'})\tr{\gamma_S\gamma_{S'}}\right)\,\\
&=\sum_{j=1}^n\lambda_{2j}^{-1}\sum_{R\in\mathcal{D}_{2j}}\bra{z}\gamma_R\ket{z}\sum_{S'\in\mathcal{S}_{2j}}\det(O_{R,S'})\delta_{S,S'}\\
&=\lambda_{2k}^{-1}\sum_{R\in\mathcal{D}_{2k}}\bra{z}\gamma_R\ket{z}\det(O_{R,S})\,,
\end{align*}
where we have used $\tr{\gamma_S}=0$ and $\tr{\gamma_S\gamma_{S'}}=2^n\delta_{S,S'}$. Consider, for example, $k=1$, such that $1\pm\eta_2\tr{\gamma_S\hat \rho_z^O}\geq 0$ becomes
\begin{equation}
1\pm\eta_2\lambda_2^{-1}\sum_{j=1}^n\bra{z}\gamma_{2j-1}\gamma_{2j}\ket{z}\det(O_{\{2j-1,2j\},S})\geq 0\,,
\end{equation}
for all $S\in\mathcal{S}_{2}$. Since $O\in \text{Alt}(2n)$, for a fixed $j\in[n]$, $\exists!S\in\mathcal{S}_{2}$ such that $\det(O_{\{2j-1,2j\},S})\neq 0$, i.e., $\det(O_{\{2j-1,2j\},S'})=0$ for all $S'\neq S$, and has unit determinant otherwise. Furthermore, $\forall j'\neq j$, $\det(O_{\{2j'-1,2j'\},S})=0$. Hence, $1+\eta_2\lambda_2^{-1}\geq 0$ and therefore $\eta_2 \leq 1/(2n-1)$. This argument generalises to any $k\in[n]$, such that the set of noisy observables $\M_S(\pm)$ are jointly measurable if $1-\eta_{2k}\lambda_{2k}^{-1}\geq 0$, i.e., $\eta_{2k}\leq \binom{n}{k}\binom{2n}{2k}^{-1}$.

\section{Two-observable marginals}\label{A:2-obs-marginals}

Consider a $2q$--local Hamiltonian, consisting of even Majorana monomials of degree at most $2q$, written as
\begin{equation}\label{eq:ham}
H=\sum_{k=1}^{q}\sum_{S\in \mathcal{S}_{2k}}\alpha_{S}\gamma_{S}\,,
\end{equation}
with $\alpha_S\in\mathbb{R}$. Our joint measurement strategy provides an unbiased estimator of $\tr{H\rho}$ given by
\begin{equation}\label{A:estimator}
\hat H=\sum_{k=1}^{q}\sum_{S\in \mathcal{S}_{2k}}\eta_S^{-1}\alpha_{S}e_{S}\,,
\end{equation}
where each $e_{S}\in\{\pm 1\}$ is the outcome associated with the unsharp measurement $\M_{S}(e_S)=\frac{1}{2}(\id+e_S\eta_S\gamma_S)$ obtained by classical post-processing of the parent POVM $\G^O$ defined in Eq. (\ref{eq:jm_gen}), or equivalently Eq. (\ref{eqA:equiv_jm}). Note that, for each observable $\gamma_S$, we assume a unique post-processing determined by the index set labeled $R(S)$ such that $\eta_S:=\eta_{R(S),S}=\max\{\eta_{R,S}\,|\,R\in\mathcal{D}_{2k}\}$, as described in Eq. (\ref{eq:optimal_eta_S}). For simplicity, we also assume all Majorana observables can be jointly measured by a single implementation of the parent POVM $\G^O$, i.e., without a randomization over multiple fermionic Gaussian unitaries. In expectation, the estimator satisfies
\begin{equation*}
\mathbb{E}[\hat H]=\sum_{S}\eta_S^{-1}\alpha_{S}\mathbb{E}[e_{S}]=\sum_{S}\alpha_{S}\tr{\gamma_{S}\rho}=\tr{H\rho}\,,
\end{equation*}
where the expectation value $\mathbb{E}$ is over the outcome statistics of the POVM $\G^O$ on the state $\rho$. To calculate the variance,
\begin{eqnarray}\label{A:var}
\var[\hat H]&=&\sum_{S,S'}\frac{\alpha_{S}\alpha_{S'}}{\eta_S\eta_{S'}}\mathbb{E}[e_{S}e_{S'}]-(\tr{H\rho})^2\,,
\end{eqnarray}
we evaluate $\mathbb{E}[e_{S}e_{S'}]$ for all pairs of observables $\gamma_S$ and $\gamma_{S'}$ with $|S|=2k$ and $|S'|=2k'$. Thus, we require the two-observable marginal of the parent POVM $\G^O$ of Eq. (\ref{eqA:equiv_jm}), given by
\begin{equation}\label{Aeq_2_obs_marg}
\M_{S,S'}(e_{S},e_{S'})=\sum_{{\bf q},\,{\bf x}}D(e_S|S,R,{\bf q},{\bf x})D(e_{S'}|S',R',{\bf q},{\bf x})\G^O({\bf q},{\bf x})\,.
\end{equation}
Here, and for the remainder of the appendix, we will simplify the notation $R(S)$ and $R(S')$ to $R$ and $R'$, respectively, when no ambiguity arises.

\subsection{Single Majoranas}

Before moving to the general setting, we consider observables written as linear combinations of single Majorana operators, such that $H=\sum_{\mu=1}^{2n}\alpha_{j}\gamma_{j}$, where $\alpha_{j}\in\mathbb{R}$. The estimator takes the form
\begin{equation}
\hat H=\sum_{j=1}^{2n}\eta^{-1}\alpha_{j}e_{j}=\sqrt{2n}\sum_{j=1}^{2n}\alpha_{j}e_{j}\,,
\end{equation}
where $\eta=\frac{1}{\sqrt{2n}}$ is the optimal (uniform) noise achieved by the parent POVM of Eq. (\ref{eq:jm_single}). The noisy joint measurement of $\gamma_{i}$ and $\gamma_{j}$, for $i,j\in[2n]$, is given by 
\begin{align*}
\M_{i,j}(e_{i},e_{j})&=\sum_{{\bf x}} D(e_{i}|{\bf x})D(e_{j}|{\bf x})\G({\bf x})=\sum_{\substack{{\bf x}\\ x_{i}=e_{i}\\x_{j}=e_{j}}}\G({\bf x})\\
&=\frac{1}{4}\left(\id+\frac{1}{\sqrt{2n}}(e_{i}\gamma_{i}+e_{j}\gamma_{j})\right)\,.
\end{align*}
It follows that for $i\neq j$, $\mathbb{E}[e_{i}e_{j}]=\sum_{e_{i},e_{j}}\tr{\M_{i,j}(e_{i},e_{j})\rho}e_{i}e_{j}=0$. Since $\mathbb{E}[e_{i}^2]=1$, it follows from Eq. (\ref{A:var}) that,
\begin{equation}
\var[\hat H]=2n\sum_{j=1}^{2n}\alpha_{j}^2-(\tr{H\rho})^2\leq 2n\sum_{j=1}^{2n}\alpha_{j}^2\,.
\end{equation}

\subsection{Arbitrary even degree Majoranas}

We now construct the two-observable marginal $\M_{S,S'}(e_{S},e_{S'})$, with outcomes $e_S,e_{S'}\in\{\pm 1\}$, for $S\in\mathcal{S}_{2k}$ and $S'\in\mathcal{S}_{2k'}$. Let $\tau=\text{sgn}(\det(O_{R,S}))$ and $\tau'=\text{sgn}(\det(O_{R',S'}))$ and denote $S\triangle S'=(S\cup S')\setminus(S\cap S')$ as the symmetric difference. For $R=\cup_{i=1}^kR_i\in\mathcal{D}_{2k}$, with $R_i\in\mathcal{D}_2$, let $q_{R}=\prod_{i=1}^k q_{R_i}$ such that $q_Rq_{R'}=q_{R\triangle R'}$. Furthermore, let $x_S=\prod_{j\in S}x_j$ with $x_j\in\{\pm 1\}$ such that $x_Sx_{S'}=x_{S\triangle S'}$. From Eq. (\ref{Aeq_2_obs_marg}) we have,
\begin{eqnarray}\label{eq:obs_marg_gen}
\M_{S,S'}(e_{S},e_{S'})&=&\sum_{{\bf q},\,{\bf x}}D(e_S|S,R,{\bf q},{\bf x})D(e_{S'}|S',R',{\bf q},{\bf x})\G^O({\bf q},{\bf x})\\
&=&\frac{1}{2^{3n}}\sum_{\substack{{\bf q},\,{\bf x}\\e_S=\tau x_Sq_R\\e_{S'}=\tau' x_{S'}q_{R'}}}\left(\id+\sum_{k=1}^n \sum_{Y\in\mathcal{D}_{2k}}q_Y\sum_{T\in\mathcal{S}_{2k}}x_T\det(O_{Y,T})\gamma_{T}\right)\,.\\
\end{eqnarray}

Notice that terms which do not contain either $x_Sq_R$, $x_{S'}q_{R'}$ or 
$q_{R\triangle R'}x_{S\triangle S'}$ are unconstrained (taking values $\pm 1$) and therefore cancel in the summation over ${\bf x}$ and ${\bf q}$. Hence, we are left with
\begin{multline*}
\M_{S,S'}(e_S,e_S)=\frac{1}{2^{3n}}\sum_{\substack{{\bf q},\,{\bf x}\\\tau x_Sq_R=e_S\\\tau' x_{S'}q_{R'}=e_{S'}}}(\id+q_Rx_S\det(O_{R,S})\gamma_S+q_{R'}x_{S'}\det(O_{R',S'})\gamma_{S'}\\
+q_{R\triangle R'}x_{S\triangle S'}\det(O_{R\triangle R',S\triangle S'})\gamma_{S\triangle S'}\delta_{|R\triangle R'|,|S\triangle S'|})\,.
\end{multline*}
Summing over ${\bf x}=(x_1,\ldots,x_{2n})\in\{\pm 1\}^{2n}$ and ${\bf q}\in\{\pm 1\}^n$, with the constraints $\tau x_Sq_R=e_S$ and $\tau' x_{S'}q_{R'}=e_{S'}$, yields
\begin{align}
\M_{S,S'}(e_S,e_{S'}) 
&= \frac{1}{4}\Bigl(
    \id
    + q_R x_S \det(O_{R,S})\,\gamma_S
    + q_{R'} x_{S'} \det(O_{R',S'})\,\gamma_{S'}
    + q_{R\triangle R'} x_{S\triangle S'} \det(O_{R\triangle R',S\triangle S'})\,\gamma_{S\triangle S'} \notag \\
&\qquad\qquad
    \times \delta_{|R\triangle R'|,\,|S\triangle S'|}
\Bigr) \notag \\
&= \frac{1}{4}\Bigl(
    \id
    + \frac{e_S}{\tau}\,\det(O_{R,S})\,\gamma_S
    + \frac{e_{S'}}{\tau'}\,\det(O_{R',S'})\,\gamma_{S'}
    + \frac{e_S e_{S'}}{\tau\tau'}\,\det(O_{R\triangle R',S\triangle S'})\,\gamma_{S\triangle S'} \notag \\
&\qquad\qquad
    \times \delta_{|R\triangle R'|,\,|S\triangle S'|}
\Bigr) \notag \\
&= \frac{1}{4}\Bigl(
    \id
    + e_S\,|\det(O_{R,S})|\,\gamma_S
    + e_{S'}\,|\det(O_{R',S'})|\,\gamma_{S'} 
    + \frac{e_S e_{S'}}{\tau\tau'}\,\det(O_{R\triangle R',S\triangle S'})\,\gamma_{S\triangle S'} \notag \\
&\qquad\qquad
    \times \delta_{|R\triangle R'|,\,|S\triangle S'|}
\Bigr).
\end{align}

It follows that,
\begin{equation}
\mathbb{E}[e_{S}e_{S'}]=\sum_{e_{S},e_{S'}}\tr{\M_{S,S'}(e_{S},e_{S'})\rho}e_{S}e_{S'}=\frac{1}{\tau\tau'}\det(O_{R\triangle R',S\triangle S'})\tr{\gamma_{S\triangle S'}\rho}\delta_{|R\triangle R'|,|S\triangle S'|}\,.
\end{equation}
Finally, the variance of the estimator is given by
\begin{equation}
\var[\hat H]=\sum_{S}\alpha^2_{S}\eta_S^{-2}+\sum_{S\neq S'}\delta_{|R\triangle R'|,|S\triangle S'|}\alpha_{S}\alpha_{S'}\frac{\det(O_{R\triangle R',S\triangle S'})}{\det(O_{R,S})\det(O_{R',S'})}\tr{\gamma_{S\triangle S'}\rho}-(\tr{H\rho})^2\,,
\end{equation}
in agreement with Prop. \ref{cor:var} in Sec. \ref{sec:hamiltonians}.

\end{document}